\documentclass[11pt,a4paper,oneside]{article}
\usepackage[top=3cm, bottom=3cm, left=2cm, right=2cm]{geometry}
\usepackage[english]{babel}
\usepackage[utf8]{inputenc}
\usepackage{enumerate}
\usepackage{enumitem}
\usepackage[T1]{fontenc}
\usepackage{setspace}
\usepackage{xcolor}
\usepackage{booktabs}
\usepackage{multirow}
\usepackage{xr}
\usepackage{tikz}
\usetikzlibrary{positioning}
\usetikzlibrary{fit}

\usepackage{amsthm,amsmath,amssymb,mathrsfs,dsfont,mathtools, mathabx}
\usepackage{bbm}
\usepackage{bm}
\usepackage{graphicx}
\usepackage[colorlinks=true, allcolors=blue]{hyperref}
\usepackage{comment}
\usepackage[all]{nowidow}

\usepackage[authoryear, round]{natbib}
\setcitestyle{authoryear}

\hypersetup{
    colorlinks=true,
    citecolor=blue,
    linkcolor=blue,
    filecolor=blue,      
    urlcolor=blue
    }

\def\bs{\boldsymbol}
\def\lb{\left\{ }
\def\rb{\right\}}
\def\R{\mathbb{R}}
\def\E{\mathbb{E}}

\def\tx{\textrm}
\def\X{\bs X}

\def\P{\mathbb{P}}
\def\P{\mathbb{P}}
\def\G{\mathcal{G}}

\def\VI{\tx{VI}}

\def\N{\mathcal{N}}

\def\cFOLD{\bs c_{\tx{FOLD}}}
\def\Lc{\mathcal{L}(\widehat{\bs c}, \bs \theta)}

\def\D{\mathcal D}
\def\chat{\widehat{\bs c}}
\def\genq{\underset{\sim}{>}}
\def\lenq{\underset{\sim}{<}}
\def\Ilam{\mathcal{I}^{(\lambda)}}
 
\newcommand{\norm}[1]{\left\lVert#1\right\rVert}

\theoremstyle{definition}
\newtheorem{definition}{Definition}[section]
\newtheorem{theorem}{Theorem}
\newtheorem{assumption}{Assumption}
\newtheorem{remark}{Remark}
\newtheorem{proposition}{Proposition}

\newtheorem{lemma}{Lemma}

\usepackage{authblk}
\date{ \vspace{-5mm} }
\title{Bayesian Clustering via Fusing of Localized Densities}
\author[1]{Alexander Dombowsky}
\author[1,2]{David B. Dunson}
\affil[1]{Department of Statistical Science, Duke University, Durham, NC, USA}
\affil[2]{Department of Mathematics, Duke University, Durham, NC, USA}

\usepackage[sf]{titlesec}
\usepackage{sectsty}
\allsectionsfont{\centering \normalfont\scshape} % Make all sections centered, the default font and small caps
\begin{document}
\maketitle
\begin{abstract}
Bayesian clustering typically relies on mixture models, with each component interpreted as a different cluster. After defining a prior for the component parameters and weights, Markov chain Monte Carlo (MCMC) algorithms are commonly used to produce samples from the posterior distribution of the component labels. The data are then clustered by minimizing the expectation of a clustering loss function that favors similarity to the component labels. Unfortunately, although these approaches are routinely implemented, clustering results are highly sensitive to kernel misspecification. For example, if Gaussian kernels are used but the true density of data within a cluster is even slightly non-Gaussian, then clusters will be broken into multiple Gaussian components. To address this problem, we develop Fusing of Localized Densities (FOLD), a novel clustering method that melds components together using the posterior of the kernels. FOLD has a fully Bayesian decision theoretic justification, naturally leads to uncertainty quantification, can be easily implemented as an add-on to MCMC algorithms for mixtures, and favors a small number of distinct clusters. We provide theoretical support for FOLD including clustering optimality under kernel misspecification. In simulated experiments and real data, FOLD outperforms competitors by minimizing the number of clusters while inferring meaningful group structure.
\end{abstract}

\noindent%

\newpage

\section{Introduction}

Clustering data into groups of relatively similar observations is a canonical task in exploratory data analysis.
 Algorithmic clustering methods, such as k-means, k-medoids, and hierarchical clustering, rely on dissimilarity metrics, an approach which is often heuristic but may perform well in practice; see \cite{hastie2009elements}, \cite{jain2010data}, and \cite{kiselev2019challenges} for an overview. In comparison, model-based clustering methods utilize mixtures of probability kernels to cluster data, ordinarily by inferring the component labels \citep{fraley2002model}. Choices of kernel depend on the type of data being considered, with the Gaussian mixture model (GMM) particularly popular for continuous data. A conceptual advantage of model-based methods is the ability to express uncertainty in clustering. For example, the expectation-maximization (EM) algorithm uses maximum likelihood estimates of the component-specific parameters and weights to calculate each observations' posterior component allocation probabilities, which are interpreted as a measure of clustering uncertainty
 \citep{bensmail1997inference}. As estimation of the weights and component-specific parameters is crucial for model-based clustering, there is a rich literature on quantifying rates of convergence for various mixture models, which is usually expressed in terms of the mixing measure (see \citealp{guha2021posterior} and references therein).

Bayesian mixture models have received increased attention as a clustering method in recent years. The Bayesian framework can account for uncertainty in the number of components and incorporate prior information on the component-specific parameters, with Markov chain Monte Carlo (MCMC) algorithms employed to generate posterior samples for the mixture weights, kernel parameters, and component labels for each data point. Based on posterior samples of the component labels, one can obtain Monte Carlo approximations to Bayes clustering estimators. For $p$-dimensional data $\bs X = (X_1, \dots, X_n)$, a Bayes estimator $\bs c^*$ is the minimizer of an expected loss function conditional on $\bs X$ over all possible clusterings $\bs c = (c_1, \dots, c_n)$.
There are several popular choices of loss, including Binder's \citep{binder1978bayesian} and the Variation of Information (VI) \citep{meilua2007comparing}, which are distance metrics over the space of partitions and favor clusterings that are similar to the component labels. Further details on set partitions and their role in Bayesian clustering can be found in \cite{meilua2007comparing}, \cite{wade2018bayesian}, and \cite{paganin2021centered}. 
Several measures of uncertainty in clustering with Bayesian mixtures exist, including the posterior similarity matrix and credible balls of partitions \citep{wade2018bayesian}. 

\begin{figure}[t]
    \centering
    \includegraphics[scale=0.35]{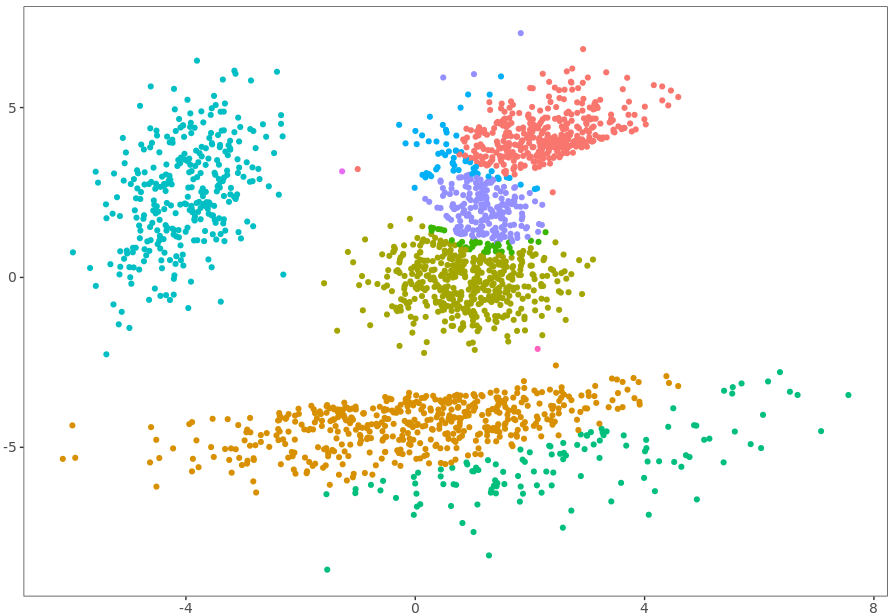}
    \caption{A Bayesian GMM with $10$ components and Dirichlet prior concentration parameter equal to $1/10$ is fit to data generated from a mixture of skew Gaussian distributions.}
    \label{fig:intro}
\end{figure}

However, model-based clustering approaches, including Bayesian implementations, can be brittle in applications due to unavoidable impacts of model misspecification. While the data are assumed to come from a model $X_i \sim f= \sum_{k=1}^K a_k g(\cdot; \tilde \theta_k)$, in reality, data are generated via a true process $f^0$, which may or may not be in the support of the prior for $f$. A specific example is \textit{kernel misspecification}, which occurs when $f^0$ is in fact a mixture model whose components are not contained in the kernel family mixed over by the fitted model. Mixtures will often induce over-clustering in this setting by approximating a single component in $f^0$ with many components from $f$. An example of this phenomena is shown in Figure \ref{fig:intro}, where a $10$ component Bayesian GMM is fit to data generated from a mixture of bivariate skew Gaussian distributions. Despite using a concentration parameter of $1/10$ in the symmetric Dirichlet prior to induce a small number of clusters \citep{rousseau2011asymptotic}, the GMM allocates the data into $10$ poorly defined groups. A group of observations in the bottom third of the sample space is split into two, while several observations are placed into their own groups despite being near dense collections of data. 

Several approaches have been proposed to address the issue of over-clustering due to kernel misspecification. A natural solution is to define a flexible class of kernels, exemplified by the mixtures in \cite{karlis2009model}, \cite{o2016clustering}, and \cite{dang2023model}. To increase flexibility further, \cite{rodriguez2014univariate} propose a mixture of nonparametric unimodal kernels. Similarly, \cite{bartolucci2005clustering}, \cite{li2005clustering}, \cite{di2007mixture}, and \cite{malsiner2017identifying} use carefully chosen mixtures of Gaussians to characterize the data within each cluster. However, there is an unfortunate pitfall with the general strategy of using flexible families of kernels. In particular, as the flexibility of the kernel increases, identifiability and optimal estimation rates for inferring the mixing measure tend to weaken, especially when the true number of mixture components is unknown \citep{nguyen2013convergence,ho2016convergence, heinrich2018strong}. Even the transition from location Gaussian kernels with known covariance to location-scale Gaussian kernels can have substantial consequences on the convergence rate of the component means \citep{manole2022refined}. Such problems motivated \cite{ho2020robust} to propose an alternative estimator of the mixing measure that is more robust than maximum likelihood and Bayesian approaches and can achieve optimal convergence rates. 

Alternatively, one can develop generalized Bayesian methods of clustering that avoid defining a fully generative probabilistic model for the data. For example, \cite{duan2021bayesian} propose to conduct model-based clustering based on a pairwise distance matrix instead of the data directly to reduce sensitivity to kernel misspecification. \cite{rigon2023generalized} instead define a Gibbs posterior for clustering, incorporating a clustering loss function in place of a likelihood function, completely bypassing modeling of the data. An alternative that maintains a fully generative model while robustifying inferences to misspecification is to use a coarsened posterior \citep{miller2018robust, gorsky2023coarsened}. This coarsening approach often has good practical performance in reducing over-clustering due to kernel misspecification.

Unfortunately, these approaches for accommodating kernel misspecification have various drawbacks. Along with slower convergence rates, flexible kernels typically require a large number of parameters, worsening the already burdensome computational cost of Bayesian clustering. The generalized Bayes approaches can perform well in certain settings. However, both Gibbs posteriors and coarsened posteriors are highly sensitive to key tuning parameters, which can be difficult to choose objectively in practice.

A key theoretical advance is given in \cite{aragam2020identifiability},
which attempts to solve the problem of clustering based on a mixture model by merging components in a Gaussian mixture. They rely on a two-stage procedure that lacks partition uncertainty quantification and assumes that the true number of kernels is known.
 However, the approach of viewing clusters as arising from merging closely overlapping kernels is promising. Related merging ideas have been implemented in both frequentist and Bayesian approaches in a variety of settings, and several algorithms exist for deciding how and when to combine components together \citep{chan2008statistical, baudry2010combining,
 hennig2010methods, melnykov2016merging,guha2021posterior, manole2021estimating}.

In this paper, we propose a novel decision theoretic method for Bayesian clustering that mitigates the effects of model misspecification. Suppose we model the data with a Bayesian mixture model, with $K$ components, component labels $\bs s = (s_1, \dots, s_n)$, component-specific atoms $\tilde \theta_1, \dots, \tilde \theta_K$, and kernels $g(\cdot; \tilde \theta_1), \dots, g(\cdot; \tilde \theta_K)$. Rather than focusing on $\bs s$, we compute clusters with the localized densities $\lb g(\cdot; \theta_i) \rb_{i=1}^n$, defined by $\theta_i = \tilde{\theta}_{s_i}$. We define a loss function for any clustering $\widehat{\bs c}$ that favors allocating $i$ and $j$ to the same cluster when the statistical distance between $g(\cdot;\theta_{i})$ and $g(\cdot;\theta_{j})$ is small, encouraging grouping observations with overlapping component kernels. We cluster the data with a Bayes estimator, interpreted as a Fusing of Localized Densities (FOLD). Our method has a fully decision theoretic justification, leads to interpretable uncertainty quantification, and can be readily implemented using the output of existing MCMC algorithms for mixtures. 

Though previous methods have utilized merging kernels to account for kernel misspecification, to our knowledge none have a formal Bayesian decision theoretic justification. Although FOLD requires combinatorial optimization when obtaining the point estimate, we suggest two reasonable approximations: one in which we reduce the optimization space to a smaller set of candidate clusterings, and another in which we use existing greedy algorithms to compute a locally-optimal solution. We also provide concrete theoretical guarantees for our procedure, including the result that as the sample size grows, the clustering point estimate converges to an oracle clustering rule. In addition, we list intuitive examples of oracle rules and empirically validate our asymptotic theory using a simple example. 

In Section \ref{section:methods}, we explain our clustering method from a Bayesian decision theoretic perspective, provide a framework for uncertainty quantification, and demonstrate how to implement FOLD in practice. In Section \ref{section:theory}, we show asymptotic concentration of FOLD for misspecified and well-specified regimes. We apply FOLD to a cell line data set in Section \ref{section:apps}, showing excellent performance. Finally, we provide concluding remarks and some extensions in Section \ref{section:discussion}. The Supplementary Material contains additional simulations and real-data examples, as well as proofs for all theoretical results.

\section{Clustering with Localized Densities} \label{section:methods}

\subsection{Notation and Setup}
Let $X_i^0 = (X_{i1}^0, \dots, X_{ip}^0) \in \R^p$ be multivariate observations, collected into $\bs X^0 = (X_1^0, \dots, X_n^0)$. Assume that $X_1^0, \dots, X_n^0$ are random and generated from an unknown mixture model: $X_i^0 \sim f^0$, where $f^0 = \sum_{m=1}^{M^0} a_m^0 g_m^0(\cdot)$, $a_m^0 > 0$ for $m = 1, \dots, M^0$, $\sum_{m=1}^{M^0} a_m^0 = 1$, and $M^0 < \infty$. Since $f^0$ is a mixture, the data-generating process can be equivalently stated with the addition of latent variables $s_{i}^0 \in \lb 1, \dots, M^0 \rb$, so that $\left(X_i^0 \mid s_i^0=m\right) \sim g_{m}^0(\cdot)$ for all $i = 1, \dots, n$. $\mathbb{P}^0$ refers to probability and convergence statements with respect to the true data generating process, with $\mathbb{P}^0(A) = \int_A f^0(x) dx$ for any set $A$.

Throughout this article, we will focus on the setting in which $\X^0$ is modeled with a mixture. Let $X_i = (X_{i1}, \dots, X_{ip}) \in \R^p$, where $\X = (X_1, \dots, X_n)$. A Bayesian mixture model assumes that $\X$ are generated via 
\begin{align} \label{eq:bayesian-mixture-model}
    \lambda & \sim \pi_\Lambda, & \theta_i \mid \lambda & \sim \lambda, & X_i \mid \theta_i \sim g(X_i; \theta_i);
\end{align}
 for all $i=1, \dots, n$. Here, $\lambda$ is a probability distribution over a parameter space $\Theta$ and known as the \textit{mixing measure}, whereas $\G = \lb g(\cdot ; \theta) : \theta \in \Theta \rb$ denotes a family of probability densities with support on $\R^p$. We will impose more conditions on both $\Theta$ and $\mathcal G$ when we develop our theoretical results in Section \ref{section:theory}. $\pi_\Lambda$ is the prior distribution for the mixing measure with support $\Lambda$, and examples of such priors are implemented in Bayesian finite mixture models, mixtures of finite mixtures (MFMs) \citep{miller2018mixture}, Dirichlet process mixtures (DPMs) \citep{ferguson1973bayesian, antoniak1974mixtures}, and Pitman-Yor process mixtures \citep{ishwaran2001gibbs, ishwaran2003some}. Generally, realizations from $\pi_\Lambda$ are almost-surely discrete (e.g., as in \cite{sethuraman1994constructive}) and can be represented as $\lambda = \sum_{k=1}^K a_k \delta_{\tilde \theta_k}$, where $1 \leq K \leq \infty$, $\tilde \theta_k \in \Theta$, $0 < a_k < 1$, and $\sum_{k=1}^K a_k = 1$. Hence, sampling $\theta_i$ from \eqref{eq:bayesian-mixture-model} results in ties amongst the \textit{atoms} $\bs \theta = (\theta_1, \dots, \theta_n)$, which are represented by unique values $\Tilde{\bs \theta} = (\tilde \theta_1, \dots, \tilde \theta_{K_n})$ for some $1 \leq K_n \leq n$. From these ties we can construct a latent partition $S = (S_1, \dots, S_{K_n})$ of the integers $[n] = \lb 1, \dots, n \rb$, where $S_k = \lb i \in [n] : \theta_i = \tilde \theta_k \rb$. The latent partition may also be represented with labels $\bs s=(s_1, \dots, s_n)$, where $s_i = k \iff i \in S_k$. We will assume that the prior for $\Tilde{\bs \theta}$ is non-atomic, so the allocation of two observations to the same component is equivalent to equality in their atoms, i.e. $s_i = s_j \iff \theta_i = \theta_j$.

Ultimately, our goal is to cluster the data $\X^0$ using the mixture model outlined in \eqref{eq:bayesian-mixture-model}. The typical approach is to infer the conditional distribution of $S$ after observing the data. More formally, by assuming that $\lambda \sim \pi_\Lambda$, we have induced a prior distribution for $S$, $\Pi(S)$. We then cluster the data by inferring $\Pi(S \mid \X = \X^0)$, which we will refer to for the remainder of this article as $\Pi(S \mid \X^0)$. To construct a point estimate of $S$, we minimize the posterior expectation of a clustering loss function that favors similarity to $S$. The decision theoretic approach to Bayesian clustering was first proposed in \cite{binder1978bayesian}, and was advocated for more recently in \cite{wade2018bayesian} and \cite{dahl2022search}. However, Bayesian methods can lead to homogenous, tiny, and often uninterpretable clusters in practice. These issues often arise in the presence of kernel misspecification. Even well-separated true clusters in $\X^0$ will be broken into multiple components from $\G$ under a minor degree of kernel misspecification. To address this issue, we propose a loss function focused on simplifying the overfitted kernels via merging rather than estimating the component labels. 

\subsection{Decision Theory Formulation}

We aim to estimate a clustering $\widehat C = \lb \hat{C}_1, \dots, \hat{C}_{\hat{K}_n} \rb$ of $\X^0$, or labels $\widehat{\bs c} = (\hat{c}_1, \dots, \hat{c}_n)$, based on merging components from the joint posterior that have similar kernels $g(\cdot; \tilde \theta_k)$. The motivation here is to remedy the cluster splitting that occurs in using multiple parametric kernels $g(\cdot; \tilde \theta_k)$ to represent each ``true" kernel $g_m^0(\cdot)$. We define the \textit{localized density} of observation $i$ to be $g(\cdot; \theta_i)$, a random probability density that takes values $g(x;\theta_i)$ for all $x \in \mathbb{R}^p$, where $\theta_i$ is the atom of observation $i$ as defined in \eqref{eq:bayesian-mixture-model}. Observe that the localized density is the distribution for $X_i$ under the fitted mixture model given $\theta_i$, i.e. $(X_i \mid \theta_i = \tilde \theta_k) \sim g(\cdot; \tilde \theta_k)$. Prior variation in the localized densities is modeled by sampling $\lambda \sim \pi_\lambda$, and then 
sampling $\theta_i \sim \lambda$, for $i=1,\ldots,n$.

\begin{figure}[ht]
\centering
\includegraphics[scale=0.22]{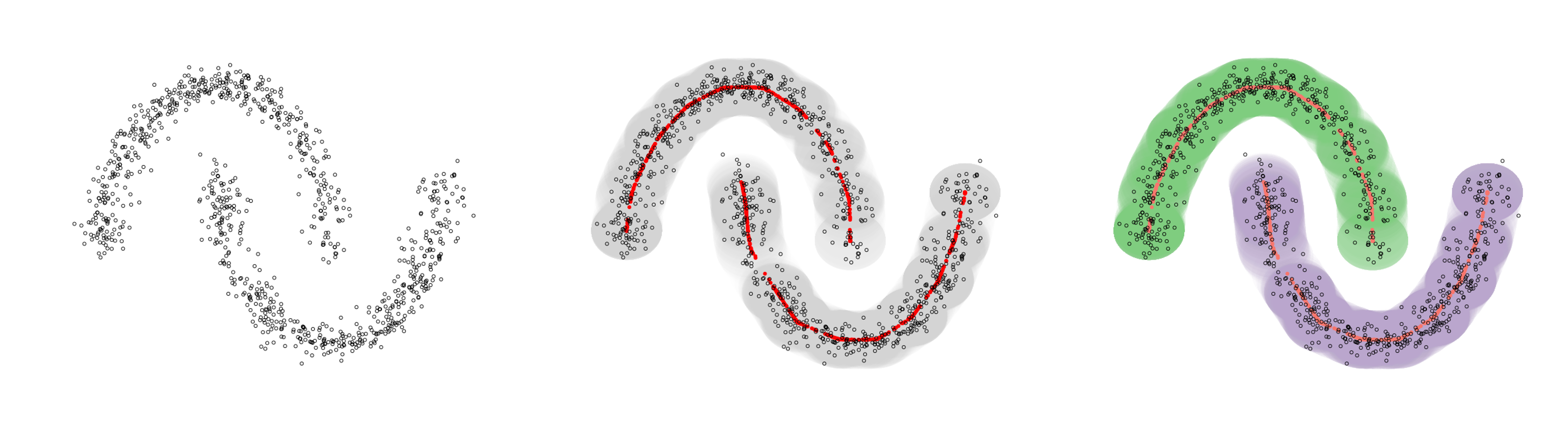}
\caption{A Bayesian location Gaussian mixture is fit to a version of the \texttt{moons} data set. The localized densities are inferred, then merged together to recover the crescents.}
\label{fig:moons}
\end{figure}

To counteract cluster splitting, we define the loss of assigning two observations into the same or different clusters as a function of the statistical distance between their localized densities. Figure \ref{fig:moons} shows the behavior of the posterior distributions of the localized densities when the model $f(x)=\sum_{k=1}^{30} a_k \N_2(x;\tilde \theta_k, 0.02 I)$ is fit to a version of the \texttt{moons} data. The red points plot $\E_\Pi(\theta_{i} \mid \X^0)$ for each $i=1, \dots, n$, with gray circles representing the $95\%$ high density regions of a $\N_2(\E_\Pi(\theta_{i} \mid \X^0), 0.02 I )$ distribution. The crescent clusters are split into multiple overlapping Gaussian kernels, but are recovered by fusing these kernels together.

We define the loss of any clustering $\widehat{\bs c}$ resulting from the model in \eqref{eq:bayesian-mixture-model} to be
\begin{gather} \label{eq:lg}
   \Lc = \sum_{i<j} \lb \textbf{1}_{\hat{c}_i = \hat{c}_j } \D_{ij} + \omega \textbf{1}_{\hat{c}_i \neq \hat{c}_j} \left(  1 - \D_{ij} \right) \rb, %\\
\end{gather}
where $\omega > 0$, $\D_{ij} = d \lb g(\cdot; \theta_{i}), g(\cdot; \theta_{j}) \rb$, and $d$ is any distance between probability distributions assumed to be bounded in the unit interval, i.e. $0 \leq d \lb P(\cdot), Q(\cdot) \rb \leq 1$ for all measures $P$ and $Q$, and with the property that $d \lb P(\cdot), Q(\cdot) \rb = 1$ if and only if $\tx{supp}(P) \cap \tx{supp}(Q) = \emptyset$. A key example is when $d$ is chosen to be the Hellinger distance, but this need not be the case in general. Observe that $\Lc$ is non-negative for any $\widehat{ \bs c}$. The loss of assigning $\hat{c}_i = \hat{c}_j$ is $\D_{ij}$. If $s_i = s_j$, then $\theta_i = \theta_j$ and hence there is no loss incurred. When $s_i \neq s_j$, then $\theta_i \neq \theta_j$ but the loss will remain small when the kernels $g(\cdot; \theta_{i}), g( \cdot; \theta_{j})$ are similar under $d$. Conversely, allocating $i$ and $j$ to different clusters results in a loss of $\omega(1 - \D_{ij})$. If $\D_{ij}=1$, i.e. if the supports of $g(\cdot; \theta_{i})$ and $g(\cdot; \theta_j)$ are disjoint, then setting $\hat{c}_i \neq \hat{c_j}$ accumulates zero loss. Otherwise, the loss depends on the separation in $d$ between the localized densities and the loss parameter $\omega$.

Our loss $\Lc$ is invariant to permutations of the data indices and of the labels in either $\widehat{\bs c}$ or $\bs s$, which is desirable for clustering losses \citep{binder1978bayesian}. In addition, $\Lc$ is a continuous relaxation of Binder's loss function \citep{binder1978bayesian,lau2007bayesian},
\begin{equation} \label{eq:lb}
    \mathcal{L}_{\tx{B}}(\widehat{\bs c}, \bs s) = \sum_{i<j} \lb \textbf{1}_{\hat{c}_i = \hat{c}_j} \textbf{1}_{s_i \neq s_j} + \omega \textbf{1}_{\hat{c}_i \neq \hat{c}_j} \textbf{1}_{s_i = s_j} \rb.
\end{equation}
A key property of Binder's loss is that it is a quasimetric over the space of partitions \citep{wade2018bayesian, dahl2022search}. $\Lc$ is not a quasimetric, but instead aims to mitigate the problem of cluster splitting 
arising under kernel misspecification. Note the difference in notation in the arguments of \eqref{eq:lb} and \eqref{eq:lg}: we omit $\bs s$ from the notion for $\Lc$ to reflect that we will incorporate information from the atoms in clustering $\X^0$, not just the latent partition. However, we can also state \eqref{eq:lb} in terms of $\bs \theta$ by recalling that $s_i = s_j \iff \theta_i = \theta_j$.
One can show that $\Lc$ can be rewritten as the sum of $\mathcal{L}_{\tx{B}}(\widehat{\bs c}, \bs s)$ and a remainder term $\mathcal{B}(\widehat{\bs c}, \Tilde{\bs \theta})$ that depends on the unique values $\Tilde{\bs \theta}$, and that $\mathcal{L}(\bs s, \bs \theta) = 0$ only when the components of $f$ are completely separated under $d$. The closed form of the remainder $\mathcal{B}(\widehat{\bs c}, \Tilde{\bs \theta})$ is given in the Supplement, and is interpreted as an added cost to Binder's loss that encourages merging components.
\begin{proposition} \label{prop:Lg}
    $\Lc = \mathcal{L}_{\tx{B}}(\widehat{\bs c}, \bs s) 
   + \mathcal{B}(\widehat{\bs c}, \Tilde{\bs \theta}),$
    where $\mathcal{B}(\bs s, \Tilde{\bs \theta}) = 0$ if and only if \\ $d \{ g(\cdot; \tilde \theta_k), g(\cdot; \tilde \theta_{k^\prime}) \} = 1$ for all component pairs $k \neq k^\prime = 1, \dots, K$. 
\end{proposition}

The parameter $\omega$ calibrates separation of the clusters. For example, suppose we compare the clustering $\widehat{\bs c}_1$, which includes clusters $\hat{C}_{h}$ and $\hat{C}_{h^\prime}$, with the clustering $\widehat{\bs c}_2$, which is equivalent to $\widehat{\bs c}_1$ but now contains the merged cluster $\hat{C}_{h} \cup \hat{C}_{h^\prime}$. The difference in their losses is $\mathcal{L}(\widehat{\bs c}_1, \bs \theta) - \mathcal{L}(\widehat{\bs c}_2, \bs \theta) = \omega \sum_{i \in \hat{C}_{h}, j \in \hat{C}_{h^\prime} } \left( 1 - \D_{ij} \right) - \sum_{i \in \hat{C}_{h}, j \in \hat{C}_{h^\prime} } \D_{ij}.$ The loss of $\widehat{\bs c}_2$ is less than that of $\widehat{\bs c}_1$ when
\begin{equation} \label{eq:hcutoff}
    \frac{1}{|\hat{C}_h| |\hat{C}_{h^\prime}|}\sum_{i \in \hat{C}_{h}, j \in \hat{C}_{h^\prime} } \D_{ij} < \gamma := \frac{\omega}{1 + \omega}.
\end{equation}
This implies that large values of $\omega$ promote fusing clusters,
while smaller values lead to more clusters having less between-cluster heterogeneity. When $\widehat{\bs c}_1 = \bs s$, it is clear that a smaller loss can be attained by merging components with average pairwise statistical distance less than $\gamma$. Trivial clusterings are favored when $\omega$ is taken to its lower and upper limits, similar to both Binder's loss and the VI loss \citep{wade2018bayesian, dahl2022search}. As $\omega \to 0$, our loss is minimized by placing each observation in its own cluster, and, as $\omega \to \infty$, all observations are placed in a single cluster. Binder's loss exhibits a similar interpretation, and one can show that, under the same settings outlined above, $\mathcal L_{\tx{B}}(\chat_2, \bs s) < \mathcal L_{\tx{B}}(\chat_1, \bs s)$ when $|\hat{C}_h|^{-1}|\hat{C}_{h^\prime}|^{-1}\sum_{i \in \hat{C}_{h}, j \in \hat{C}_{h^\prime} }  \textbf{1}_{s_i \neq s_j} = |\hat{C}_h|^{-1}|\hat{C}_{h^\prime}|^{-1} m_{hh^\prime} < \gamma$, where $m_{h h^\prime}$ counts the number of pairs across clusters $h$ and $h^\prime$ that are allocated to different mixture components. However, the loss is always minimized when $\chat = \bs s$.

The risk of any clustering $\widehat{\bs c}$ is the posterior expectation of its loss, which integrates over the uncertainty in $\bs \theta$ after observing the data;
\begin{gather}
    \mathcal{R}(\chat) = \E_\Pi \lb \Lc \mid \X^0 \rb 
    =  \sum_{i<j} \lb \textbf{1}_{\hat{c}_i = \hat{c}_j} \Delta_{ij} + \omega \textbf{1}_{\hat{c}_i \neq \hat{c}_j} \left(  1 - \Delta_{ij} \right) \rb, \label{eq:risk}
\end{gather}
where $\Delta_{ij} = \E_\Pi (\D_{ij} \mid \X^0)$ is the posterior expected distance between the localized densities $g(\cdot; \theta_{i})$ and $g(\cdot; \theta_{j})$ for observations $i$ and $j$, respectively. The point estimator of the clustering of $\X^0$ is given by a Bayes estimator, denoted $\cFOLD$, which is the minimizer of (\ref{eq:risk}) over the space of all possible cluster allocations, i.e. $\cFOLD = \underset{\widehat{\bs c}}{\tx{argmin  }} \mathcal{R}(\chat).$ The expected statistical distance terms in the risk (\ref{eq:risk}) exhibit several bounded and metric properties. 
\begin{proposition} \label{prop:Delta}
    Let $\Delta = (\Delta_{ij})_{i,j}$, where $\Delta_{ij} = \E_\Pi(\D_{ij} \mid \X^0)$. Then $\Delta$ satisfies the following with $\mathbb{P}^0$-probability equal to $1$ for all $i,j,l=1, \dots, n$: (a) $0 \leq \Delta_{ij} \leq 1$; (b) $\Delta_{ii} = 0$, $\Delta_{ji} = \Delta_{ij}$, and $\Delta_{il} \leq \Delta_{ij} + \Delta_{jl}$; and (c) $\Delta_{ij} \leq \Pi(s_i \neq s_j \mid \X^0)$.
\end{proposition}

Properties (a) and (b) follow naturally from the assumptions we have made regarding the statistical distance $d$. To show (c), recall that, conditional on $\lambda = \sum_{k=1}^K a_k \delta_{\tilde \theta_k}$, variation in $\theta_i$ and $\theta_j$ is explained by sampling from $\lambda$ in \eqref{eq:bayesian-mixture-model}. This implies that $\E_\Pi(\D_{ij} \mid \X^0, \lambda) = \sum_{k \neq k^\prime} d \lb g(\cdot; \tilde \theta_k), g(\cdot; \tilde \theta_{k^\prime}) \rb \Pi(s_i=k, s_j = k^\prime \mid \X^0, \lambda)$. Since $0 \leq d(\cdot, \cdot) \leq 1$, we have that $\E_\Pi(\D_{ij} \mid \X^0, \lambda) \leq \sum_{k \neq k^\prime} \Pi(s_i = k, s_j = k^\prime \mid \X^0, \lambda) = \Pi(s_i \neq s_j \mid \X^0, \lambda)$. Taking the expectation with respect to the posterior of $\lambda$ on both sides of the inequality results in (c). Properties (a) and (b) imply that $\mathcal{R}(\chat)$ is non-negative for any clustering and that the sum in \eqref{eq:risk} may be taken over the indices $i<j$. Property (c) states that FOLD will generally favor a fewer number of clusters than Binder's loss. One can show that Binder's loss prefers assigning $i$ and $j$ to the same cluster when $\Pi(s_i \neq s_j \mid \X^0) < \gamma$. If $\Delta_{ij} < \gamma < \Pi(s_i \neq s_j \mid \X^0)$, FOLD will disagree with Binder's loss, with FOLD preferring to merge $i$ and $j$ into the same cluster, implying that $\cFOLD$ is often more coarse then $\bs c_{\tx{B}} = \underset{\widehat{\bs c}}{\tx{argmin  }} \E_\Pi \lb \mathcal L_{\tx{B}}(\chat, \bs s) \mid \X^0 \rb$.

Recently, the Variation of Information (VI) \citep{meilua2007comparing} has attracted attention in the literature for use as a clustering loss function \citep{wade2018bayesian, de2023bayesian, denti2023common}, and we compare FOLD to the VI from a methodological point of view in the Supplement. The VI is primarily motivated from an information theoretic perspective, but it, like Binder's loss, is a metric on the partition space \citep{meilua2007comparing,dahl2022search} and encourages similarity to the component labels when used in Bayesian clustering. While the VI has been observed to yield fewer clusters than Binder's loss empirically on synthetic and real data \citep{wade2018bayesian}, we show in a simulation study in the Supplementary and an application in Section \ref{section:apps} that $\bs c_\VI$ is still vulnerable to the over-clustering problem, albeit usually to a lesser degree than $\bs c_{\tx{B}}$, whereas $\cFOLD$ is more robust. 

\subsection{Uncertainty Quantification}
 
In applications of Bayesian clustering to real data, we often find that there is substantial uncertainty in the component labels a posteriori. To avoid overstating the significance of $\cFOLD$, it is thus important to express this uncertainty in a clear and interpretable manner. We focus on the clustering that minimizes (\ref{eq:lg}), i.e. $\bs c_{\bs \theta} = \underset{\widehat{\bs c}} {\tx{argmin }}\Lc.$ Observe that $\bs c_{\bs \theta}$ depends on $\omega$ and may not be equal to $\bs s$. To see this, we recall (\ref{eq:hcutoff}), which shows that we can acquire a clustering with loss smaller than $\bs s$ by merging any two components with statistical distance below $\gamma$. Thus, we formulate our measures of uncertainty with respect to the FOLD posterior $\Pi(\bs c_{\bs \theta} \mid \X^0)$, instead of  $\Pi(\bs s \mid \X^0)$. 

Accordingly, we adapt the notion of a credible ball for Bayesian clustering estimators \citep{wade2018bayesian} to our procedure. The 95\% credible ball around $\cFOLD$ is defined as $B(\cFOLD) = \lb \bs c : \mathbb{D}(\cFOLD, \bs c) \leq \epsilon_{\tx{FOLD}} \rb$, where $\epsilon_{\tx{FOLD}} \geq 0$ is the smallest radius such that
\begin{equation*}
    \Pi \lb \mathbb{D}(\cFOLD, \bs c_{\bs \theta}) \leq \epsilon_{\tx{FOLD}} \mid \X^0 \rb \geq 0.95.
\end{equation*}
We interpret $B(\cFOLD)$ as a neighborhood of clusterings centered at $\cFOLD$ with posterior probability mass of at least $0.95$. Larger $\epsilon_{\tx{FOLD}}$ means that the FOLD posterior distributes its mass across a larger set around $\cFOLD$, implying that we are more uncertain about the cluster allocations of the point estimator. In practice, we characterize the credible ball with bounds that give a sense of the clusterings contained within it \citep{wade2018bayesian}, much like we would for an interval on the real line. The horizontal bounds of the credible ball
consist of the clusterings $\bs c \in B(\cFOLD)$ for which $\mathbb{D}(\cdot, \cFOLD)$ attains its maximum value. We can also impose further restrictions on our bounds, such as requiring that they contain the minimum or maximum number of clusters over $B(\cFOLD)$ while also maximizing $\mathbb{D}(\cdot, \cFOLD)$ amongst groupings with the same number of clusters, giving rise to the notions of vertical upper and vertical lower bounds, respectively. Alternatively, one can display the posterior similarity matrix (PSM) for $\bs c_{\bs \theta}$, $\mathcal{P}_{\tx{FOLD}} = ( \Pi \lb  c_{i \bs \theta} = c_{j \bs \theta} \mid \X^0 \rb )_{i,j}$, as a heatmap, with the entries ordered by $\cFOLD$, as is common practice for both Binder's loss and the VI. 

\subsection{Implementation with MCMC Output} \label{section:implementation}
Our methodology relies on Markov chain Monte Carlo (MCMC) samples from $\Pi(\bs \theta \mid \X^0)$, which are used to estimate pairwise distances between the localized densities, $\Delta_{ij} \approx (1/T) \sum_{t=1}^T \D_{ij}^{(t)},$ where $\D_{ij}^{(t)} =  d \lb g( \cdot; \theta_{i}^{(t)}), g(\cdot; \theta_{j}^{(t)}) \rb$, and $\theta_{i}^{(t)}, \theta_{j}^{(t)}$ can be computed directly from samples of the components $\bs s^{(t)}$ and unique values $\Tilde{\bs \theta}^{(t)}$ via $\theta_i^{(t)} = \tilde{\theta}^{(t)}_k \iff s_i^{(t)}=k$. When $d$ is chosen to be the Hellinger distance, many families of kernels, such as the Gaussian family, admit a closed form of $\D_{ij}$. The label-switching problem, which results from the inherent non-identifiability of mixture models with any permutation of the labels yielding the same likelihood \citep{redner1984mixture}, has no impact on FOLD, as $\Delta_{ij}$ is invariant to labeling. FOLD can be easily applied to virtually any mixture model for which samples from $\Pi(\bs \theta \mid \X^0)$ can be obtained and pairwise statistical distances can be estimated.

\begin{figure}[ht]
\centering
\includegraphics[scale=0.2]{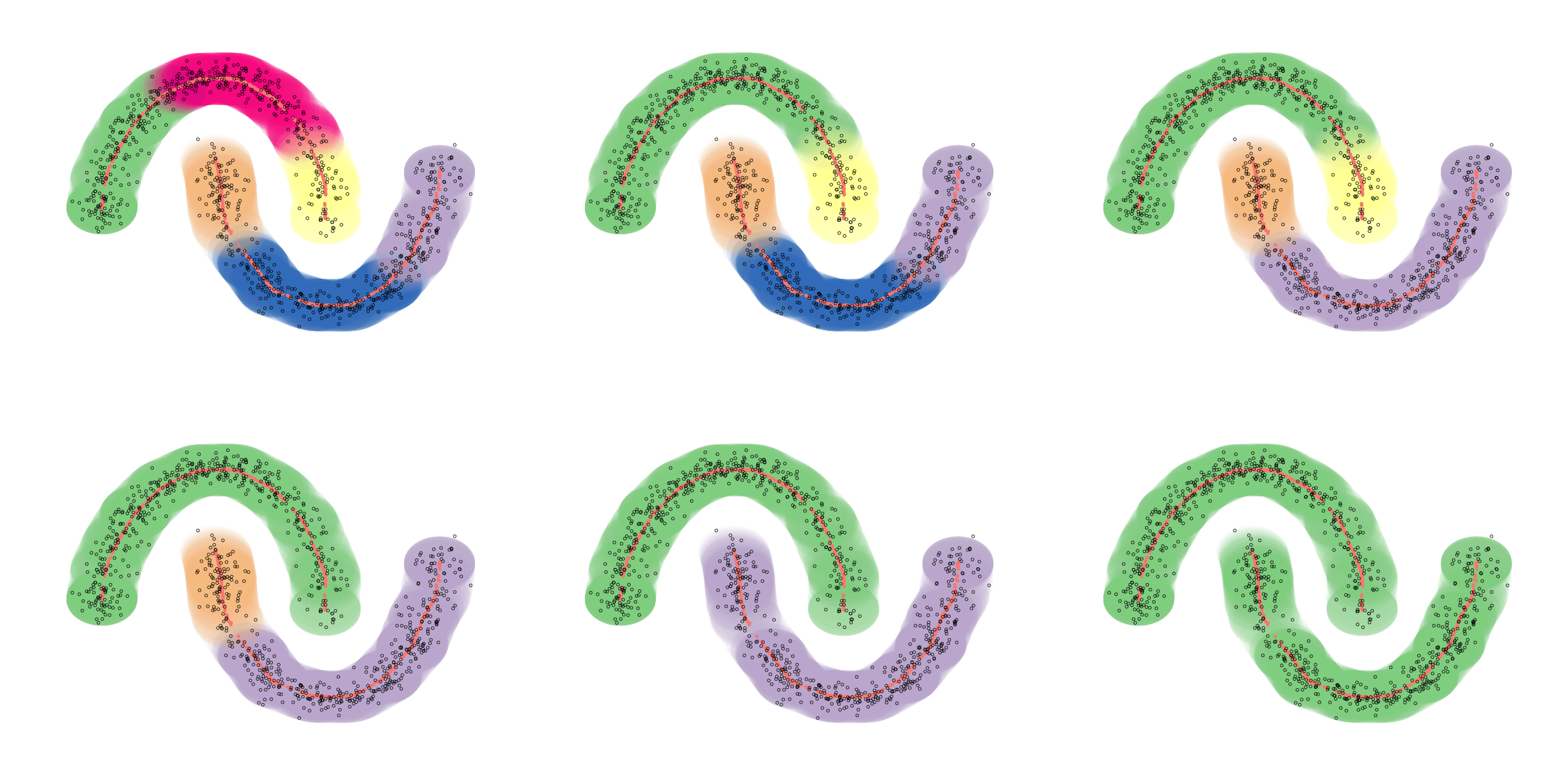}
\caption{Candidate clusterings of the \texttt{moons} data under a location Gaussian mixture model, generated from average linkage hierarchical clustering on $\Delta$.}
\label{fig:moonsclusters}
\end{figure}

In practice, the minimization of \eqref{eq:risk} must be approximated due to the enormous size of the partition space for even modest $n$. We discuss two methods for obtaining an approximate minimizer. The first is to simply reduce the set of partitions over which we minimize by only considering a tree of clusterings produced by hierarchical clustering on $\Delta$. We suggest generating candidates using average linkage due to the appearance of the average linkage dissimilarity metric as a criterion for merging clusters in \eqref{eq:hcutoff}. Figure \ref{fig:moonsclusters} displays candidate clusterings of the \texttt{moons} data resulting from average linkage clustering on $\Delta$. The true grouping of the data is present amongst the candidates. Since these groupings are hierarchical, Figure \ref{fig:moonsclusters} can be interpreted as an enumeration of clusterings favored by $\mathcal{R}(\widehat{\bs c})$, arranged by increasing values of $\omega$. A similar method with single-linkage is defined in \cite{aragam2020identifiability}, and  \cite{medvedovic2002bayesian} and \cite{fritsch2009improved} have previously used hierarchical clustering on $\textbf{1}_n \textbf{1}_n^T - \mathcal P$ to compute clustering point estimators, where $\mathcal P = (\Pi \lb s_i = s_j \mid \X^0 \rb)_{i,j}$.

Recent developments for minimizing risk functions have focused on greedy algorithms that take an existing candidate partition and make minor, but locally-optimal updates, typically by reallocation of objects to clusters. Examples include a method motivated by the Hasse diagram in \cite{wade2018bayesian}, the greedy algorithm in \cite{rastelli2018optimal}, and the SALSO algorithm of \cite{dahl2022search}. SALSO provides a point estimate by first initializing a clustering, then reallocating observations to clusters and, finally, breaking existing clusters and then reallocating observations again over several parallel runs. The final point estimate is the partition that exhibits the smallest risk across all the runs. In our simulations and application, we implement FOLD with both the hierarchical clustering heuristic and SALSO, and find that they generally agree on the optimal clustering. 

To express uncertainty in $\cFOLD$, we rely on samples from $\Pi(\bs c_{\bs \theta} \mid \X^0)$, defined by
\begin{equation} \label{eq:mcmcloss}
    \bs c_{\bs \theta}^{(t)} = \underset{\widehat{\bs c}} {\tx{argmin }} \sum_{i<j} \lb \textbf{1}_{\hat{c}_i = \hat{c}_j } \D^{(t)}_{ij} + \omega \textbf{1}_{\hat{c}_i \neq \hat{c}_j} \left(  1 - \D_{ij}^{(t)} \right) \rb.
\end{equation}
As in the point estimation case, $\bs c_{\bs \theta}^{(t)}$ cannot be computed exactly and so we approximate the minimization in \eqref{eq:mcmcloss} either using a set of candidates generated from hierarchical clusterings of $\D^{(t)} = (\D^{(t)}_{ij})_{i,j}$ or via a greedy algorithm such as SALSO. The posterior probability of a credible ball is approximated by $\Pi \lb \mathbb{D}(\bs c_{\bs \theta}, \cFOLD) \leq \epsilon \mid \X^0 \rb \approx (1/T) \sum_{t=1}^T \textbf{1}_{  \mathbb{D}(\bs c_{\bs \theta}^{(t)}, \cFOLD) \leq \epsilon }$ and FOLD computes $\epsilon_{\tx{FOLD}}$ by incrementally increasing the value of $\epsilon$ over a grid. The entries of $\mathcal{P}_{\tx{FOLD}}$ are estimated via $\Pi( c_{i \bs \theta} =  c_{j \bs \theta} \mid \X^0) \approx (1/T)\sum_{t=1}^T \textbf{1}_{c^{(t)}_{i \bs \theta} = c^{(t)}_{j \bs \theta}}$.

\subsection{Selection of $\omega$}

The loss parameter $\omega$ controls the separation among the inferred clusters in $\cFOLD$. This parameter is also influential in $\bs c_{\tx{VI}}$ and $\bs c_{\tx{B}}$ \citep{dahl2022search}, though there has been little discussion in the literature on how to select an appropriate $\omega$ for these loss functions. The formulation of \eqref{eq:lg} in terms of statistical distances leads to two sensible methods for determining the loss parameter. First, similarly to common practice in choosing key tuning parameters in algorithmic clustering methods such as k-means and DBSCAN, we propose an elbow plot diagnostic to select $\omega$. At a grid of possible $\omega$ values, we calculate $r_{\omega} = (\sum_{h=1}^{K_{\omega}^*} r_{\omega h})/{\sum_{i<j} \Delta_{ij}}$, where $r_{\omega h} = \sum_{i,j \in C_{\omega h}^*} \Delta_{ij}$, $\bs c_w^*$ is the FOLD point estimator for that specific $\omega$, and $C^*_\omega = \lb C_{\omega 1}^*, \dots, C_{\omega K_w^*}^* \rb$ is the partition associated with $\bs c_\omega^*$. The numerator in $r_w$ sums over $r_{\omega h} = \E_\Pi \lb \sum_{i,j \in C_{\omega h}^*} \D_{ij} \mid \X^0 \rb$, or the expected statistical distance between all localized densities associated with the $h$th cluster in $\bs c_\omega^*$, and the denominator normalizes so that $r_\omega \in [0,1]$ for all $\omega$. 

To build some intuition about the utility of the elbow plot, recall that as $\omega \to \infty$, \eqref{eq:risk} is minimized by placing all observations in a single cluster, i.e. $C_\infty^* = \lb [n] \rb$ and $r_\infty = 1$. In many cases, such as the data in Figure \ref{fig:moonsclusters}, decreasing $\omega$ (and subsequently adding more clusters in $C_\omega^*$) will decrease $r_\omega$, e.g. because $\Delta_{ij} \approx 1$ for any $i,j$ pairs in different crescents. Of course, at the other extreme, if we set $\omega$ to be very small, we place all objects in their own cluster, making $C_0^* = \lb \lb 1 \rb, \dots, \lb n \rb \rb$ and $r_0 = 0$. We construct the elbow plot to quantify the amount of improvement each time we increment $\omega$ across this grid. At some point, the improvement in $r_\omega$ becomes negligible as we decrease $\omega$ because we begin to cut dense clusters into smaller and smaller sub-clusters. On an elbow plot, this causes the curve of $r_\omega$ to bend, creating the so-called ``elbow". It is at this threshold that we fix $\omega$. The threshold is easier to identify when the elbow plot is monotonic, and we verify this is the case if one uses hierarchical clustering to generate candidates.

\begin{proposition} \label{prop:elbow}
    If candidate clusterings are selected using average linkage hierarchical clustering on $\Delta$, then there is a positive relationship between $\omega$ and $r_\omega$.
\end{proposition}
A corollary of Proposition \ref{prop:elbow} is that any jump in the elbow plot corresponds to decreasing the number of clusters by one, meaning that there is a negative relationship between $K_\omega^*$ and $r_\omega$. Therefore, one can relabel the x-axis of the elbow plot in terms of the number of clusters rather than $\omega$, leading to a more interpretable figure. An example of using an elbow plot to select $\omega$ is given in Section \ref{section:apps}.

Alternatively, one can use the default value $\omega^{\tx{AVG}} = \gamma^{\tx{AVG}}/(1-\gamma^{\tx{AVG}})$, with 
$\gamma^{\tx{AVG}} = {n \choose 2}^{-1} \sum_{i<j} \Delta_{ij}$. Under this choice of $\omega$, candidate clusters are combined when the average of $\Delta_{ij}$ between the clusters is smaller than the average of $\Delta_{ij}$ across the entire sample. $\gamma^{\tx{AVG}}$ estimates $\bar{\D} = {n \choose 2}^{-1} \sum_{i<j} \D_{ij} = \sum_{k<k^\prime} (|S_k||S_{k^\prime}|/{n \choose 2}) d \lb g(\cdot ; \tilde \theta_{k}), g( \cdot, \tilde \theta_{k^\prime}) \rb$, a weighted pairwise sum across the components. If we use $\bar{\D}/(1-\bar{\D})$ in place of $\omega$ in (\ref{eq:lg}),  (\ref{eq:hcutoff}) implies that FOLD will favor merging mixture components when $d \lb g(\cdot; \tilde \theta_k), g(\cdot; \tilde \theta_{k^\prime}) \rb < \bar{\D}$. Importantly, the decision to merge components will depend on how separated they are from the others and their sizes. To see this, consider the following example, in which we fit a mixture with $K=3$ components, where $d\lb g(\cdot; \tilde \theta_1), g(\cdot; \tilde \theta_2) \rb = \epsilon>0$, $d \lb g(\cdot; \tilde \theta_1), g(\cdot; \tilde \theta_3) \rb = d \lb g(\cdot; \tilde \theta_2), g(\cdot; \tilde \theta_3) \rb = \delta >0$, and $|S_1|=|S_2|=|S_3| = n/3$. Then, as $n$ grows large, $\bar{\D} \to (2/9) \epsilon + (4/9) \delta$, and so FOLD will favor merging $S_1$ and $S_2$ into one cluster when $\epsilon < (4/7) \delta$. The smaller $\delta$ is, the smaller $\epsilon$ must be in order to merge $S_1$ and $S_2$. Hence, $\omega^{\tx{AVG}}$ excels at problems in which $f^0$ is composed of well-separated kernels that are approximated by multiple components in $f$. We provide simulation studies in the Supplement that show that $\omega^{\tx{AVG}}$ performs very well in this setting.

\section{Asymptotic Analysis} \label{section:theory}

Though there is a rich literature on asymptotic properties of Bayesian mixture models for estimating the density \citep{ghosal1999posterior,ghosal2007posterior},  number of components \citep{rousseau2011asymptotic, miller2014inconsistency, cai2021finite, ascolani2023clustering}, and the mixing measure \citep{nguyen2013convergence, ho2016convergence, ho2016strong,guha2021posterior}, little attention has been given to the large sample behavior of clustering estimators. \cite{rajkowski2019analysis} focused on the maximum a posteriori (MAP) estimator for $\bs s$ when data are generated from a DPM with Gaussian components and the kernels are correctly specified. They show multiple properties of the MAP estimator, including the key result that the intersection between the convex hulls of two clusters is at most a single observation. In this section, we show convergence of FOLD towards the \textit{oracle clustering procedure} in which the objects are partitioned into groups using knowledge of the true dating generating process. We take a different approach than \cite{rajkowski2019analysis} by focusing on posterior contraction of the mixing measure $\lambda$, which allows us to consider misspecifed models, either in the kernels or number of components. 

\subsection{Assumptions}
We first require that the parameter space $\Theta \subset \mathbb{R}^p$ is compact and $f$, defined formally as $f(x) = \int g(x;\theta) \lambda(d \theta)$, is identifiable with respect to the mixing measure, that is,  $\int g(x;\theta) \lambda_1(d \theta) = \int g(x;\theta) \lambda_2(d \theta)$ for all $x \in \chi$ if and only if $\lambda_1 = \lambda_2$ \citep{teicher1961identifiability}. Identifiability of the mixing measure is satisfied by location and location-scale GMMs, as well as various exponential family mixtures \citep{barndorff1965identifiability} and location family mixtures \citep{teicher1961identifiability}. More specifically, our results rely on contraction of the mixing measure $\lambda$ to some oracle measure $\lambda^*$, which implies convergence of $f$ to an oracle model $f^* = \int g(\cdot;\theta) \lambda^*(d \theta)$. An oracle measure is defined as any Kullback-Liebler (KL) divergence minimizer between our class of models and the true data generating process, i.e. $\lambda^* \in \underset{\lambda \in \Lambda}{\tx{argmin }}\tx{KL}(f^0, f)$, where recall $\Lambda = \tx{supp}(\pi_\Lambda)$. We now present the general assumptions required for our results.
\begin{assumption} \label{assump:general-assumptions}
    Suppose that the following conditions hold.
    \begin{itemize}
        \item[] (A1) There exists an $L>0$ such that $|\tx{supp}(\lambda)| \leq L$ for all $\lambda \in \Lambda$.
        \item[] (A2) The KL minimizer $\lambda^* = \sum_{m=1}^{M^*} a_m^* \delta_{\theta_m^*}$ exists and is unique.
        \item[] (A3) $f^*$ and $f^0$ are such that $\tx{supp}(f^0) \subseteq \tx{supp}(f^*)$.
        \item[] (A4): $\mathcal{G} = \lb \Tilde{g}(x-\theta) : \theta \in \Theta \rb$ for some bounded probability density function $\Tilde{g}(\cdot)$, $\tilde{g}(\cdot)$ is $\zeta$-H{\"o}lder continuous for some $\zeta>0$, and for any fixed $\theta^\prime$, $D_g(\theta, \theta^\prime) = d \lb g(\cdot; \theta), g(\cdot;\theta^\prime) \rb$ is continuous in $\theta$.
        \item[] (A5) There exists a non-negative sequence $\epsilon_n$ so that $\epsilon_n \to 0$ and 
        \begin{equation}
            \rho_n(\X^0) = \Pi \lb \lambda \in \Lambda: W_2(\lambda, \lambda^*) \genq \epsilon_n \mid \X^0 \rb \overset{\P^0}{\longrightarrow} 0.
        \end{equation}
    \end{itemize}
\end{assumption}
(A1) is satisfied by any mixture model in which the number of components is bounded, and examples include a finite mixture model, a truncated MFM, or truncated DPM. Assumption (A2) ensures that the limiting values of the clustering procedure are well-defined and exist. A simple example of when (A2) holds is when the kernels are correctly specified and $\Pi(K=M^0) >0$, in which case $\lambda^* = \lambda^0$ and $f^* = f^0$. We refer to the case in which $\Pi(K=M^0)=1$ as the exact-fitted and well-specified regime, which is permitted under assumptions (A1)-(A2). Under any misspecification, (A2) holds when $\min_{\lambda \in \Lambda} \tx{KL}(f^0,f) = \min_{\lambda \in \Omega} \tx{KL}(f^0,f)$, where $\Omega$ is the space of all probability measures over $\Theta$ \citep{guha2021posterior}. That is, assumption (A2) is satisfied in the misspecified regime when the minimum KL divergence is attained over $\Lambda$. This still allows for an exact-fitted and misspecified regime, in which $\Pi(K=M^*)=1$. (A3) is a technical assumption to ensure there are no regions in the sample space under which $f^0$ gives positive probability but $f^*$ does not. Assumption (A4) restricts our results to location families, and the continuity conditions are satisfied by a variety of models including the location GMM with fixed covariance. Assumption (A5) is a crucial statement on the estimation of parameters in a mixture model. (A5) will hold under various Lipschitz and strong identifiability conditions on $\mathcal G$ \citep{nguyen2013convergence,ho2016convergence,guha2021posterior} that are satisfied by the location GMM \citep{manole2022refined}, and we discuss these conditions in detail in the Supplementary Material. (A5) will hold with $\epsilon_n = (\log n/n)^{1/4}$ for second-order identifiable families, and first-order identifiable families with $\epsilon_n = (\log n/n)^{1/2}$ if the kernels are correctly specified and $\Pi(K=M^0)=1$ \citep{guha2021posterior}; see \cite{ho2016convergence,ho2016strong} for a characterization of first and second-order identifiable kernels.

\subsection{Convergence to The Oracle Rule}

Synonymous with the oracle mixing measure is the oracle clustering procedure or rule. We define the oracle clustering procedure to be the minimizer of the oracle risk function $\mathcal R^*(\widehat{\bs c}) = \E_\Pi[\Lc \mid \lambda^*, \X^0]$,
\begin{equation} \label{eq:oracle-procedure}
    \cFOLD^* = \underset{\widehat{\bs c}} {\tx{argmin }} \mathcal R^*(\chat) = \underset{\widehat{\bs c}} {\tx{argmin }} \sum_{i<j} \lb \textbf{1}_{\hat{c}_i = \hat{c}_j} \Delta_{ij}^* + \omega \textbf{1}_{\hat{c}_i \neq \hat{c}_j} (1 - \Delta_{ij}^*)   \rb,
\end{equation}
where $\Delta^*_{ij} = \E _\Pi(\mathcal D_{ij} \mid \lambda^*, \X^0) = \sum_{m < m^\prime} d \lb g(\cdot; \theta_m^*), g(\cdot; \theta_{m^\prime}^*) \rb q_{ij}^{mm^\prime*}$, with $q_{ij}^{m m^\prime *} = \Pi(s_i = m, s_j = m^\prime \mid \X^0, \lambda^*) + \Pi(s_i = m^\prime, s_j = m \mid \X^0, \lambda^*)$, and
\begin{equation*}
    \Pi(s_i = m, s_j = m^\prime \mid \X^0, \lambda^*) = \frac{a_m^* a_{m^\prime}^* g(X_i^0; \theta_m^*) g(X_j^0; \theta_{m^\prime}^*) }{f^*(X_i^0) f^*(X_j^0)}.
\end{equation*}
Observe that $\Delta_{ij}^*$ is a weighted sum of the total statistical distance between the oracle components, where the weight of each component pair is given by the conditional probability that $s_i = m$ and $s_j = m^\prime$ (or vice versa) given $\X=\X^0$ and $\lambda = \lambda^*$. We interpret the oracle clustering procedure as a rule for grouping the observations in $\X^0$ based on the knowledge of the optimal parameter values in the mixture model. We can construct two simple examples of oracle rules which follow from the definition in \eqref{eq:oracle-procedure}.
\begin{proposition} \label{prop:oracles}
    Instances of the oracle clustering procedure include:
    \begin{enumerate}[label=(\alph*)]
        \item If $ d\lb g(\cdot; \theta_m^*), g(\cdot; \theta_{m^\prime}^*) \rb = 1$ for all $m,m^\prime$, then $\cFOLD^* = \bs s^*$, where $s^*_i = s^*_j$ if and only if there exists a unique $m \in [M^*]$ such that $g(X_i^0; \theta_m^*)g(X_j^0;\theta_m^*) > 0$. 
        \item  If $M^*=2$ and $ d\lb g(\cdot; \theta_1^*), g(\cdot; \theta_{2}^*) \rb < \gamma$, then $\cFOLD^* = \bs c_0$, where $c_{0i} = c_{0j}$ for all $i,j$. 
    \end{enumerate}
\end{proposition}

We refer to the procedure on Proposition \ref{prop:oracles}(a) as the ``match rule" in that it allocates each observation to a mixture component with compatible support, and call the procedure in Proposition \ref{prop:oracles}(b) the ``merge rule", i.e., as the statistical distance between oracle components becomes smaller, we are more likely to combine them and subsequently place all observations in one cluster. A similar quantity exists for Binder's loss function, denoted $\bs c_{\tx{B}}^* = \underset{\widehat{\bs c}}{\tx{argmin }} \E_\Pi[\mathcal{L}_{\tx{B}}(\chat, \bs s) \mid \lambda^*, \X^0]$. Under the settings of Proposition \ref{prop:oracles}(a), $\bs c_{\tx{B}}^* = \cFOLD^* = \bs s^*$, meaning that the match rule is the optimal clustering procedure under both FOLD and Binder's loss when the oracle components are perfectly separated. Intuitively, this is because these settings reflect an ideal scenario in which all mixture components are perfectly separated and no merging is required. Otherwise, say in the setting of Proposition \ref{prop:oracles}(b), the two methods can differ. A sufficient condition for $\bs c_{\tx{B}}^* = \bs c_0$ is $q_{ij}^{12} < \omega (1 - q_{ij}^{12})$ for all pairs $i,j$. Any observations in the tails of the two oracle components will blow up the value of $q_{ij}^{12}$, hence, it is possible that $\bs c_{\tx{B}}^* \neq \bs c_0$ even for relatively large $\omega$. This implies that a concrete advantage of our approach is robustness to outliers when the oracle components are not well separated. We now state our primary theoretical result on the relationship between $\cFOLD$ and $\cFOLD^*$.
\begin{theorem}
\label{thm:consistency}
Fix $0 < \delta < 1$. Then under assumptions (A1)-(A5), with $\P^0$-probability tending to $1$,
\begin{equation} \label{eq:rvns-consistency}
   {n \choose 2}^{-1} | \mathcal R(\cFOLD) - \mathcal R^* (\cFOLD^*) | \lenq \max(\omega, 1) \lb \rho_n(\X^0) \bar{\Delta}^* + \mathcal{Q}(\X^0) \rb,
\end{equation}
where $\bar{\Delta}^* = {n \choose 2}^{-1} \sum_{i<j} \Delta^*_{ij}$, $\mathcal{Q}(\X^0) = {n \choose 2}^{-1} \sum_{i<j} \tau^*_n(X_i^0, X_j^0)$, and $\tau_n^*(X_i^0, X_j^0)$ is a function of $(X_i^0, X_j^0, \epsilon_n, \rho_n(\X^0), \lambda^*, \zeta, \delta)$ so that $\tau_n^*(X_i, X_j) \overset{\P^0}{\longrightarrow} 0$. 
\end{theorem} 
To prove Theorem \ref{thm:consistency}, we first note that we can write $\Delta_{ij}$ as the posterior expectation of a quantity that only depends on $\lambda = \sum_{k=1}^K a_k \delta_{\tilde \theta_k}$,
\begin{equation} \label{eq:rvns-delta-integral}
    \Delta_{ij} = \int_{\lambda \in \Lambda} \Pi(d \lambda \mid \X^0) \sum_{k < k^\prime} D_g(\tilde \theta_k, \tilde \theta_{k^\prime})  a_k a_{k^\prime} \frac{g(X_i^0;\tilde \theta_k)g(X_j^0;\tilde \theta_{k^\prime}) + g(X_j^0;\tilde \theta_k)g(X_i^0;\tilde \theta_{k^\prime})}{f(X_i^0)f(X_j^0)}.
\end{equation}
Next, we decompose $\Lambda = B_{\epsilon_n}(\lambda^*) \cup (\Lambda \setminus B_{\epsilon_n}(\lambda^*))$, where $B_{\epsilon_n}(\lambda^*) = \lb \lambda \in \Lambda: W_2(\lambda, \lambda^*) \lenq \epsilon_n \rb$, and split \eqref{eq:rvns-delta-integral} into a sum of integrals over these domains. Since $0 \leq D_g(\theta, \theta^\prime) \leq 1$, the integral over $\Lambda \setminus B_{\epsilon_n}(\lambda)^*$ is bounded above by $\rho_n(\X^0)$, which goes to $0$ in probability by (A5). Then, using a similar proof technique given in \cite{nguyen2013convergence}, we show that there exists a finite $N$ so that for all $n \geq N$ and $\lambda \in B_{\epsilon_n}(\lambda^*)$, $\lambda$ has at least $M^*$ support points and for any $m \in [M^*]$ there exists a set of indices $\mathcal{I}_{m}^{(\lambda)} \subset [K]$ so that $\max_{k \in \mathcal{I}_m^{(\lambda)}} \norm{\tilde \theta_k - \theta_m^*} \lenq \epsilon_n$ and $\bigg | \sum_{k \in \Ilam_m} a_k - a_m^*  \bigg | \lenq \max(\epsilon_n, \epsilon_n^{2 \delta})$. Furthermore, if we denote $\Ilam = \bigcup_{m=1}^{M^*} \Ilam_m$, (A5) also implies that $\sum_{k \not \in \Ilam} a_k \lenq \epsilon_n^{2 \delta}$. By (A4), these bounds translate into similar upper and lower bounds for $\Delta_{ij}$. We then sum over all $i,j$ pairs and show that ${n \choose 2}^{-1}|\mathcal R(\chat) - \mathcal R^*(\chat)|$ is uniformly bounded by the right hand side of \eqref{eq:rvns-consistency} for all $\chat$. That is, the rate of convergence of $\mathcal R(\chat)$ to $\mathcal R^*(\chat)$ does not depend on the clustering estimator itself. Finally, we show that this implies convergence of the minimizers at the same rate. 

The bound in \eqref{eq:rvns-consistency} is split into two remainder terms. It follows from (A5) that $\rho_n(\X^0) \bar{\Delta}^* \overset{\P^0}{\longrightarrow} 0$ since $0 \leq \bar{\Delta}^* \leq 1$. Therefore, smaller values of $\bar{\Delta}^*$ will cause this term to diminish faster. This can arise when $f^*$ is comprised of many closely overlapping components. Hence, convergence to the oracle rule will actually benefit from overfitting. $\mathcal Q(\X^0)$ is a more general remainder term that results from translating the posterior contraction of $\tilde \theta_k$ and $a_k$ to the sum in \eqref{eq:rvns-delta-integral}. In particular, this remainder term will depend on $\delta$, which controls the rate at which the irrelevant mixture components are emptied, the H{\"o}lder constant $\zeta$, which controls convergence of $g(X_i^0; \tilde \theta_k)$ to $g(X_i^0;\theta_{m}^*)$, and $L-M^*$, or the degree to which the number of components is misspecified, with larger $L$ implying a slower rate of convergence. 

In the Supplementary Material we empirically validate Theorem \ref{thm:consistency} on a simple simulated example. We simulate data from $M^0=4$ Gaussian components with covariance equal to $0.25I$, then fit a $4$-component Bayesian mixture model, meaning that $f^0(x) = f^*(x) = \sum_{m=1}^4 a_m^0 \N(x;\mu_m^0, 0.25I)$. We compute $\cFOLD$, $\cFOLD^*$, $\bs c_{\tx{B}}$, and $\bs c_{\tx{B}}^*$, the latter of which denotes the oracle clusters corresponding to Binder's loss, for increasing sample size. We show that both $\mathcal R(\chat) \to \mathcal R^*(\chat)$ and $\cFOLD \to \cFOLD^*$ as $n \to \infty$. Additionally, there are substantial differences between $\cFOLD^*$ and $\bs c_{\tx{B}}^*$ because oracle FOLD corresponds to the merge rule for some of the Gaussian components, while oracle Binder prefers estimation of each individual component. 

\section{Example: GSE81861 Cell Line Dataset} \label{section:apps}

We apply FOLD to the GSE81861 cell line dataset \citep{li2017reference}, which measures single cell counts in $57,241$ genes from $630$ single-cell transcriptomes. There are $7$ distinct cell lines present in the dataset, so we compare FOLD and other model-based clustering methods to the true cell line labels as a performance benchmark. First, we apply routine pre-processing steps for RNA sequence analysis (as in \cite{chandra2023escaping}). We discard cells with low read counts, giving $n = 519$ total cells. We then normalize the data using SCRAN \citep{lun2016pooling} and select informative genes with M3Drop \citep{andrews2019m3drop}. We use principal component analysis (PCA) for dimension reduction by taking $\X^0$ to be the projection of the normalized cell counts onto the first $p = 5$ principal components, then scale the projections to have zero mean and unit variance. 

We fit a $p$-dimensional Bayesian Gaussian finite mixture to $\X^0$ with $K = 50$ components, and simulate $25,000$ posterior samples after a burn-in of $1,000$. Every fourth iteration is discarded, leaving $6,000$ samples remaining. Along with $\cFOLD$, we calculate model-based clusterings with the VI loss, Binder's loss, and \texttt{mclust} \citep{scrucca2016mclust}, the latter of which implements the EM algorithm and chooses the number of clusters via the Bayesian information criterion (BIC). We also compare to algorithmic clustering methods, including average linkage hierarchical clustering (HCA), k-means, DBSCAN, and spectral clustering. Details on hyperparameter choices for all methods are given in the Supplement. For FOLD, we create candidate clusterings by average-linkage clustering on $\Delta$, then choose $\omega$ by consulting the elbow plot in Figure \ref{fig:rnaelbow}. The plot suggests $6$ clusters, which corresponds to $\omega = 25$. We also implement SALSO over $150$ independent runs to minimize \eqref{eq:risk}. However, over repeated replications we find that the SALSO estimate is either equal to or has a higher risk than the point estimate given by hierarchical clustering on $\Delta$. Binder's loss and the VI loss are implemented with $\omega=1$ (the default value in \texttt{mcclust} \citep{fritsch2022package}), which results in $\bs c_\VI = \bs c_{\tx{B}}$.

\begin{figure}[ht]
    \centering
    \includegraphics[scale=0.35]{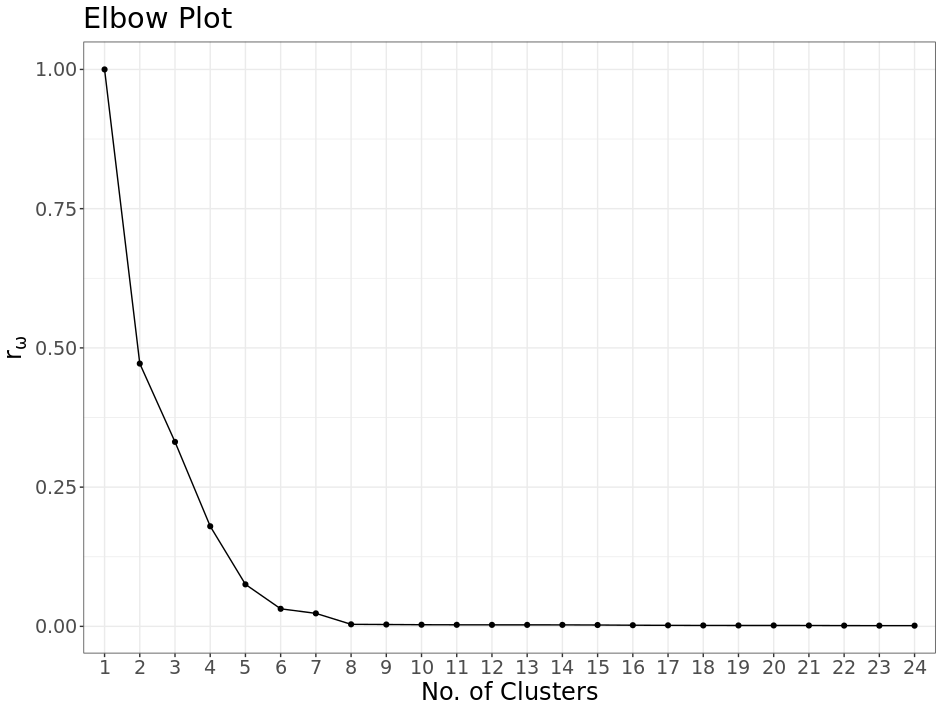}
    \caption{Elbow plot for choosing the number of clusters in $\cFOLD$ with the cell line dataset.
    }
    \label{fig:rnaelbow}
\end{figure}

\begin{figure}
    \centering
    \includegraphics[scale=0.49]{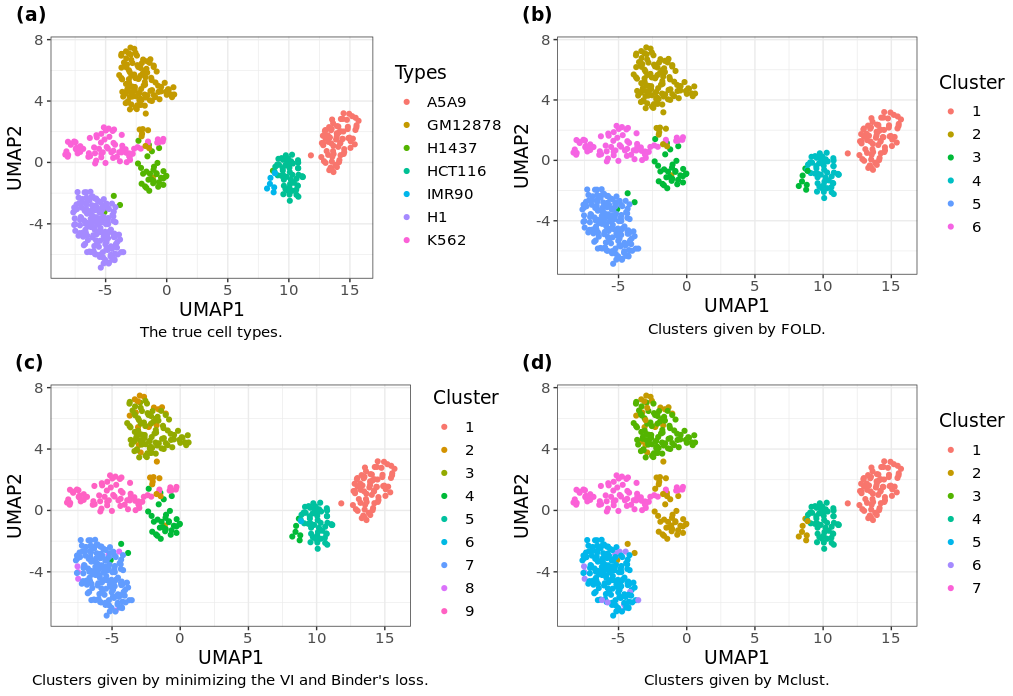}
    \caption{UMAP plots of the cell line dataset with colors corresponding to the true cell types, $\cFOLD$, VI/Binder, and \texttt{mclust}, respectively.}
    \label{fig:umap}
\end{figure}

 Figure \ref{fig:umap} shows the UMAP plots \citep{mcinnes2018} of the original normalized count data along with colors indicating the true cell types and the clusterings from the three model-based methods. The adjusted Rand index (ARI) \citep{hubert1985comparing} with the true cell types and number of clusters ($\hat K_n$) for $\cFOLD$ and all other clusterings are given in Table \ref{table:cells-results}. First, observe that the Bayesian clustering methods outperform the competitors in the ARI, indicating that the overfitted Gaussian components are accurately approximating the distribution within each type. However, FOLD notably excels over the other Bayesian methods with an ARI of $0.995$, only misclassifying six observations.
Note that the VI, Binder's loss, and \texttt{mclust} split the GM12878 and H1 cell types into two clusters each, while FOLD correctly identifies both types, keeping each as single clusters. The splitting of types is not unexpected in this application since the GM12878 and H1 types were each sequenced in two separate batches \citep{li2017reference}, indicating that FOLD is more robust to the batch effect than other Bayesian methods. However, $\cFOLD$ underestimates the number of groups by combining the H1437 and IMR90 cell types. The other model-based methods merge these types as well, with the VI and Binder's loss producing $9$ clusters and \texttt{mclust} giving $7$ clusters. Of the algorithmic methods, k-means performs the best, estimating the correct number of types while achieving an ARI of $0.904$. However, k-means splits the GM12878 and A549 types into two clusters each, the latter of which is always kept as one cluster by the model-based methods.

\begingroup
\renewcommand{\arraystretch}{0.49}
\begin{table}[ht]
\centering
\begin{tabular}{rrrrrrrrr}
  \hline
 & FOLD & VI & Binder & Mclust & HCA & K-Means & DBSCAN & Spectral \\ 
  \hline
ARI & 0.995 & 0.915 & 0.915 & 0.854 & 0.622 & 0.904 & 0.679 & 0.620 \\ 
  $\hat K_n$ & 6 & 9 & 9 & 7 & 7 & 7 & 5 & 5 \\ 
   \hline
\end{tabular}
\caption{Adjusted Rand index (ARI) with the true cell lines and the number of clusters ($\hat{K}_n$) for FOLD and competitors on the cell line dataset.
}
\label{table:cells-results}
\end{table}
\endgroup

\begin{figure}
    \centering
    \includegraphics[scale=0.49]{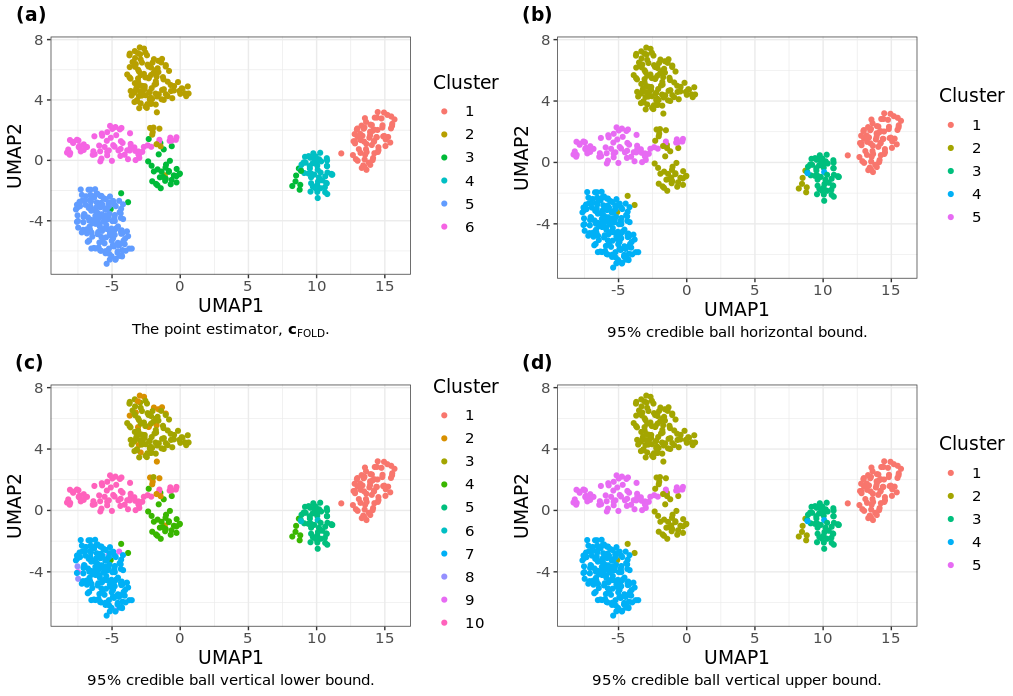}
    \caption{The clustering $\cFOLD$ along with the horizontal, vertical lower, and vertical upper bounds for the cell line dataset. Here, $\mathbb D(\cdot, \cdot)$ is the VI and the bounds are unique.}
    \label{fig:rnaCB}
\end{figure}

The 95\% credible ball around $\cFOLD$ is displayed in Figure \ref{fig:rnaCB}. The credible ball communicates that there is substantial uncertainty in cell types where batch effects occur. The horizontal and vertical upper bounds effectively merge the GM12878 type with the H1437 cell type, which is likely due to the proximity of these two types in the sample space. Conversely, the vertical lower bound splits the GM12878 cell type into two clusters. We are also uncertain in our classification of the H1 cell type, which is similarly split into multiple clusters by the vertical lower bound because of separate batching. Interestingly, in all bounds, the H1437 and IMR90 cell types are allocated to the same cluster. This type merging could be explained by the fact that IMR90 consists of a small number of cells or that  both types are isolated from the lung \citep{li2017reference}. In the Supplementary Material, we further evaluate our methodology by applying FOLD and competitors to six clustering datasets.

\section{Discussion} \label{section:discussion}

Fusing of Localized Densities (FOLD) is a Bayesian method for cluster analysis that characterizes clusters as possibly containing multiple mixture components. We first detailed the decision theoretic justification behind our approach, in which we obtain a point estimate of the clustering by minimizing a novel loss function. Our loss function has several appealing properties, such as favoring the merging of overlapping kernels, simplification to Binder's loss when all the mixture components are well separated, and invariance to permutations of the labels. Uncertainty in cluster allocations is expressed with a credible ball and posterior similarity matrix. We have given concrete guidance on tuning the loss parameter, including an elbow plot method and default value that performs excellently on simulated examples.

Throughout the article, we have primarily focused on the Gaussian mixture model because of its ubiquity in the literature and useful theoretical properties. However, FOLD can be applied to any parametric mixture in which a bounded statistical distance between kernels is simple to compute. For the Hellinger distance, this includes the beta, exponential, gamma, and Weibull families. A near identical approach can be applied to discrete families, where localized mass functions would replace the role of the localized densities.

We have shown that under regularity conditions, our clustering point estimator converges to the oracle clustering procedure as $n \to \infty$ for both misspecified and well-specified kernel regimes. The oracle rule categorizes observations into clusters using knowledge about the KL-minimizer. We gave two concrete examples of oracle rules, including the match rule and the merge rule. Our findings imply similar behavior for Binder's loss, meaning that analogous rules could be potentially derived for other methods such as the VI. These results would provide valuable insight into the overall asymptotic performance of Bayesian clustering. In addition, our results could be extended by focusing on contraction of the mixing measure in terms of the Voronoi loss function \citep{manole2022refined}, a loss specifically intended for strongly identifiable families that has been shown to contract at a faster rate. 

Though we implemented FOLD with the Hellinger distance in our simulations and application, other distribution metrics could be used instead. For a general statistical distance $d$, one could set $\mathcal{D}_{ij} = 1- \exp \lb  - d(g(\cdot; \theta_i), g(\cdot; \theta_j)) \rb$. The loss of co-clustering $i$ and $j$ is $1- \exp \lb  - d(g(\cdot; \theta_i), g(\cdot; \theta_j)) \rb$. This implies that we would favor co-clustering $i$ and $j$ when $d(g(\cdot; \theta_i), g(\cdot; \theta_j)) < - \log (1- \gamma)$. For $\theta_i \in \R^p$, an even simpler variant is to choose $\D_{ij} = 1 - \exp \lb - \rho \norm{\theta_i - \theta_j}^2 \rb$, where $\rho > 0$ is some fixed bandwidth, which would promote merging components with similar atoms. Alternatively, one can view $\Delta$ as a stochastic distance matrix and employ k-medoids, hierarchical clustering, or spectral clustering to cluster the data. Finally, the general notion of specifying loss functions for Bayesian clustering using $\bs \theta$ and not just $\bs s$ provides a promising direction for future research.

\section*{Acknowledgements}
This work was supported by the National Institutes of Health (NIH) under Grants 1R01AI155733, 1R01ES035625, and 5R01ES027498; the United States Office of Naval Research (ONR) Grant N00014-21-1-2510; and the European Research Council (ERC) under Grant 856506. 

\bibliography{sources}
\bibliographystyle{apalike}

\appendix

\section*{Supplementary Material}

The Supplementary Material contains additional simulation studies, proofs for all results, and technical details. Our methodology is implemented in the package \texttt{foldcluster}, available at \url{https://github.com/adombowsky/FOLD}.

\section{Point Estimation Simulation Studies} \label{section:sims-supplementary}

In light of the asymptotic results shown in Section \ref{section:theory}, we examine the effect of increasing $n$ on FOLD and other model-based clustering methods. For a given $f^0$, we repeatedly simulate observations with varying $n \in \lb 100, 500, 1000, 2500 \rb$ for $100$ replications each. We compare FOLD with $\omega = \omega^{\tx{AVG}}$ to clusterings returned by minimizing Binder's loss and the VI loss using the packages \texttt{mcclust} \citep{fritsch2022package} 
and \texttt{mcclust.ext} \citep{wade2015package},
respectively. Along with these Bayesian approaches, we also cluster the data using \texttt{mclust} \citep{scrucca2016mclust}, which groups observations based on the EM algorithm for a Gaussian mixture. 

For each replication and $n$, we run the EM algorithm and fit a Bayesian location-scale GMM,
\begin{align*}
    X_i^0 \mid \bs a, \bs \mu, \bs \Sigma & \sim \sum_{k=1}^{K} a_l \N_2(\tilde \mu_k, \tilde \Sigma_k); \label{eq:gmm} &
    (\tilde \mu_{k}, \tilde \Sigma_k) & \sim \mathcal{NIW}(\mu, \kappa, \nu, \Psi); &
    \bs a & \sim \tx{Dir} \left( \alpha 1_K \right);
\end{align*}
where $1_K$ is the $K$-dimensional vector consisting of ones. We then compute clusterings with FOLD, VI, and Binder's loss using the posterior samples from the Bayesian model. For each clustering, we save the number of clusters and the adjusted Rand index \citep{rand1971objective, hubert1985comparing} with $\bs s^0$. The hyperparameters are set at $K=30$, $\alpha = 1/2$, and $\mu = 0_p$, $\kappa = 1$, $\nu = p+2$, and $\Psi = I$. We run a Gibbs sampler for $9,000$ iterations with a burn-in of $1,000$, then use every third iteration for computing the Bayesian clusterings. For \texttt{mclust}, the number of clusters is automatically selected each replication by minimizing the Bayesian information criterion (BIC). 

\begin{figure}
\centering
\includegraphics[scale=0.9]{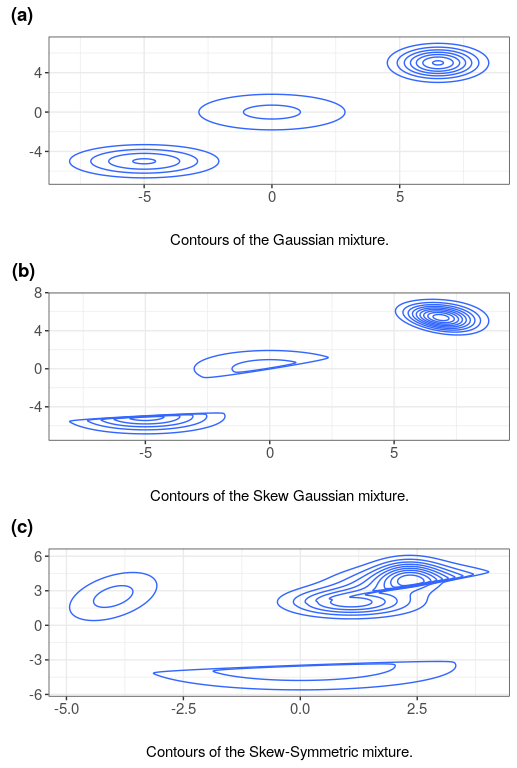}
\caption{Contour plots of the (a) Gaussian mixture, (b) skew Gaussian mixture and (c) skew-symmetric mixture.}
\label{fig:contours}
\end{figure}

\begin{table}[ht!]
\centering
\footnotesize
\begin{tabular}{lcccccc}
\toprule
 &  $n$  & FOLD    & Oracle FOLD     & VI     & Binder's      & Mclust       \\ \midrule
No. of Clusters & 100  & 2.93 (0.383) & 3.00 (0.00) & 8.02 (5.714) & 19.76 (0.818) & 3.10 (0.414) \\ 
& 500  & 3.06 (0.278) & 3.00 (0.00) & 3.53 (0.926)  & 19.89 (0.852) & 3.00 (0.00) \\ 
& 1000 & 3.05 (0.219) & 3.00 (0.00) & 3.30 (0.541)  & 19.96 (0.400) & 3.01 (0.100) \\ 
& 2500 & 3.05 (0.219) & 3.00 (0.00) & 3.08 (0.273)  & 20.00 (0.00) & 3.00 (0.00) \\ \hline
Adj. Rand Index & 100  & 0.903 (0.087) & 0.985 (0.021) & 0.852 (0.088) & 0.822 (0.061) &  0.962 (0.052) \\ 
& 500  & 0.979 (0.013) & 0.998 (0.008) & 0.974 (0.014) & 0.947 (0.012) & 0.987 (0.008) \\ 
& 1000 & 0.985 (0.006) & 0.987 (0.006) & 0.982 (0.007) & 0.968 (0.007) & 0.985 (0.024)\\ 
& 2500 & 0.986 (0.004) & 0.987 (0.004) & 0.984 (0.004) & 0.979 (0.005) & 0.987 (0.004)\\ \bottomrule 
\end{tabular}
\caption{Averages and standard deviations (in parentheses) for the number of clusters and adjusted Rand index with $\bs s^0$ on 100 replications from a mixture of bivariate Gaussian kernels.}
\label{table:symmtab}
\end{table}

\begin{table}[ht!]
\centering
\footnotesize
\begin{tabular}{lccccc}
\toprule
 & $n$  & FOLD         & VI           & Binder's      & Mclust       \\ \midrule
No. of Clusters & 100  & 3.02 (0.141) & 3.16 (0.395) & 13.84 (6.080) & 3.14 (0.377) \\ 
& 500  & 3.03 (0.171) & 3.42 (0.755) & 6.15 (2.258)  & 3.23 (0.529) \\ 
& 1000 & 3.20 (0.449) & 3.53 (0.881) & 9.32 (3.816)  & 4.33 (1.429) \\ 
& 2500 & 3.19 (0.443) & 3.40 (0.682) & 8.93 (2.409)  & 7.99 (0.959) \\ \hline
Adj. Rand Index & 100  & 0.992 (0.016) & 0.987 (0.021) & 0.915 (0.043) & 0.989 (0.030) \\ 
& 500  & 0.999 (0.003) & 0.998 (0.003) & 0.992 (0.006) & 0.980 (0.049) \\ 
& 1000 & 0.999 (0.002) & 0.997 (0.007) & 0.990 (0.009) & 0.866 (0.149) \\ 
& 2500 & 0.999 (0.001) & 0.990 (0.020) & 0.952 (0.014) & 0.576 (0.071) \\ \bottomrule 
\end{tabular}
\caption{Averages and standard deviations (in parentheses) for the number of clusters and adjusted Rand index with $\bs s^0$ on 100 replications from a mixture of bivariate skew Gaussian kernels.}
\label{table:skewtab}
\end{table}

\begin{table}[ht!]
\centering
\footnotesize
\begin{tabular}{lccccc}
\toprule
 & $n$  & FOLD         & VI           & Binder's      & Mclust       \\ \midrule
No. of Clusters & 100  & 3.25 (0.500) & 3.89 (1.014) & 12.62 (4.292) & 3.56 (0.925) \\ 
& 500  & 3.23 (0.601) & 3.64 (0.894) & 13.94 (4.552)  & 4.17 (0.779) \\ 
& 1000 & 3.10 (0.302) & 4.83 (1.025) & 15.73 (3.795)  & 5.08 (0.442) \\ 
& 2500 & 3.10 (0.302) & 4.58 (0.843) & 18.99 (2.634)  & 6.20 (0.876) \\ \hline
Adj. Rand Index & 100  & 0.982 (0.029) & 0.974 (0.026) & 0.800 (0.094) & 0.884 (0.177) \\ 
& 500  & 0.995 (0.032) & 0.921 (0.136) & 0.799 (0.147) & 0.715 (0.159) \\ 
& 1000 & 0.975 (0.085) & 0.691 (0.057) & 0.665 (0.037) & 0.574 (0.054) \\ 
& 2500 & 0.967 (0.098) & 0.679 (0.012) & 0.639 (0.027) & 0.512 (0.056) \\ \bottomrule 
\end{tabular}
\caption{Averages and standard deviations (in parentheses) for the number of clusters and adjusted Rand index with $\bs s^0$ on 100 replications from the skew-symmetric mixture.}
\label{table:skewsymmtab}
\end{table}

\begin{figure}[ht!]
    \centering
    \includegraphics[scale=0.31]{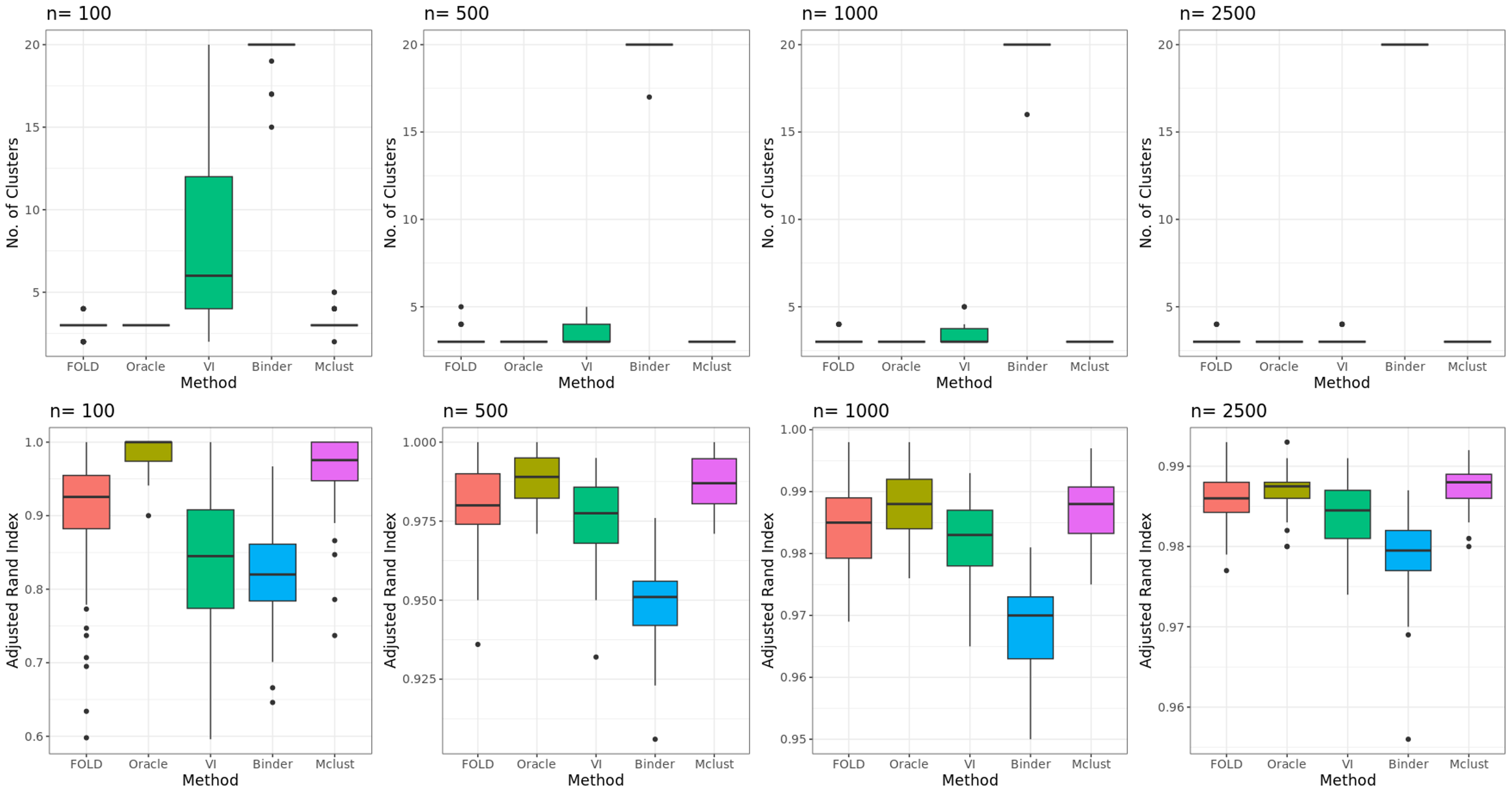}
    \caption{Comparison of the number of clusters and the adjusted Rand index with $\bs s^0$ on $100$ replications from a mixture of well-separated multivariate Gaussian kernels.}
    \label{fig:symmbox}
\end{figure}

\begin{figure}[ht!]
    \centering
    \includegraphics[scale=0.31]{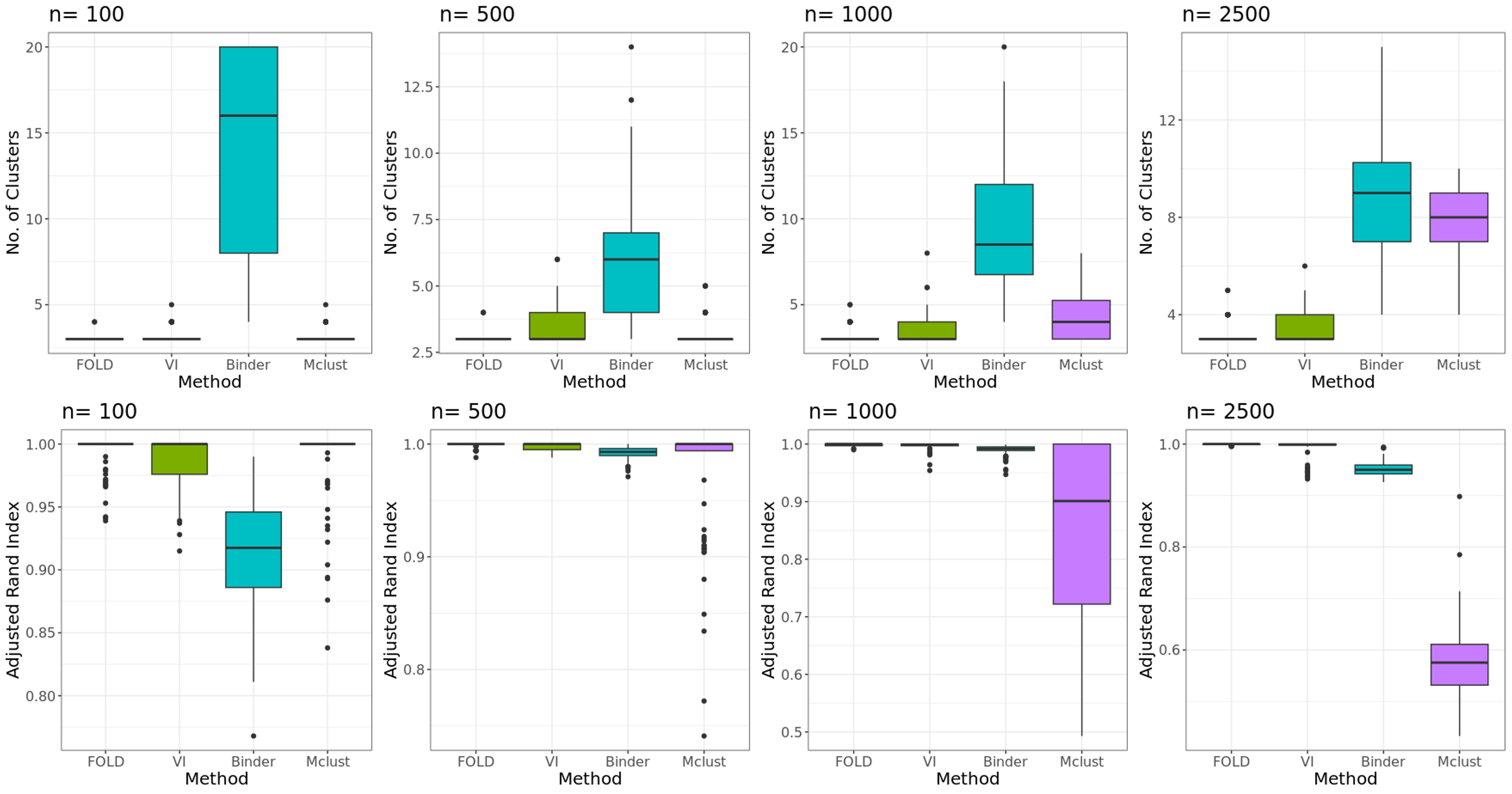}
    \caption{Comparison of the number of clusters and the adjusted Rand index with $\bs s^0$ on $100$ replications from a mixture of multivariate skew Gaussian distributions.}
    \label{fig:skewnorm}
\end{figure}

\begin{figure}[ht!]
    \centering
    \includegraphics[scale=0.31]{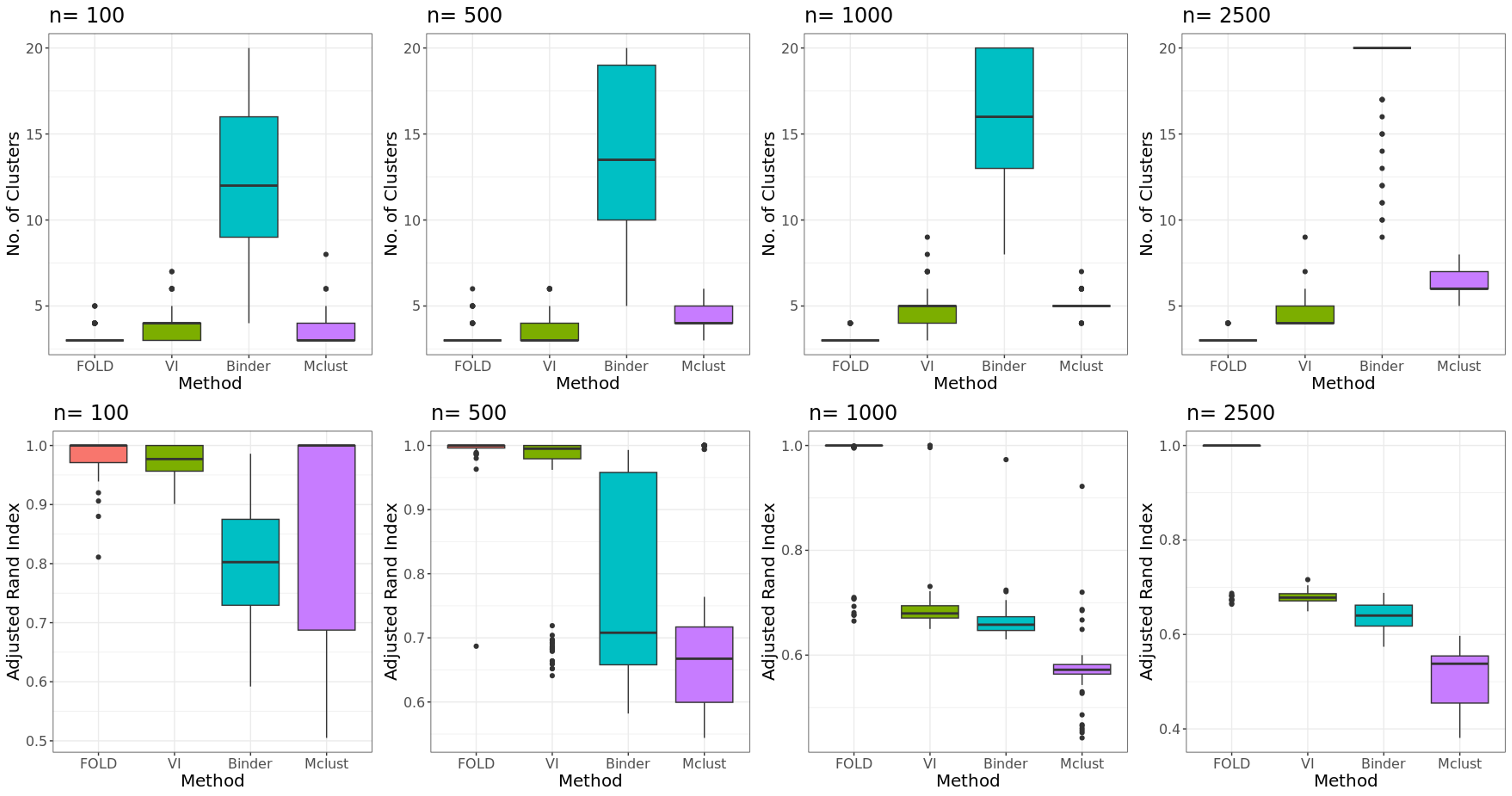}
    \caption{Comparison of the number of clusters and the adjusted Rand index with $\bs s^0$ on $100$ replications from the skew-symmetric mixture.}
    \label{fig:skewsymmbox}
\end{figure}

In the first scenario,
$f^0(x) = \sum_{m=1}^{3}a_m^0 \N_2(x; \mu_{m}^0, \Sigma_{m}^0)$, with weights $\bs a^0=(0.45,0.25,0.3)$ and atoms $\mu_{1}^0 = (6.5,5)$, $\Sigma_{1}^0 = I$, $\mu_{2}^0 = (0,0)$, $\Sigma_{2}^0 = \tx{diag}(5,2)$, and $\mu_{3}^0 = (-5,-5)$, $\Sigma_{3}^0 = \tx{diag}(3,1)$. As shown in Figure \ref{fig:contours}(a), the three kernels are well separated. The averages and standard deviations for the number of clusters and adjusted Rand index are given in Table \ref{table:symmtab}, with boxplots of these benchmarks displayed in Figure \ref{fig:symmbox}. FOLD, oracle FOLD, and \texttt{mclust} achieve high adjusted Rand index with increasing sample size. Oracle FOLD is perfect at computing the true number of clusters, and  \texttt{mclust} and FOLD are also excellent. Furthermore, the number of clusters and adjusted Rand index for $\cFOLD$ gradually approach the values produced by $\cFOLD^{*}$. The VI loss improves in the number of clusters as $n$ increases, but Binder's loss falls short, consistently producing $20$ clusters with diminishing variation across replications. VI and Binder's loss produce high adjusted Rand indices that improve with $n$.

Next, we let  $f^0(x) = \sum_{m=1}^3 a_m^0 \mathcal{SN}_2(x; \mu_{m}^0, \Sigma_{m}^0, \psi_{m}^0)$,
where $\mathcal{SN}_2(x; \cdot, \cdot, \cdot)$ denotes the PDF of a bivariate skew Gaussian distribution \citep{azzalini1996multivariate, azzalini1999statistical} and $\bs a^0 =(0.45,0.25,0.3).$ These kernels have the same location and scale parameters as the bivariate Gaussian mixture in simulation case 1, and skewness parameters $\psi_{1}^0 = (1,1)$, $\psi_{2}^0 = (-10, 15)$, and $\psi_{3}^0 = (4, -17)$. Figure \ref{fig:contours}(b) shows a contour plot of $f^0$. Results across the replications are reported in Figure \ref{fig:skewnorm} and Table \ref{table:skewtab}. For all values of $n$, FOLD attains high accuracy with a small number of clusters. As $n$ increases, all four clustering methods achieve high levels of the adjusted Rand index, though FOLD results in values close to $1$ while the VI loss and Binder's loss decline slightly. FOLD tends to produce between $3$ and $4$ clusters for each $n$ and reports less than or equal to the number of clusters induced by the VI loss in $94.5\%$ of all instances.

In the final simulation case, we take $f^0$ to be a bivariate mixture with weights $\bs a^0 = (0.55, 0.3, 0.15)$, and kernels
\begin{gather*}
    g_{1}^0 = 0.364 \cdot \mathcal{SN}_p\left( \begin{pmatrix}
        2.50 \\
        3.50
    \end{pmatrix},
    I,
    \begin{pmatrix}
        -10 \\
        15
        \end{pmatrix}
    \right)
    \\
     +
        0.212 \cdot \mathcal{N}_p\left( \begin{pmatrix}
        2.325 \\
        4.381
    \end{pmatrix},
    \tx{diag}(0.20, 0.80) \right)  \\ +
   0.424 \cdot \N_p \left(
    \begin{pmatrix}
        1.085 \\
        2.009
        \end{pmatrix},
    \tx{diag}(0.70, 0.60) \right);
    \\ g_{2}^0 = \mathcal{SN}_p\left( \begin{pmatrix}
        0 \\
        -3.50
    \end{pmatrix},
    \tx{diag}(5,2),
    \begin{pmatrix}
        4 \\
        -17
        \end{pmatrix}
    \right); \\
    g_{3}^0 = \mathcal{N}_p\left( \begin{pmatrix}
        -4 \\
        -2.50
    \end{pmatrix},
    \begin{pmatrix}
        0.50 & 0.50 \\
        0.50 & 2.50
    \end{pmatrix} \right).
\end{gather*}
$g_1^0$ is a mixture of one skew Gaussian kernel and two Gaussian kernels, $g_2^0$ is a skew Gaussian kernel, and $g_3^0$ is a Gaussian kernel. We refer to $f^0$ in this case as a skew-symmetric mixture. As can be seen in Figure \ref{fig:contours}(c), $g_1^0$ is a non-Gaussian, multimodal density. Similarly, data generated from $g_2^0$, the oblong kernel in the lower third of the sample space, will most likely be approximated by multiple components, despite only constituting one cluster. The results of the simulations are summarized in Table \ref{table:skewsymmtab} and Figure \ref{fig:skewsymmbox}. FOLD generally allocates observations to $3$ clusters, but the VI loss favors between $3$ and $6$ clusters. In $96.5\%$ of the replicates, the number of clusters for FOLD is less than or equal to the number for VI. As in simulation case 2, Binder's loss and \texttt{mclust} frequently return a larger number of clusters than the truth for larger $n$. The adjusted Rand index between $\cFOLD$ and $\bs s^0$ is close to $1$ across all sample sizes. In contrast, the three other methods achieve high adjusted Rand index for small $n$, but these values sharply drop for $n\geq 1000$. This is usually the result of splitting $g_1^0$ or $g_2^0$ into multiple clusters. 

\section{Illustration of Theorem \ref{thm:consistency}: Location GMMs} \label{section:illustration}
We now empirically validate the consistency result in Theorem \ref{thm:consistency} with a simple example. We set $M^0 = M^* = 4$ and $f^0 = f^* = \sum_{m=1}^4 (1/4)\N_2(\mu_m^0, (1/4)I)$, where $\mu_1^0 = (1,1)$, $\mu_2^0 = (1.75,1.75)$, $\mu_3^0 = (-1.75, -1.75)$, and $\mu_4^0 = (-1,-1)$, then fit the following location Gaussian model with $K=4$ components,
\begin{align} \label{eq:illustration-fitted-model}
    X_i & \sim \N_2(\theta_i, (1/4)I); 
    & \tilde \theta_k & \sim \N_2(0, I); 
    & \bs a & \sim \tx{Dir}(1/4, \dots, 1/4).
\end{align}
We vary $n \in \lb 50, 100, 500, 1000 \rb$, then fit a Gibbs sampler to extract samples from $\pi(\bs \theta, \bs s \mid \X)$. We implement FOLD, oracle FOLD, Binder's loss, and Binder's oracle loss using SALSO to minimize the risk functions, and we set $d(\cdot, \cdot)$ to be the Hellinger distance. For FOLD, we set $\omega = \omega^{\tx{AVG}}$ and for the FOLD oracle, we take $\omega = \omega^{*\tx{AVG}}$, which is constructed by setting $\gamma = \bar{\Delta}^*$ in the oracle risk. For Binder's loss, we take $\omega=1$, matching the choice used in packages such as \texttt{mcclust}. Figure \ref{fig:risk-convergence} shows $\mathcal R(\chat)$ and $\mathcal R^*(\chat)$ for increasing $n$, where $\chat$ are taken to be clusterings generated by average linkage hierarchical clustering on $\Delta$. We can see that as $n$ grows, $\mathcal R(\chat) \to \mathcal R^*(\chat)$. In all plots, the minimum risk over the candidates is attained at $2$ clusters, though for larger $n$ the risk of $3$ clusters becomes similar to that of $2$ clusters. In addition, the point estimates $\cFOLD$ and $\cFOLD^*$ are very similar as well even for smaller sample size, with adjusted Rand indicies between the two partitions being $0.98$ for $n=50$, then $1.0$ for $n=100,500,1000$. At $n=50$, $\cFOLD$ has $3$ clusters, but for $n>50$ both $\cFOLD$ and $\cFOLD^*$ have exactly $2$ clusters. Figure \ref{fig:large-sample-clusterings} shows point estimates for $\cFOLD$ and $\cFOLD^*$ along with $\bs c_{\tx{B}}$ and $\bs c_{\tx{B}}^*$ at the largest sample size. Binder's loss and its oracle prefer clusterings that approximately correspond to the original $4$ Gaussian components, while FOLD and its oracle prefer to merge these. Additionally, several small clusters are present in Binder's clusterings which seem to exist in the tails of the Gaussian components, leading to $5$ clusters in $\bs c_{\tx{B}}$ and $7$ clusters in $\bs c_{\tx{B}}^*$. 

\begin{figure}
    \centering
    \includegraphics[scale=0.57]{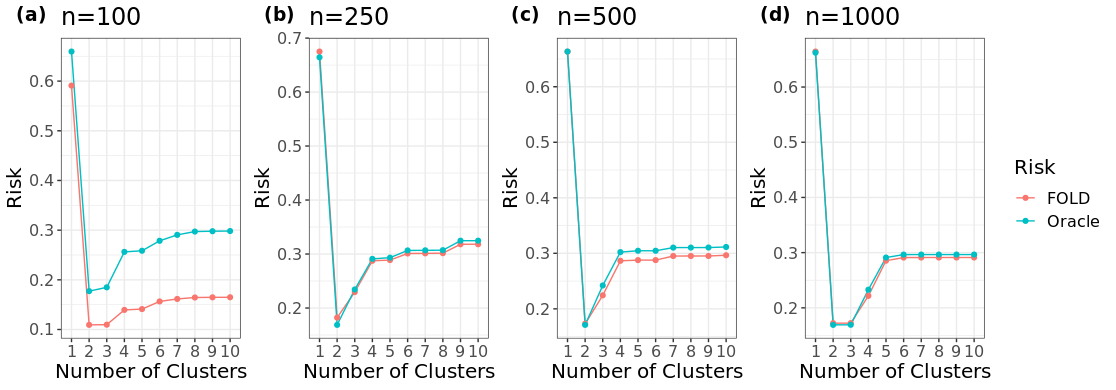}
    \caption{$\mathcal{R}(\chat)$ in red and $\mathcal{R}^*(\chat)$ in green over a set of partitions produced by average linkage hierarchical clustering on $\Delta$ and ordered by the number of clusters for varying $n \in \lb100, 250, 500, 1000 \rb$. As $n$ grows, the two risk functions and their minimizers begin to closesly resemble each other.}
    \label{fig:risk-convergence}
\end{figure}

\begin{figure}
    \centering
    \includegraphics[scale=0.65]{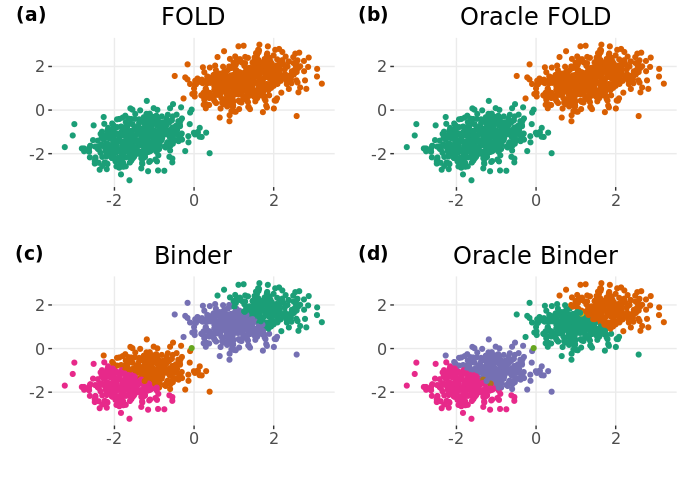}
    \caption{$\cFOLD$, $\cFOLD^*$, $\bs c_{\tx{B}}$, and $\bs c_{\tx{B}}^*$ for $n=1000$ simulated data points from a mixture of $4$ Gaussian kernels. The number of clusters for both $\cFOLD$ and $\cFOLD^*$ is $2$, whereas for $\bs c_{\tx{B}}$ there are $5$ clusters and for $\bs c_{\tx{B}}^*$ there are $7$ clusters.}
    \label{fig:large-sample-clusterings}
\end{figure}

To better understand the results in Figure \ref{fig:large-sample-clusterings}, note that the Hellinger distance between the kernels in $f^*$ is directly related to the Euclidean distance between their centers,
\begin{gather*}
    1- d^2 \lb \N_2(\cdot; \theta_m^0, 1/4), \N_2(\cdot; \theta_{m^\prime}^0,1/4) \rb \propto \exp \lb - \frac{1}{2} \norm{\theta_m^0 - \theta_{m^\prime}^0}^2 \rb.
\end{gather*}
We have deliberately chosen the true component means so that $\theta_1^0$ and $\theta_2^0$ are close, and $\theta_3^0$ and $\theta_4^0$ are close. This means that, despite the fact that we originally sampled from $4$ mixture components, the scatterplot of the data in Figure \ref{fig:large-sample-clusterings} resembles $2$ ovular clusters rather than $4$ spherical ones. $\cFOLD^*$ creates ovular clusters because it is less reliant on the fitted model in \eqref{eq:illustration-fitted-model} than Binder's loss. For example, if $X_i^0$ is in a high density region of $\N_2(\cdot; \theta_1^0,1/4)$ and $X_j^0$ is in a high density region of $\N_2(\cdot; \theta_2^0,1/4)$, then the Hellinger distance terms in \eqref{eq:rvns-delta-integral} shrink the value of $\Delta_{ij}^*$ close to $0$, ultimately promoting co-clustering of objects $i$ and $j$. In contrast, Binder's oracle is more reliant on assigning the observations to mixture components and would most likely place $X_i^0$ and $X_j^0$ in different clusters. 

\section{Application to Benchmark Datasets}

\begin{figure}[ht]
    \centering
    \includegraphics[scale=0.25]{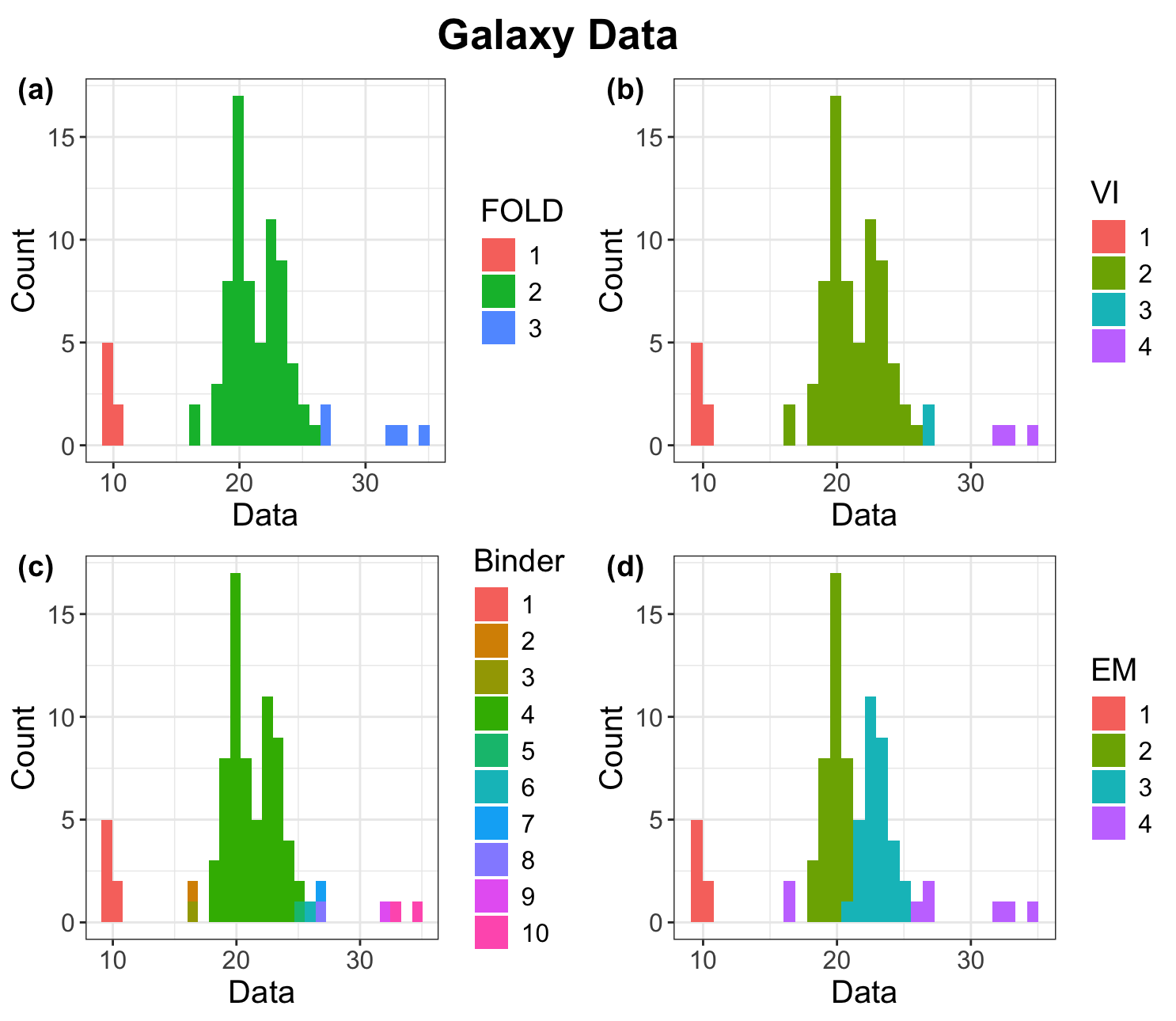}
    \caption{Model-based clusterings of the galaxy dataset.}
    \label{fig:galaxy-clusterings}
\end{figure}

\begin{figure}[ht]
    \centering
    \includegraphics[scale=0.25]{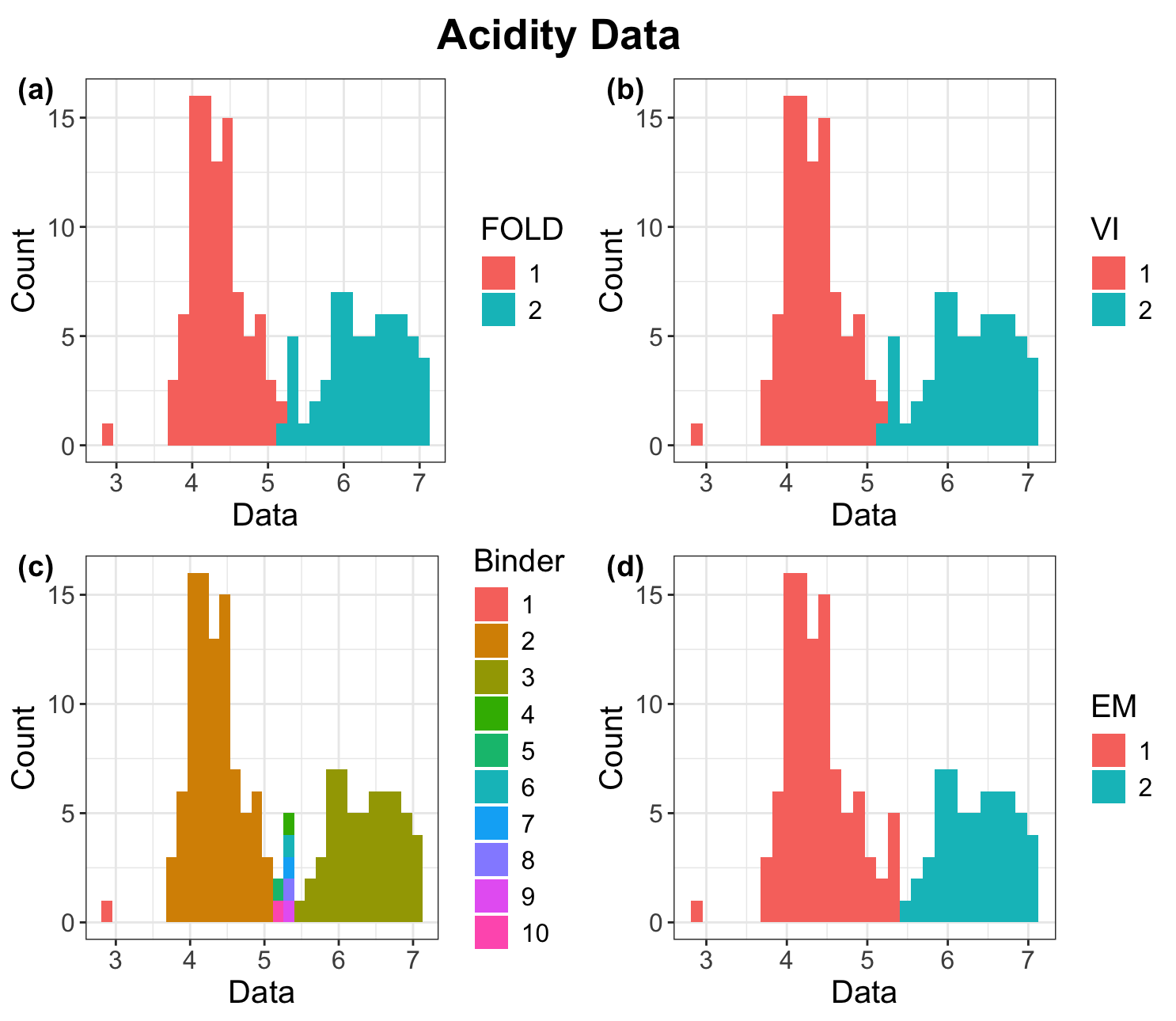}
    \caption{Model-based clusterings of the acidity dataset.}
    \label{fig:acidity-clusterings}
\end{figure}

\begin{figure}[ht]
    \centering
    \includegraphics[scale=0.25]{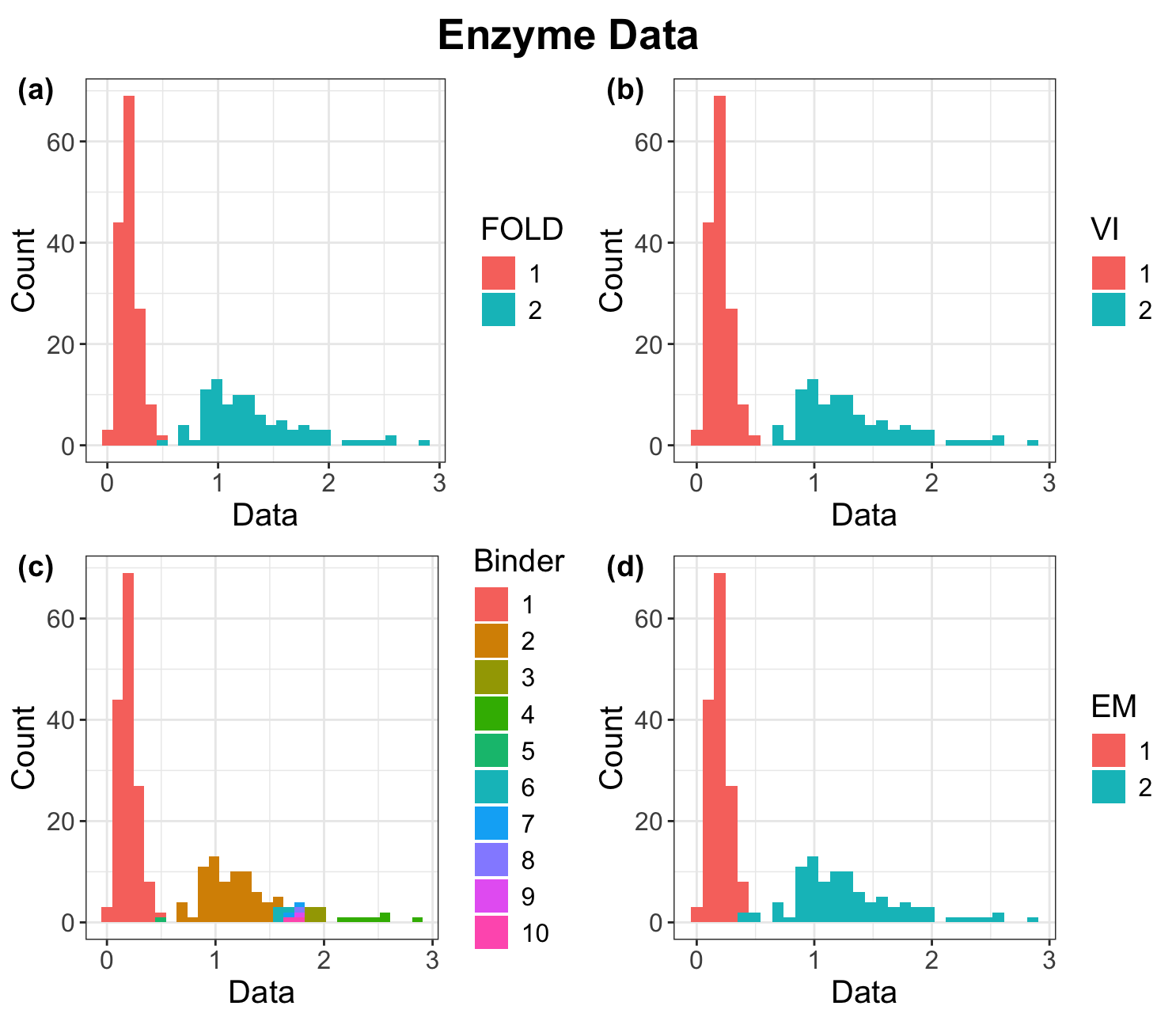}
    \caption{Model-based clusterings of the enzyme dataset.}
    \label{fig:enzyme-clusterings}
\end{figure}

\begin{table}[ht]
\footnotesize
\centering
\begin{tabular}{rrrrrrrrrr}
  \toprule
 Dataset & Metric & FOLD & VI & Binder's & EM & K-Means & HCA & DBSCAN & Spectral \\ 
  \midrule
  \multirow{2}{*}{Iris} & $\hat K$ & 2 &   2 &   2 &   2 &   2 &   2 &   3 &   2 \\ 
  & ARI & 1.000 & 1.000 & 1.000 & 1.000 & 0.920 & 1.000 & 0.859 & 1.000 \\  
  \midrule
   \multirow{2}{*}{Flea-Beetles} & $\hat K$ & 3 &  10 &  10 &   3 &   3 &   3 &   5 &   5 \\ 
  & ARI & 0.917 & 0.807 & 0.807 & 0.958 & 0.371 & 0.635 & 0.069 & 0.601 \\
  \midrule
   \multirow{2}{*}{Wine} &$\hat K$  & 3 &   5 &  10 &   3 &   3 &   3 &   3 &   4 \\ 
  & ARI & 0.912 & 0.893 & 0.806 & 0.829 & 0.895 & 0.754 & 0.439 & 0.922 \\  
   \bottomrule
\end{tabular}
\caption{Comparison of FOLD to competitors on the iris, flea-beetles, and wine datasets based on the number of clusters and ARI.}
\label{table:benchmarks}
\end{table}

In addition to the cell line application, we apply FOLD and other clustering methods to multiple benchmark datasets. First, we focus on three univariate examples that have been routinely used to evaluate model-based clustering procedures: the galaxy \citep{roeder1990density}, acidity \citep{crawford1992modeling,crawford1994application}, and enzyme \citep{bechtel1993population} datasets. For all examples, we fit FOLD, VI, Binder's loss, and EM, then compute the clustering point estimates. We center and scale the datasets to have mean $0$ and unit standard deviation. The three Bayesian methods are implemented with the model $X_i \sim \N(\mu_i, \sigma_i^2)$, $K=10$, $(\tilde \mu_k, \tilde \sigma_k^2) \sim \mathcal{NIG}(0,1,1,1)$, and $\bs a \sim \tx{Dir}(1/2, \dots, 1/2)$. The loss parameter for FOLD is selected with an elbow plot and the number of clusters in the EM algorithm is chosen by minimizing BIC. 

Figure \ref{fig:galaxy-clusterings} displays the point estimates for the galaxy data. FOLD selects $3$ clusters, whereas VI prefers $4$, Binder's loss prefers $10$, and the EM algorithm prefers $4$. FOLD favors coarse, non-Gaussian shaped clusters, while VI, Binder's loss, and the EM algorithm split the FOLD clusters into multiple sub-clusters. In particular, the EM algorithm divides the multi-modal group of objects in the center of the sample space into two Gaussian clusters. Similarly, VI and Binder's loss tend to split the clusters in the left and right tails of the data. FOLD, VI, and the EM algorithm perform similarly when applied to the acidity dataset (Figure \ref{fig:acidity-clusterings}), though there is some disagreement between the methods on how to allocate objects that are between the two clusters. Binder's loss is also sensitive to this ambiguity as it splits the data into $10$ clusters. Finally, FOLD, VI, and EM result in similar point estimates for the enzyme data (Figure \ref{fig:enzyme-clusterings}), with some disagreement on objects near the boundary of each cluster, whereas Binder's loss splits the non-Gaussian cluster on the right hand side of the dataset into several sub-clusters.

Next, we compare FOLD to both model-based and algorithmic clustering methods on the iris, flea-beetles \citep{lubischew1962use}, and wine \citep{winedataset} datasets. In comparison to the galaxy, acidity, and enzyme data, these examples have ground-truth class labels that correspond to separated clusters in the sample space. Hence, we can evaluate the accuracy of FOLD and other approaches on capturing a well-defined, true grouping of the observations. First, we center and scale the datasets. For the iris dataset, we merge the \textit{Iris virginica} and \textit{Iris versicolor} species into one true cluster, as there is no separation between these species, meaning that $M^0=2$. The flea-beetles dataset comprises three species (\textit{concinna}, \textit{heikertingeri}, and \textit{heptapotamica}), meaning $M^0=3$. The wines in the wine dataset are derived from three cultivars, hence $M^0=3$. Finally, for the wine example, we cluster the projection of the original dataset onto the first two principal components. The Bayesian methods fit the mixture $X_i \sim \N_2(\mu_i, \Sigma_i)$, $K=50$, $(\tilde \mu_k, \tilde \Sigma_k) \sim \mathcal{NIW}(0_2, 1, p+2, I)$, and $\bs a \sim \tx{Dir}(1/2, \dots, 1/2).$ All hyperparameters are chosen in the same manner as the application to the cell line dataset (see Section \ref{subsect:cells-hyperpars}), with the exception that we set \texttt{minPts}=$5$ when implementing DBSCAN.

The number of clusters $(\hat K)$ and ARI with the true labels are displayed in Table \ref{table:benchmarks}. In the iris dataset, all methods except K-means and DBSCAN result in the true partition of the flowers. FOLD and the EM algorithm excel when applied to the flea-beetles data, with both resulting in the true number of species and attaining high ARIs. Finally, FOLD performs particularly well on the wine dataset, scoring the second highest ARI (0.912), while the method with higher ARI (spectral clustering) overestimates the number of cultivars. In both the flea-beetles and wine examples, the FOLD point estimate has fewer clusters than the point estimates of both VI and Binder's, which are computed using the same MCMC samples as FOLD.

\section{Basic Properties of FOLD}

\subsection{Proof of Proposition \ref{prop:Lg}} \label{proof:prop1}
\begin{proof}
    We proceed as in \cite{dahl2022search} and decompose the loss function into separate terms that relate to the contingency table between $\widehat{\bs c}$ and $\bs s$. For any $k,k^\prime \in [K]$, let $\eta_{kk^\prime} = d(g(\cdot; \tilde \theta_k), g(\cdot; \tilde \theta_{k^\prime}))$. For index sets $\hat{C}_{h}$ and $S_k$, set $n_{h \cdot} = |\hat{C}_h|$, $n_{\cdot k} = |S_k|$, and $n_{hk} = |\hat{C}_h \cap S_k|$.
\begin{gather*}
    \Lc =  \sum_{i<j} \lb \textbf{1}_{\hat{c}_i = \hat{c}_j} \mathcal{D}_{ij} + \omega \textbf{1}_{\hat{c}_i \neq \hat{c}_j} \left(  1 - \mathcal{D}_{ij} \right) \rb \\
    = \sum_{i<j} \lb \textbf{1}_{\hat{c}_i = \hat{c}_j} \left( 1 - \left( 1 - \mathcal{D}_{ij}\right) \right) + \omega \left( 1 - \textbf{1}_{\hat{c}_i = \hat{c}_j} \right) \left(  1 - \mathcal{D}_{ij} \right)  \rb \\
    = \sum_{i < j} \textbf{1}( \hat{c}_i = \hat{c}_j ) + \omega \sum_{i<j} \left( 1 - \mathcal{D}_{ij} \right) - (1 + \omega) \sum_{i < j} \textbf{1}_{\hat{c}_i = \hat{c}_j} \left( 1 - \mathcal{D}_{ij} \right) \\
    = \sum_{h=1}^{\hat{K}_n} {n_{h \cdot} \choose 2} + \omega \sum_{k=1}^K {n_{\cdot k} \choose 2} + \omega \sum_{k<k^\prime} n_{\cdot k} n_{\cdot k^\prime} \left( 1 - \eta_{kk^\prime}\right) \\
    - (1 + \omega) \sum_{h=1}^{\hat{K}_n} \sum_{k=1}^K {n_{hk} \choose 2} - (1 + \omega) \sum_{h=1}^{\hat{K}_n} \sum_{k<k^\prime} n_{hk} n_{hk^\prime} \left( 1 - \eta_{kk^\prime} \right).
\end{gather*} 
Consequently, this shows that the loss is related to the Binder's loss with unit costs $a = 1$ and $b = \omega$,
\begin{gather*}
    \Lc = \mathcal{L}_{\tx{B}}(\widehat{\bs c}, \bs s) 
   + \mathcal{B}(\widehat{\bs c}, \Tilde{\bs \theta})
\end{gather*}
where 
\begin{align*}
   \mathcal{B}(\widehat{\bs c}, \Tilde{\bs \theta}) & = \omega \sum_{k<k^\prime} n_{\cdot k} n_{\cdot k^\prime} \left( 1 - \eta_{kk^\prime}\right) \\
   & - (1 + \omega) \sum_{h=1}^{\hat{K}_n} \sum_{k<k^\prime} n_{hk} n_{hk^\prime} \left( 1 - \eta_{kk^\prime} \right).
\end{align*}
It is clear that 
\begin{equation*}
    \mathcal{B}(\bs s, \Tilde{\bs \theta}) = \omega \sum_{k<k^\prime} n_{\cdot k} n_{\cdot k^\prime} \left( 1 - \eta_{kk^\prime}\right)
\end{equation*}
which is equal to zero if and only if $\eta_{kk^\prime} = 1$ for all $k \neq k^\prime$ pairs. 
\end{proof}

\subsection{Proof of Proposition \ref{prop:Delta}}
\begin{proof}
We proceed by showing that the bounded metric properties of the Hellinger distance are preserved by taking the posterior expectation.
    \begin{enumerate}
    \item For each $i \leq j$, 
    $0 \leq d(g(\cdot ; \theta_i), g(\cdot ; \theta_j)) \leq 1$
    almost surely, and so 
    %\begin{gather*}
        $0 \leq \Delta_{ij} \leq 1.$
    %\end{gather*}
    \item We next verify the metric axioms on the entries of $\Delta$. 
    \begin{enumerate}
        \item For each $i=1, \dots, n$,
    %\begin{gather*}
        $d(g(\cdot ;  \theta_i), g(\cdot ;  \theta_i)) = 0$
    %\end{gather*}
    almost surely, and so $\Delta_{ii}=0$.
    \item $\Delta_{ij} = \Delta_{ji}$ since the Hellinger distance is symmetric and $(\theta_{1}, \dots, \theta_{n})$ are exchangeable conditional on $\X$.
    \item $\Delta_{il} \leq \Delta_{ij} + \Delta_{jl}$ since the triangle inequality in the Hellinger distance holds almost surely for $\theta_i, \theta_j, \theta_{l}$ conditional on $\X$.
    \end{enumerate}
    \item Observe that
    \begin{gather*}
        \Delta_{ij} = \E_\Pi[\E_\Pi[d(g(\cdot ; \theta_i), g(\cdot ; \theta_j)) \mid \lambda, \X] \mid \X] \\
        = \int_{\Lambda} \sum_{k<k^\prime}\eta_{k k^\prime} q_{ij}^{kk^\prime}  \:  \Pi(d \lambda \mid \X),
        \end{gather*}
    where
    \begin{gather*}
        q^{kk^\prime}_{ij} = \Pi(\lb s_i = k, s_j = k^\prime \rb \cup \lb s_i = k^\prime, s_j = k \rb \mid \lambda, \X), \\
        \eta_{kk^\prime} = d(g(\cdot; \tilde \theta_k), g(\cdot; \tilde \theta_{k^\prime})),
    \end{gather*}
    for all $1 \leq k < k^\prime \leq K$. Since $\eta_{kk^\prime} \leq 1$, 
    \begin{gather*}
        \Delta_{ij} \leq \int_{\Lambda} \sum_{k<k^\prime} q_{ij}^{kk^\prime} \Pi(d\lambda \mid \X) = \Pi(s_i \neq s_j \mid \X).
    \end{gather*}
\end{enumerate}
\end{proof}

\subsection{Proof of Proposition \ref{prop:elbow}}
\begin{proof}
    To see why there is a positive relationship between $\omega$ and $r_\omega$, recall that the candidate clusterings over which we optimize $\mathcal R(\chat)$ are nested. Given a value of $\omega$ and minimizer $\bs c_\omega^*$, suppose we now increment the tuning parameter from $\omega$ to $\omega + \epsilon$ for some $\epsilon>0$. Denote $\bs c^\prime$ to be the child clustering of $\bs c^*_\omega$ (i.e., it is next in the agglomerative hierarchy), and assume that it results from merging, say, $C^*_{\omega 1}$ and $C^*_{\omega 2}$ in $\bs c_\omega^*$. Then $\mathcal{R}(\bs c^\prime) < \mathcal{R}(\bs c^*_\omega)$ if $\sum_{i \in C^*_{\omega 1}, j \in C^*_{\omega 2}} \Delta_{ij} < (\omega + \epsilon)/(1+\omega + \epsilon)$. For large enough $\epsilon$, this will hold, and so $\bs c_{\omega+\epsilon}^*$ will be equal to $\bs c^\prime$ (or one of its child clusterings in the hierarchy). Suppose then that $\epsilon$ is such that $\bs c^*_{\omega + \epsilon} = \bs c^\prime$. Then clearly, $r_{\omega + \epsilon} > r_{\omega}$ because any pairing $i \in C^*_{\omega 1}$ and $j \in C^*_{\omega 2}$ will now contribute a non-negative term to $r_{\omega + \epsilon}$. Put differently, we have that $r_{\omega + \epsilon} = r_{\omega} + (\sum_{i \in C^*_{\omega 1}, j \in C^*_{\omega 2}} \Delta_{ij})/(\sum_{i<j} \Delta_{ij})$.
\end{proof}

\subsection{Proof of Proposition \ref{prop:oracles}}
\begin{proof}
\begin{enumerate}[label=(\alph*)]
    \item This oracle rule follows by noting that if $d\lb g(\cdot; \theta_m^*), g(\cdot; \theta_{m^\prime}^*) \rb = 1$ for all component pairs, then all components have disjoint support. Hence, $\mathcal R^*(\chat)$ is just Binder's loss of $\chat$ to $\bs s^*$, where $s_i^* = m$ if and only $X_i^0 \in \tx{supp}(g(\cdot; \theta_m^*))$, which is uniquely minimized at $\chat = \bs s^*$. 
    \item We can verify this by noting that, under the described conditions, $\bs c_0$ minimizes $\mathcal R^*(\chat)$ if 
    \begin{equation*}
        d\lb g(\cdot; \theta_1^*), g(\cdot; \theta_{2}^*) \rb \sum_{i<j} q_{ij}^{12*} \leq \mathcal R^*(\chat).
    \end{equation*}
    A sufficient condition for this to hold is when $d\lb g(\cdot; \theta_1^*), g(\cdot; \theta_{2}^*) \rb < \gamma/q_{ij}^{12*}$ for all $i,j$ pairs. Since $0 \leq q_{ij}^{12*} \leq 1$, this criteria can be satisfied when $ d\lb g(\cdot; \theta_1^*), g(\cdot; \theta_{2}^*) \rb < \gamma$. 
\end{enumerate}

\end{proof}

\section{Preliminaries for Asymptotic Analysis}
Recall the assumptions for Theorem \ref{thm:consistency} that we make in the main article.
\begin{assumption}
    Suppose that the following conditions hold.
    \begin{itemize}
        \item[] (A1) $\Lambda$ consists of measures with uniformly finite support, i.e. there exists an $L>0$ so that $|\tx{supp}(\lambda)| \leq L$ for all $\lambda \in \Lambda$.
        \item[] (A2) The KL minimizer $\lambda^* = \sum_{m=1}^{M^*} a_m^* \delta_{\theta_m^*}$ exists and is unique.
        \item[] (A3) $f^*$ and $f^0$ are such that $\tx{supp}(f^0) \subseteq \tx{supp}(f^*)$.
        \item[] (A4): $\mathcal{G} = \lb \Tilde{g}(x-\theta) : \theta \in \Theta \rb$ for some bounded probability density function $\Tilde{g}(\cdot)$, there exists a $\zeta>0$ so that $\tilde{g}(z)$ is $\zeta$-H{\"o}lder continuous, and for any fixed $\theta^\prime$, $D_g(\theta, \theta^\prime) = d \lb g(\cdot; \theta), g(\cdot;\theta^\prime) \rb$ is continuous in $\theta$.
        \item[] (A5) The mixing measure of the fitted model contracts to the oracle in the Wasserstein 2-distance, i.e. there exists a non-negative sequence $\epsilon_n$ so that $\epsilon_n \to 0$ and 
        \begin{equation}
            \rho_n(\X) = \Pi \lb \lambda \in \Lambda: W_2(\lambda, \lambda^*) \genq \epsilon_n \mid \X \rb \overset{\P^0}{\longrightarrow} 0.
        \end{equation}
    \end{itemize}
\end{assumption}
In this section, we will be primarily concerned with introducing concepts necessary for understanding the proof of Theorem \ref{thm:consistency}. We will also detail conditions under which (A2) and (A5) hold, the latter of which is a statement on the posterior contraction of the mixing measure. The main contribution that (A5) has to the proof of Theorem \ref{thm:consistency} is that it establishes consistency of the cluster-specific parameters $\bs a$ and $\Tilde{\bs \theta}$. Conditions under which (A5) hold have been extensively explored in \cite{nguyen2013convergence}, \cite{ho2016strong}, and \cite{guha2021posterior}.

\subsection{Additional Notation and Definitions}
Given $x = (x_1, \dots, x_p)^T \in \mathbb{R}^p$, let $\norm{x} = \sqrt{\sum_{s=1}^p x_{s}^2}$. For $q \geq 1$, the $q$-th order Wasserstein distance between the finite measures $\lambda = \sum_{k=1}^K a_k \delta_{\theta_k}$ and $\lambda^\prime = \sum_{k^\prime=1}^{K^\prime} a^\prime_{k^\prime} \delta_{\theta_{k^\prime}^\prime}$ is
\begin{equation*}
    W_q(\lambda, \Tilde{\lambda} ) = \inf_{\bs b \in \mathcal{C}(\bs a, \bs a^\prime)} \left( \sum_{k,k^\prime} b_{kk^\prime} \norm{\theta_k - \theta^\prime_{k^\prime}}^q  \right)^{1/q},
\end{equation*}
where $\mathcal{C}(\bs a, \bs a^\prime)$ is the set of all couplings with marginal probability vectors $\bs a$ and $\bs a^\prime$. That is, for all $k=1, \dots, K$, $\sum_{k^\prime=1}^{K^\prime}b_{kk^\prime} = a_k$ and for all $k^\prime = 1, \dots, K^\prime$, $\sum_{k=1}^K b_{kk^\prime} = a_{kk^\prime}$. The Wasserstein q-distance is a metric over the space of probability measures on $\Theta$, and therefore $W_q(\lambda, \lambda) = 0$, if $\lambda \neq \lambda^\prime$ then $W_q(\lambda, \lambda^\prime) > 0$, $W_q(\lambda, \lambda^\prime) = W_q(\lambda^\prime, \lambda)$, and $W_q$ obeys the triangle inequality. We will only consider mixtures having identifiable mixing measures. That is, if $\int g(x;\theta) \lambda(d \theta) = \int g(x; \theta) \lambda^\prime(d \theta)$ for all $x$, then $\lambda = \lambda^\prime$. This identifiability condition is satisfied by Gaussian mixtures, along with several exponential family mixtures \citep{barndorff1965identifiability} and location family mixtures \citep{teicher1961identifiability}. 

Recall that $\mathbb{P}^0$ refers to probability statements with respect to the true data generating process, with $\mathbb{P}^0(A) = \int_A f^0(x) dx$. For any sequences $y_n$ and $z_n$, $y_n \underset{\sim}{<} z_n$ if there exists a $N \in \mathbb{N}$ and constant $B>0$ so that $|y_n| \leq B |z_n|$ for all $n \geq N$; $y_n \asymp z_n$ if $y_n \underset{\sim}{<} z_n$ and $z_n \underset{\sim}{<} y_n$; and $y_n \sim z_n$ if $y_n/z_n \to 1$. A random variable $Z_n = o_{\mathbb{P}^0}(1)$ if $Z_n \to 0$ in $\mathbb{P}^0$-probability. Finally, we say that an event $A_n$ holds with high probability if $\P^0(A_n) \to 1$ as $n \to \infty$. 

\subsection{Uniqueness of $\lambda^*$}
While (A2) is trivial in the exact-fitted and well-specified regime, uniqueness of $\lambda^*$ may not hold in general. Section 4 of \cite{guha2021posterior} details conditions for uniqueness in the misspecified regime, which we briefly summarize. Let $\mathcal M^* = \{ \nu^* \in \Omega : \nu^* \in  \underset{\nu \in \Omega}{\tx{arg min }} \tx{KL}(f^0, f) \}$, where recall that $\Omega$ is the space of all probability measures on $\Theta$. Additionally, let $\mathcal R^* = \{ \lambda^* \in \Lambda : \lambda^* \in  \underset{\lambda \in \Lambda}{\tx{arg min }} \tx{KL} (f^0, f) \}$, the set of KL-minimizers constrained to be contained within the prior support. By Lemma 4.2 of \cite{guha2021posterior}, the minimizing density $f^{*}$ over $\Omega$ is unique. Hence, under identifiability of the mixing measure, we have that $\mathcal M^* = \lb \nu^* \rb$ for some unique $\nu^* \in \Omega$. Therefore, as long as the KL divergence is minimized over $\Lambda$, we have that $\mathcal M^* = \mathcal R^*$, meaning that $\lambda^* = \nu^*$ and so the oracle is finite and unique.

\subsection{Strong Identifiability and Lipschitz Conditions}
As stated in Section \ref{section:theory}, Wasserstein contraction of $\lambda$ will generally require that $\mathcal{G}$ is strongly identifiable, referring to a linear independence condition on the derivatives of $g(x;\theta)$. We formally define these concepts below. Recall that $\Theta \subset \mathbb{R}^p$
is assumed to be compact and $f$ to be identifiable in $\lambda$. In these definitions, $\nabla_g(\theta)$ and $\textbf{H}_g(\theta)$ denote the gradient and Hessian of $g(x;\theta)$ with respect to $\theta$. 
\begin{definition}
The family $\mathcal{G} = \lb g(\theta) : \theta \in \Theta \rb$ is \textit{first order identifiable} if $g(x;\theta)$ is differentiable in $\theta$ and, given $r$ atoms $\theta_1, \dots, \theta_r \in \Theta $, if there exists $\beta_1, \dots, \beta_r \in \mathbb{R}$ and $\rho_1, \dots, \rho_r \in \mathbb{R}^p$ so that for all $x$
\begin{equation*}
    \sum_{s=1}^r \lb \beta_s g(x;\theta_s) + \rho_s^T \nabla_g(\theta_s) \rb =0,
\end{equation*}
then $\beta_s = 0$  and $\rho_s = \mathbf{0}$ for all $s=1, \dots, r$. 
\end{definition}

\begin{definition}
The family $\mathcal{G} = \lb g(\theta) : \theta \in \Theta \rb$ is \textit{second order identifiable} if $g(x;\theta)$ is twice differentiable in $\theta$ and, given $r$ atoms $\theta_1, \dots, \theta_r \in \Theta $, if there exists $\beta_1, \dots, \beta_r \in \mathbb{R}$, $\rho_1, \dots, \rho_r \in \mathbb{R}^p$, and $\nu_1, \dots, \nu_r \in \mathbb{R}^p$ so that for all $x$
\begin{equation*}
    \sum_{s=1}^r \lb \beta_s g(x;\theta_s) + \rho_s^T \nabla_g(\theta_s) +  \nu_s^T \textbf{H}_g(\theta_s) \nu_s \rb =0,
\end{equation*}
then $\beta_s = 0$  and $\rho_s = \nu_s = \mathbf{0} $ for all $s=1, \dots, r$.
\end{definition}

It is clear from the two definitions that second-order identifiability implies first-order identifiability. These notions are commonly used in the literature because they bridge the connection between convergence from $f$ to $f^0$ (or $f^*$) in the Hellinger distance and convergence from $\lambda$ to $\lambda^0$ (or $\lambda^*$) in Wasserstein distance. This will be useful for showing concentration of the weights and atoms, as this is straightforward to extract from convergence of the mixing measure. A canonical example of a second-order identifable family is Gaussian location mixtures for a known covariance matrix $\Sigma$, i.e. $\mathcal{G} = \lb \N_p(\theta, \Sigma) : \theta \in \Theta \rb$, whereas location-scale Gaussian families $\mathcal G = \lb \N_p(\mu, \Sigma): \mu \in \R^p, \Sigma \in \mathbb{S}^p \rb$, where $\mathbb{S}^p$ is the set of $p \times p$ positive semi-definite matrices, are only first-order identifiable. 

In addition to strong identifiability, various Lipschitz conditions on the derivatives of $g(x; \theta)$ are required  for contraction of the mixing measure. The strongest condition is the second order integral Lipschitz property, as introduced in \cite{guha2021posterior}. 
\begin{definition}
    $\mathcal{G}$ satisfies the integral Lipschitz property of the second order with respect to the mixing measures $\lambda^0$ and $\lambda^*$ if $g(x;\theta)$ is twice differentiable in $\theta$ and for all $x \in \chi$,
    \begin{equation} \label{eq:lipschitz}
        |\zeta^T (\textbf{H}_g(\theta_1) - \textbf{H}_g(\theta_2)) \zeta| \leq M(x) \norm{\theta_1 - \theta_2}^\xi \norm{\zeta}^2
    \end{equation}
    for any $\zeta \in \mathbb{R}^p$ and for some $\xi>0$ that does not depend on $x$ and $\theta_1, \theta_2$, where $M(x)$ is a function of $x$ so that $\E_0[M(X)/f^*(X)] < \infty$.  
\end{definition}
Note that this is not only a condition on $\mathcal{G}$, but also on $\lambda^0$ and $f^0$. The above inequality is effectively a smoothness condition on the Hessian of $g(x;\theta)$, where the Lipschitz constant is a function of $x$. \cite{guha2021posterior} invokes the integral Lipschitz property when studying the misspecificed regime. Related to the integral Lipschitz property is the uniform Lipschitz property, which was introduced in \cite{ho2016strong}.
\begin{definition}
    $\mathcal{G}$ satisfies the uniform Lipschitz property of the second order if $g(x;\theta)$ is twice differentiable in $\theta$ and for all $x \in \chi$,
    \begin{equation*} \label{eq:ulipschitz}
        |\zeta^T (\textbf{H}_g(\theta_1) - \textbf{H}_g(\theta_2)) \zeta| \leq M \norm{\theta_1 - \theta_2}^\xi \norm{\zeta}^2
    \end{equation*}
    for any $\zeta \in \mathbb{R}^p$ and for some $\xi, M>0$ that do not depend on $x$ and $\theta_1, \theta_2$.
\end{definition}
The uniform Lipschitz property is generally invoked in well-specified regimes. As in (\ref{eq:lipschitz}), it is also a Lipschitz condition on the Hessian of $g(x;\theta)$. The second order Lipschitz property is satisfied by the location-scale Gaussian family and location-scale Student's t family \citep{ho2020robust}. A similar notion exists for the gradient, known as the first-order Lipschitz property. 

Lipschitz conditions on the derivatives, particularly on the Hessian, of a probability density function are frequently assumed in the literature on mixing measure concentration. See, for example, Lemma 2 in \cite{chen1995optimal} and Theorem 1 in \cite{nguyen2013convergence}. While it is difficult to characterize what choices of $f^0$ and $\G$ can admit the integral Lipschitz property in the misspecified regime, the property will hold if $\G$ is uniformly Lipschitz of the second order and $\int_{\mathbb{R}^p} \frac{f^0(x)}{f^*(x)} dx < \infty$. Hence, for uniformly Lipschitz families, such as location Gaussians, the integral Lipschitz property can be reformulated as a tail condition on $f^*$ and $f^0$.

\subsection{Properties of Wasserstein Concentration}
The proof of Theorem \ref{thm:consistency} relies on characterizing the mixing measures in the set
\begin{equation*}
    B_{\epsilon_n}(\lambda^*) = \lb \lambda \in \Lambda: W_2(\lambda, \lambda^*) \lenq \epsilon_n \rb,
\end{equation*}
or the $\epsilon_n$ ball around $\lambda^*$. (A1) becomes a crucial assumption for doing so, as it tells us that each $\lambda \in \Lambda$ has at most $L$ support points, where $L<\infty$. In the exact fitted case where $\Pi(K=M^*)=1$, \cite{nguyen2013convergence} shows that for all $m \in [M^*]$, there exists a unique $k \in [K]$ so that $\norm{\theta_k - \theta_m^*} \lenq \epsilon_n$ and $|a_k - a_m^*| \lenq \epsilon_n$. In this next proposition, we will show that a similar result holds for the overfitted case, i.e., when $\Pi(K > M^*) > 0$. 
\begin{proposition} \label{prop:rmengMS}
    Let $\lambda = \sum_{k=1}^K a_k \delta_{\theta_k} \in B_{\epsilon_n}(\lambda^*)$, where $M^* \leq K \leq L < \infty$ and $\epsilon_n \to 0$, and fix $0 < \delta < 1$. 
    \begin{enumerate}
        \item For large enough $n$, for all $m \in [M^*]$ there exists an index set $\Ilam_m \subset [K]$ so that
        \begin{equation} \label{eq:prop-rmengMS-consistency}
            \max_{k \in \Ilam_m} \norm{\theta_k - \theta_m^*} \lenq \epsilon_n,\qquad \bigg | \sum_{k \in \Ilam_m} a_k - a_m^* \bigg | \lenq \max(\epsilon_n, \epsilon_n^{2 \delta}),
        \end{equation}
        and $\Ilam_m \cap \Ilam_{m^\prime} = \emptyset$ for all $m \neq m^\prime = 1, \dots, M^*$. 
        \item Let $\Ilam = \bigcup_{m=1}^{M^*} \Ilam_m$. Then,
        \begin{equation} \label{eq:prop-rmengMS-overfit}
            \sum_{k \not \in \Ilam} a_k \lenq \epsilon_n^{2 \delta}. 
        \end{equation}
    \end{enumerate}
\end{proposition}
Proposition \ref{prop:rmengMS} gives us the necessary tools to establish convergence to the oracle. The statement in \eqref{eq:prop-rmengMS-consistency} tells us that for large enough $n$, we can allocate a subset of the components in $\lambda$ to $M^*$ groups, given by $\Ilam_m$ for $m \in [M^*]$. In the $m$th group, all of the atoms contract to $\theta_m^*$, and the total weight of the $m$th group contracts to $a_m^*$. \eqref{eq:prop-rmengMS-overfit} tells us that any components left over, i.e. any components with atoms that \textit{do not} contract to an atom in $\lambda^*$, will have total weight diminishing at an $\epsilon_n^{2 \delta}$ rate. This behavior is comparable to Theorem 1 in \cite{rousseau2011asymptotic}, which established that under certain settings of overfitted finite mixture models, the total weight of the overfitted components empty at an approximately $n^{-1/2}$ rate. Under the settings of Theorem 4.3 from \cite{guha2021posterior}, \eqref{eq:prop-rmengMS-overfit} implies that the redundant components empty at an approximately $\left( \log n / n \right)^{1/2} $ rate for misspecified models. We now present the proof of the proposition.

 \subsection{Proof of Proposition \ref{prop:rmengMS}}
 \begin{proof}
    Let $\lambda \in B_{\epsilon}(\lambda^*)$, where $K \geq M^*$. We now show concentration of the weights and atoms using techniques from \cite{nguyen2011wasserstein} and \cite{nguyen2013convergence}. Let $0 < \delta < 1$ be a small positive constant. For large enough $n$, $a_{\min}^* \genq \epsilon_n^{2 \delta}$, implying
    \begin{equation*}
        W_2^{2(1-\delta)}(\lambda,\lambda^*) \geq \min_{k} \norm{\theta_k - \theta_m^*}^2.
    \end{equation*}
    For any $m \in [M^*]$ let
    \begin{equation*}
        \Ilam_m = \lb k \in [K]: \norm{\theta_k - \theta_m^*}^2 \leq W_2^{2(1-\delta)}(\lambda,\lambda^*) \rb
    \end{equation*}
    and $\Ilam = \bigcup_{m=1}^{M^*} \Ilam_m$. It follows that
    \begin{gather}
        W_2^2(\lambda, \lambda^*) \geq \sum_{k \not \in \mathcal{I}^{(\lambda)}} a_k \min_{m} \norm{\theta_k - \theta_m^*}^2 > W_2^{2(1-\delta)}(\lambda, \lambda^*) \sum_{k \not \in \mathcal{I}^{(\lambda)}} a_k \nonumber \\
        \implies \sum_{k \not \in \mathcal{I}^{(\lambda)}} a_k < W_2^{2 \delta}(\lambda, \lambda^*) \lenq \epsilon_n^{2 \delta}. \label{eq:twodelta}
    \end{gather}

    Now, we will show concentration of the weights towards their oracle values. For fixed $m$,
    \begin{gather*}
        \epsilon_n^2 \genq W_2(\lambda, \lambda^*) \geq \sum_{k \in \mathcal{I}^{(\lambda)}\setminus\mathcal{I}_m^{(\lambda)}} b_{km} \norm{\theta_k - \theta_m^*}^2 \geq \min_{k \in \Ilam \setminus \Ilam_m} \norm{\theta_k - \theta_m^*}^2 \sum_{k \in \mathcal{I}^{(\lambda)}\setminus\mathcal{I}_m^{(\lambda)}} b_{km} \\
        = \min_{k \in \Ilam \setminus \Ilam_m} \norm{\theta_k - \theta_m^*}^2 \left( a_m^* - \sum_{k \in \Ilam_m} b_{km} - \sum_{k \not \in \Ilam} b_{km} \right) \\
        = \min_{k \in \Ilam \setminus \Ilam_h} \norm{\theta_k - \theta_m^*}^2 \left( a_m^* - \sum_{k \in \Ilam_m} a_k + \sum_{m \in \Ilam_h} \sum_{m^\prime \neq m} b_{km^\prime} - \sum_{k \not \in \Ilam } b_{km} \right) 
    \end{gather*}
    \begin{gather*}
        \geq \min_{k \in \Ilam \setminus \Ilam_m} \norm{\theta_k - \theta_m^*}^2 \left( a_m^* - \sum_{k \in \Ilam_m} a_k - \sum_{k \not \in \Ilam } b_{km} \right) \\
        \geq \min_{k \in \Ilam \setminus \Ilam_m} \norm{\theta_k - \theta_m^*}^2 \left( a_m^* - \sum_{k \in \Ilam_m} a_k -\sum_{k \in \Ilam \setminus \Ilam_m} b_{km} - \sum_{k \not \in \Ilam } b_{km} \right) \\
        = \min_{k \in \Ilam \setminus \Ilam_m} \norm{\theta_k - \theta_m^*}^2 \left( \sum_{k \in \Ilam_m} b_{km} - \sum_{k \in \Ilam_m} a_k  \right) \\
        = - \min_{k \in \Ilam \setminus \Ilam_m} \norm{\theta_k - \theta_m^*}^2 \sum_{k \in \Ilam_m} \sum_{m^\prime \neq m} b_{km^\prime}.
    \end{gather*}
    For any $m^\prime \neq m$, $k^\prime \in \Ilam_{m^\prime}$, and $k \in \Ilam_m$, 
    \begin{gather*}
        \min_{k \in \Ilam \setminus \Ilam_m} \norm{\theta_k - \theta_m^*} \leq \norm{\theta_{k^\prime} - \theta_m^*} \\
        \leq \norm{\theta_{k^\prime} - \theta_{m^\prime}^*} + \norm{\theta_k - \theta_{m^\prime}^*} + \norm{\theta_k - \theta_m^*} \lenq \epsilon_n + \norm{\theta_k - \theta_{m^\prime}^*}.
    \end{gather*}
    Since $\Theta$ is compact, this shows that
    \begin{gather*}
        - \min_{k \in \Ilam \setminus \Ilam_m} \norm{\theta_k - \theta_m^*}^2 \sum_{k \in \Ilam_m} \sum_{m^\prime \neq m} b_{km^\prime} \\
        \genq - \epsilon_n - \sum_{ k \in \Ilam_m} \sum_{m^\prime \neq m} b_{km^\prime} \norm{\theta_k - \theta_{m^\prime}^*}^2 \genq - \epsilon_n - W_2(\lambda, \lambda^*)^2. 
    \end{gather*}
    Therefore,
    \begin{equation*}
        \min_{k \in \Ilam \setminus \Ilam_m} \norm{\theta_k - \theta_m^*}^2 \bigg | a_m^* - \sum_{k \in \Ilam_m} a_k- \sum_{k \not \in \Ilam } b_{km} \bigg | \lenq \epsilon_n.
    \end{equation*}
    Observe that there exists $m^\prime \neq m$ so that
    \begin{equation*}
        \min_{k \in \Ilam \setminus \Ilam_m} \norm{\theta_k - \theta_m^*}^2 \genq \norm{\theta_m^* - \theta_{m^\prime}^*}^2 - \epsilon_n.
    \end{equation*}
    The above derivation implies that there exists constants $M, M^\prime,$ and $M^{\prime \prime} > 0$ so that
    \begin{gather*}
        \bigg | a_m^* - \sum_{k \in \Ilam_m} a_k - \sum_{k \not \in \Ilam } b_{km}  \bigg | \leq \frac{M \epsilon_n}{ \norm{\theta_m^* - \theta_{m^\prime}^*}^2 - M^\prime \epsilon_n } \\
        \implies \bigg | a_m^* - \sum_{k \in \Ilam_m} a_k \bigg | \leq \nu_n := \frac{M \epsilon_n}{ \norm{\theta_m^* - \theta_{m^\prime}^*}^2 - M^\prime \epsilon_n } + (K - M^*)M^{\prime \prime }\epsilon_n^{2 \delta} \underset{\sim}{<} \max(\epsilon_n, \epsilon_n^{2 \delta}). \label{eq:nu}
    \end{gather*}
    \end{proof} 
    
    \section{Proof of Theorem \ref{thm:consistency}}
    \begin{proof}
    Recall that $\lambda^* = \sum_{m=1}^{M^*} a_m^* \delta_{\theta_m^*}$ and that by (A1), if $\lambda \in \Lambda$ then $\lambda = \sum_{k=1}^K a_k \delta_{\theta_k}$ for some $K \leq L$. Let $B_{\epsilon_n}(\lambda^*) = \lb \lambda \in \Lambda: W_2(\lambda, \lambda^*) \lenq \epsilon_n \rb$ and $\Lambda_{M^*,L} = \lb \lambda \in \Lambda : |\tx{supp}(\lambda)| \geq M^* \rb$. We decompose the posterior expectation of $\mathcal D_{ij}$ into
    \begin{gather*}
        \Delta_{ij} = \E[\mathcal D_{ij} \mid \X] \nonumber \\
        = \int_{\Lambda_{M^*,L} \cap B_{\epsilon_n}(\lambda^*)} \mathcal D_{ij} \Pi(d \phi_{ij} \mid \X) + \int_{\Lambda_{M^*,L}^c \cap B_{\epsilon_n}(\lambda^*)} \mathcal D_{ij} \Pi(d \phi_{ij} \mid \X) + \int_{B_{\epsilon_n}^c(\lambda^*)} \mathcal D_{ij} \Pi(d \phi_{ij} \mid \X),
    \end{gather*}
    where $\phi_{ij} = (\theta_i, \theta_j,  \lambda )$. This implies
    \begin{equation*}
        \int_{\Lambda_{M^*,L} \cap B_{\epsilon_n}(\lambda^*)} \mathcal D_{ij} \Pi(d \phi_{ij} \mid \X) \leq \Delta_{ij} \leq \int_{\Lambda_{M^*,L} \cap B_{\epsilon_n}(\lambda^*)} \mathcal D_{ij} \Pi(d \phi_{ij} \mid \X) + \nu_n(\X) + \rho_n(\X),
    \end{equation*}
    where $\nu_n(\X) = \Pi( |\tx{supp}(\lambda)| < M^*, W_2(\lambda, \lambda^*) \lenq \epsilon_n \mid \X)$ and $\rho_n(\X) = \Pi(W_2(\lambda, \lambda^*) \genq \epsilon_n \mid \X)$. By assumption (A5),
    $\rho_n(\X) \overset{\P^0}{\longrightarrow} 0$. In addition, $\nu_n(\X) \overset{\P^0}{\longrightarrow} 0$, and we can show that $\P^0 \lb \nu_n(\X) = 0 \rb = 1$ for large enough $n$. To see this, note that $\lambda \in W_2(\lambda, \lambda^*)$ implies that for all $\theta_m^* \in \tx{supp}(\lambda)$, there exists a $\theta_{k} \in \tx{supp}(\lambda)$ so that $\norm{\theta_{m}^* - \theta_k} \lenq \epsilon_n$. More formally,
    \begin{gather*}
        \epsilon_n^2 \genq W_2(\lambda, \lambda^*) = \sum_{k,m} b_{km} \norm{\theta_k - \theta_m^*}^2 \geq \sum_{k} b_{km} \norm{\theta_k - \theta_m^*}^2 \\
        \geq \min_{k} \norm{\theta_k - \theta_m^*}^2 a_m^* \geq \min_{k} \norm{\theta_k - \theta_m^*}^2 a_{\min}^*,
    \end{gather*}
    where $a_{\min}^* = \min_{m} a_m^*$. If we set $n$ large enough so that $\min_{m \neq m^\prime} \norm{\theta_m^* - \theta_{m^\prime}^*} \genq 2 \epsilon_n$, then we have that there must be at least $M^*$ support points in $\lambda$. Hence, for large enough $n$, $\Lambda_{M^*,L}^c \cap B_{\epsilon_n}(\lambda^*) = \emptyset$. Therefore, we will omit $\nu_n(\X)$ from the remainder of the proof by assuming that $n$ is large enough so that $\min_{m \neq m^\prime} \norm{\theta_m^* - \theta_{m^\prime}^*} \genq 2 \epsilon_n$. We then formulate the bounds of $\Delta_{ij}$ via
    \begin{equation} \label{eq:rvns-refined-bounds}
        \int_{\Lambda_{M^*,L} \cap B_{\epsilon_n}(\lambda^*)} \mathcal D_{ij} \Pi(d \phi_{ij} \mid \X) \leq \Delta_{ij} \leq \int_{\Lambda_{M^*,L} \cap B_{\epsilon_n}(\lambda^*)} \mathcal D_{ij} \Pi(d \phi_{ij} \mid \X) + \rho_n(\X).
    \end{equation}

    Let $0 < \delta < 1$ and $\lambda \in \Lambda_{M^*,L} \cap B_{\epsilon_n}(\lambda^*)$. Proposition \ref{prop:rmengMS} states that for any $m \in [M^*]$ there exists a set of indices $\Ilam_m$ so that $\max_{k \in \Ilam_m } \norm{\theta_k - \theta_m^*} \lenq \epsilon_n$ and $\bigg | \sum_{k \in \Ilam_m} a_k - a_m^* \bigg | \lenq \max(\epsilon_n, \epsilon_n^{2 \delta})$, and in addition that $ \sum_{k \not \in \Ilam } a_k \lenq \epsilon_n^{2 \delta} $, where $\Ilam = \bigcup_{m=1}^{M^*} \Ilam_m$. We will now translate contraction of the atoms and weights into contraction of the $\Delta_{ij}$ terms. Recall that we can write
    \begin{equation} \label{eq:rvns-deltaMS}
        \int_{\Lambda_{M^*,L} \cap B_{\epsilon_n}(\lambda^*)} \mathcal D_{ij} \Pi(d \phi_{ij} \mid \X) =  \int_{\Lambda_{M^*,L} \cap B_{\epsilon_n}(\lambda^*)} \sum_{k<k^\prime} \eta_{kk^\prime}  \: q_{ij}^{kk^\prime} \Pi(d \lambda \mid \X),
    \end{equation}
    We then decompose the integrand in \eqref{eq:rvns-deltaMS}:
    \begin{gather}
        \sum_{k<k^\prime} \eta_{kk^\prime} q_{ij}^{kk^\prime} \nonumber \\
        = \sum_{k<k^\prime \in \Ilam} \eta_{kk^\prime} q_{ij}^{kk^\prime} + \sum_{(k \not \in \Ilam) \tx{ or } (k^\prime \not \in \Ilam)} \eta_{kk^\prime} q_{ij}^{kk^\prime} \nonumber \\ 
        = \sum_{m=1}^{M^*} \sum_{k<k^\prime \in \Ilam_m} \eta_{kk^\prime} q_{ij}^{kk^\prime} + \sum_{m<m^\prime}\sum_{\substack{k \in \Ilam_m \\ k^\prime \in \Ilam_{m^\prime}}} \eta_{kk^\prime} q_{ij}^{kk^\prime} + \sum_{(k \not \in \Ilam) \tx{ or } (k^\prime \not \in \Ilam)} \eta_{kk^\prime} q_{ij}^{kk^\prime}. \label{eq:rvns-integrand}
    \end{gather}
     Let $\psi_m^*(\theta) = D(\theta, \theta_m^*)$. By (A4), Lemma \ref{lemma:contrate}, and Remark \ref{remark:varphi}, for any $1 \leq m \leq M^*$,
    \begin{gather*}
        \max_{k \in \Ilam_m} D(\theta_k, \theta_m^*) \leq \kappa_{nm} = \max_{\theta : \norm{\theta - \theta_h^*} \lenq  \epsilon_n} \varphi^*_m(\theta)
    \end{gather*}
    and $\kappa_{nm} \to 0$. Let $\kappa_n := \max_{m=1, \dots, M^*} \kappa_{nm}$. Then, for any $k,k^\prime \in \mathcal{I}_h^{(\lambda)}$,
    \begin{equation} \label{eq:rvns-kappa2}
        D(\theta_k, \theta_{k^\prime}) \leq 2\kappa_n,
    \end{equation}
    and for any $k \in \Ilam_m, k^\prime \in \Ilam_{m^\prime}$,
    \begin{equation*}
        D(\theta_m^*, \theta_{m^\prime}^*) - 2 \kappa_n \leq D(\theta_m, \theta_{k^\prime}) \leq D(\theta_m^*, \theta_{m^\prime}^*) + 2 \kappa_n.
    \end{equation*}
    (A4) implies similar concentration for the densities at $X_i^0$ and $X_j^0$. For $X \sim f^0$, 
    \begin{gather*}
        |g(X;\theta_k) - g(X;\theta_m^*)| = |\Tilde{g}(X-\theta_k) - \Tilde{g}(X-\theta_m*)| \\
        \lenq \norm{X - \theta_k - \theta_m^* - X}^\zeta = M\norm{\theta_k - \theta_m^*}^{\zeta} \lenq  \epsilon_n^\zeta.
    \end{gather*}
    It follows that,
    \begin{equation*}
        \bigg | \sum_{m=1}^{M^*} \sum_{k \in \Ilam_m} a_k g(X_i^0; \theta_k) - f^*(X_i^0) \bigg | \lenq \max(\epsilon_n^\zeta, \nu_n)
    \end{equation*}
    where $\nu_n$ is defined in \eqref{eq:nu}. Condition (A3) and the dominated convergence theorem implies that for large $n$, the following holds with $\P^0$-probability tending to $1$ (see Remark \ref{remark:whp}),
    \begin{equation} \label{eq:flb}
       0 \leq f_-^{(n)}(X_i^0) \lenq f(X_i^0) \lenq f_+^{(n)}(X_i^0)  + \epsilon_n^{2 \delta},
    \end{equation}
    where 
    \begin{align*}
        f_-^{(n)}(X_i^0) = f^*(X_i^0) - M \max(\epsilon_n^\zeta, \nu_n); \\
        f_+^{(n)}(X_i^0) = f^*(X_i^0) + M \max(\epsilon_n^\zeta, \nu_n),
    \end{align*}
    for some constant $M>0$. We can now begin to bound some of the terms in (\ref{eq:rvns-integrand}). For any $k,k^\prime \not \in \Ilam$,
    \begin{gather*}
         0 \leq q_{ij}^{kk^\prime} 
         = a_k a_{k^\prime} \frac{g(X_i^0; \theta_k) g(X_j^0; \theta_{k^\prime}) + g(X_i^0; \theta_{k^\prime}) g(X_j^0; \theta_k)}{f(X_i^0) f(X_j^0)} 
          \lenq \frac{\epsilon_n^{4 \delta}}{f_-^{(n)}(X_i^0) f_-^{(n)}(X_j^0)}.
    \end{gather*}
    Similarly, for any $k \not \in \Ilam$ and $k^\prime \in \Ilam$,
     \begin{gather*}
         0 \leq q_{ij}^{kk^\prime}
          \lenq \frac{\epsilon_n^{2 \delta}}{f_-^{(n)}(X_i^0) f_-^{(n)}(X_j^0)}.
    \end{gather*}
    Next, if $k,k^\prime \in \Ilam_m$, by \eqref{eq:rvns-kappa2},
    \begin{gather*}
         0 \leq \eta_{kk^\prime} q_{kk^\prime}^{kk^\prime} \leq 2 \kappa_n.
    \end{gather*}
    Additionally,
    \begin{gather*}
        \Bigg | \sum_{\substack{k \in \Ilam_m \\ k^\prime \in \Ilam_{m^\prime}}} a_k a_{k^\prime} g(X_i^0; \theta_k) g(X_j^0; \theta_{k^\prime}) - a_m^* a_{m^\prime}^* g(X_i^0; \theta_m^*) g(X_j^0; \theta_{m^\prime}^*) \Bigg | \underset{\sim}{<} \max(\epsilon_n^\zeta, \nu_n).
    \end{gather*}

    Therefore,
    \begin{gather*}
        (\eta_{mm^\prime}^* - 2 \kappa_n) R_{ij}^{mm^\prime} \leq \sum_{\substack{k \in \Ilam_m \\ k^\prime \in \Ilam_{m^\prime}}} \eta_{kk^\prime} q_{ij}^{kk^\prime}
        \leq (\eta_{mm^\prime}^* + 2 \kappa_n)  S_{ij}^{mm^\prime}; \\ 
        R_{ij}^{mm^\prime} = \frac{a_m^*a_{m^\prime}^* \lb g(X_i^0; \theta_m^*) g(X_j^0; \theta_{m^\prime}^*) + g(X_j^0; \theta_m^*) g(X_i^0; \theta_{m^\prime}^*) \rb  - M \max(\epsilon_n^\zeta, \nu_n)}{(f_+^{(n)}(X_i^0) + M^\prime \epsilon_n^{2 \delta}) ( f_+^{(n)}(X_j^0) + M^\prime \epsilon_n^{2 \delta}) };\\
        S_{ij}^{mm^\prime}=\frac{a_m^*a_{m^\prime}^* \lb g(X_i^0; \theta_m^*) g(X_j^0; \theta_{m^\prime}^*) + g(X_j^0; \theta_{m}^*) g(X_i^0; \theta_{m^\prime}^*) \rb  + M \max(\epsilon_n^\zeta, \nu_n) }{f_-^{(n)}(X_i^0) f_-^{(n)}(X_j^0)},
    \end{gather*}
    where $M,M^\prime > 0$ are constants. Finally, this gives us the following bound on (\ref{eq:rvns-integrand}) for all $\lambda \in \Lambda_{M^*,L} \cap B_{\epsilon_n}(\lambda^*)$,
    \begin{gather*}
        \Delta_{ij}^{*-} \leq \sum_{m < l} \eta_{ml}^{(\lambda)} q_{ij}^{ml(\lambda)}  \leq \Delta_{ij}^{*+},
    \end{gather*}
    where
    \begin{gather}
        \Delta_{ij}^{*-} :=  \sum_{m<m^\prime}(\eta_{mm^\prime} - 2 \kappa_n) R_{ij}^{mm^\prime}; \label{eq:rvns-deltaminus} \\
        \Delta_{ij}^{*+} := M^*(R^*+1)R^* \kappa_n + \sum_{m<m^\prime} (\eta_{mm^\prime}^* + 2 \kappa_n) S_{ij}^{mm^\prime} + \frac{M \lb {L \choose 2} - {M^* \choose 2} \rb \epsilon_n^{2 \delta}}{f_-^{(n)}(X_i^0) f_-^{(n)}(X_j^0)}, \label{eq:rvns-deltaplus}
    \end{gather}
    where $R^* = (L-M^*)$. By working through the multiplication above, we can write $\Delta_{ij}^{*-} = \Delta_{ij}^* + \delta_{ij}^-$ and $\Delta_{ij}^{*+} = \Delta_{ij}^* + \delta_{ij}^+$, where $\delta_{ij}^-, \delta_{ij}^+ \to 0$ almost surely and depend on $\lambda^*$, $R^*$, $\delta$, and $\zeta$. 

    Returning to the bounds established in \eqref{eq:rvns-refined-bounds}, we can see that with probability tending to $1$,
    \begin{equation} \label{eq:rvns-deltastar-final-bounds}
        (\Delta_{ij}^* + \delta_{ij}^-) \lb 1 - \rho_n(\X) \rb \leq \Delta_{ij} \leq (\Delta_{ij}^* + \delta_{ij}^+) \lb 1 - \rho_n(X) \rb + \rho_n(\X).
    \end{equation}
    Define
    \begin{equation*}
        \mathcal{S}(\widehat{\bs c}, \X) = \sum_{i<j} \lb \textbf{1}_{\hat{c_i}=\hat{c}_j} \Delta_{ij}^*(1-\rho_n(\X))  + \omega \textbf{1}_{\hat{c}_i \neq \hat{c}_j}(1- \Delta_{ij}^*(1-\rho_n(\X)))\rb,
    \end{equation*}
    and note that 
    \begin{equation*}
        |\mathcal R^*(\widehat{\bs c}) - \mathcal{S}(\widehat{\bs c}, \X)| \leq {n \choose 2} \rho_n(\X) \max(\omega, 1) \bar{\Delta}^*, 
    \end{equation*}
    where $\bar{\Delta}^* = {n \choose 2}^{-1}\sum_{i<j} \Delta^*_{ij}$. Using \eqref{eq:rvns-deltastar-final-bounds}, it follows that
    \begin{gather*}
        |\mathcal{R}(\widehat{\bs c}) - \mathcal{S}(\widehat{\bs c}, \X)| \leq \max(\omega, 1) {n \choose 2} \mathcal{Q}(\X); \\
        \mathcal{Q}(\X) = {n \choose 2}^{-1} \sum_{i<j} \max \lb \delta_{ij}^+(1-\rho_n(\X)) + \rho_n(\X), - \delta_{ij}^-(1-\rho_n(\X)) \rb.
    \end{gather*}
    Therefore, we have that
    \begin{equation*}
        {n \choose 2}^{-1} | \mathcal{R}(\widehat{\bs c}) - \mathcal{R}^*(\widehat{ \bs c})| \leq \max(\omega, 1) \lb \rho_n(\X) \bar{\Delta}^*  + \mathcal{Q}(\X)  \rb. 
    \end{equation*}

    Observe that $\mathcal{Q}(\X)$, $\rho_n(\X)$, and $\bar{\Delta}^*$ do \textit{not} depend on $\chat$. Hence, we can think of the above result as a notion of uniform convergence towards the KL-oracle risk. An important consequence is convergence of the optimal values of $\mathcal{R}$ and $\mathcal R^*$. Let $\cFOLD$ be a minimizer of $\mathcal R(\chat)$ and $\cFOLD^*$ be a minimizer of $\mathcal R^*(\chat)$ over the space of partitions of $[n]$. Then observe that
    \begin{gather*}
        {n \choose 2}^{-1} \lb  \mathcal{R}(\cFOLD) - \mathcal R^*(\cFOLD^*) \rb \\
        = {n \choose 2}^{-1} \lb \mathcal{R}(\cFOLD) - \mathcal R(\cFOLD^*) + \mathcal R(\cFOLD^*) \mathcal - \mathcal  R^*(\cFOLD^*) \rb \\
        \lenq 0 + \max(\omega, 1) \lb \rho_n(\X) \bar{\Delta}^* + \mathcal{Q}(\X) \rb,
    \end{gather*}
    and, similarly,
    \begin{gather*}
        {n \choose 2}^{-1} \lb \mathcal{R}(\cFOLD) - \mathcal R^*(\cFOLD^*) \rb \\
        = {n \choose 2}^{-1} \lb \mathcal{R}(\cFOLD) - \mathcal R^*(\cFOLD) + \mathcal R^*(\cFOLD)  - \mathcal R^*(\cFOLD^*) \rb \\
        \genq - \max(\omega, 1) \lb \rho_n(\X) \bar{\Delta}^* + \mathcal{Q}(\X) \rb + 0,
    \end{gather*}
    ultimately showing that
    \begin{equation*}
        {n \choose 2}^{-1} | \mathcal R(\cFOLD) - \mathcal R^* (\cFOLD^*) | \lenq \max(\omega, 1) \lb \rho_n(\X) \bar{\Delta}^* + \mathcal{Q}(\X) \rb.
    \end{equation*}
\end{proof}

\begin{lemma} \label{lemma:contrate}
    Let $\chi \subset \mathbb{R}^p$ be compact and $f : \chi \to \mathbb{R}$, where $f$ is bounded on $\chi$ and continuous at a point $y \in \chi$. Let $\delta_n$ be a sequence with $\lim_{n \to \infty} \delta_n = 0$ and $B_{\delta_n}(y) = \lb x \in \chi : \norm{x - y} \leq \delta_n \rb$. Then, there exists a sequence $\epsilon_n$ that only depends only $y$ and with $\lim_{n \to \infty} \epsilon_n = 0$ so that for all $x \in B_{\delta_n}(y)$, 
    \begin{equation*}
        |f(x) - f(y)| \leq \epsilon_n.
    \end{equation*}
\end{lemma}

\begin{proof}
  Consider sequences of within-neighborhood maximizers and minimizers of $f$,
    \begin{gather*}
        x_n = \underset{x \in B_{\delta_n}(y)}{\tx{argmax}} f(x); \qquad
        z_n = \underset{x \in B_{\delta_n}(y)}{\tx{argmin}} f(x).
    \end{gather*}
    In the case where there are multiple minimizers or maximizers, simply choose one of these points as $x_n$ or $z_n$. Since $f$ is continuous and bounded and $B_{\delta_n}(y)$ is a compact set, $x_n,z_n \in B_{\delta_n}(y)$. This means that $x_n \to y$ and $z_n \to y$, which implies $f(x_n) \to f(y)$ and $f(z_n) \to f(y)$ by continuity. Therefore, for any $x \in B_{\delta_n}(y)$, 
    \begin{gather*}
        |f(x) - f(y)| \leq \epsilon_n := |f(x_n) - f(z_n)| \to 0.
    \end{gather*}
\end{proof}

\begin{remark} \label{remark:varphi}
    Observe that if $y = \tx{argmin}_{x \in \chi} f(x)$, the proof of Lemma \ref{lemma:contrate} implies that we can instead set $\epsilon_n = |f(x_n) - f(y)|$. 
\end{remark}

\subsection{Further Remarks on Theorem \ref{thm:consistency}}

\begin{remark} \label{remark:whp}
    Note that we state Theorem \ref{thm:consistency} as a high probability statement, not an almost everywhere statement. Pragmatically, this is because of (\ref{eq:flb}), as we will need to show
    \begin{equation} \label{eq:fstarlb}
        \P^0 \lb  f^*(X_i^0) \genq \max(\epsilon_n, \nu_n) \: \: \forall i \in [n] \rb \to 1
    \end{equation}
    in order for us to construct $\Delta_{ij}^{*-}$ and $\Delta_{ij}^{*+}$ without dealing with negative numbers that could change the direction of the inequalities. It is simple to verify that the probability of (\ref{eq:fstarlb}) goes to $1$ for any probability density functions $f^*$ and $f^0$ by the dominated convergence theorem. Alternatively, one can make the assumption that there exists an $\epsilon^0 > 0$ so that $\P^0(f^*(X) > \epsilon^0) = 1$, but this would not apply for well-specified regimes. 
\end{remark}

\begin{remark}
    A special case for the well-specified kernel regime is the location mixture of Gaussians. Here, we have that for any constant $M>0$,
    \begin{gather*}
        \mathbb{P}^0 \lb \N_d(X; \theta_h^0; \Sigma) \geq M \max(\epsilon_n, \nu_n) \mid s_0 = h \rb \\
        = \mathbb{P}^0\left[\norm{\Sigma^{-1/2}(X-\theta_h^0)}^2 \leq -2 \log \lb (2 \pi)^{d/2} \tx{det}(\Sigma)^{1/2} M \max(\epsilon_n, \nu_n) \rb \mid s_0 = h \right] \\
        \geq 1 + \frac{d}{2 \log \lb (2 \pi)^{d/2} \tx{det}(\Sigma)^{1/2} M \max(\epsilon_n, \nu_n) \rb}\label{eq:lgaussconvrate}
    \end{gather*}
    by Markov's inequality as $\norm{\Sigma^{-1/2}(X-\theta_h^0)}^2 \sim \chi^2_d$. If we set $a_{\tx{min}}^0 = \min_{h} a_h^0$, we then have that for any $M>0$,
    \begin{gather*}
        \mathbb{P}^0 \lb f^0(X) \geq M \max(\epsilon_n, \nu_n) \rb = \sum_{h=1}^{M^0} a_h^0 \mathbb{P}^0 \lb f^0(X) \geq M \max(\epsilon_n, \nu_n) \mid s_0 = h \rb \\
        \geq \sum_{h=1}^{M^0} a_h^0 \mathbb{P}^0 \lb \N_d(X; \theta_h^0; \Sigma) \geq \frac{M}{a_h^0} \max(\epsilon_n, \nu_n) \mid s_0 = h\rb \\
        \geq \sum_{h=1}^{M^0} a_h^0 \mathbb{P}^0 \lb \N_d(X; \theta_h^0; \Sigma) \geq \frac{M}{a_{\tx{min}}^0} \max(\epsilon_n, \nu_n) \mid s_0 = h \rb \\
        = \mathbb{P}^0 \lb \N_d(X; \theta_h^0; \Sigma) \geq \frac{M}{a_{\tx{min}}^0} \max(\epsilon_n, \nu_n) \mid s_0 = h \rb,
    \end{gather*}
    because (\ref{eq:lgaussconvrate}) does not depend on $h$. This tells us that (\ref{eq:fstarlb}) goes to $1$ at an approximately $\log \left( \log n /n \right)$ rate.
\end{remark}

\section{Comparison to Existing Methods}
The general idea of merging mixture components to improve clustering results has been implemented in a variety of other methods. For example, \cite{chan2008statistical} merged mixture components to infer asymmetric cell subsets in an application to flow cytometry using numerical optimization of the probability density function. Also motivated by flow cytometry, \cite{baudry2010combining} proposed first selecting the number of clusters using BIC, then successively merging components by minimizing the entropy of the clustering. \cite{hennig2010methods} showcases multiple methods for merging mixture components, including methods that preserve unimodality in the clustering, a method that successively merges clusters based on the Bhattacharyya distance, and methods based off diagnostic graphs. An algorithm that combines components sharing significant overlap was proposed in \cite{melnykov2016merging}, where they also explored the misspecified kernel regime. More recently, the problem has attracted theoretical attention, with \cite{aragam2020identifiability} showing explicitly that merging mixture components in a GMM via single linkage clustering can identify a variety of clustering structures. They provide a Bayes optimal clustering algorithm known as NPMIX that implements their approach. 

Similar to our approach, these methods are post-processing procedures motivated by improving the results of the GMM, particularly in the regime of kernel misspecification. However, there are multiple key differences between FOLD and other density merging methods. First, FOLD implements component merging from a Bayesian perspective, whereas existing methods tend to expand on maximum likelihood estimation. Next, with the exception of the NPMIX algorithm, most of the existing density merging algorithms make locally-optimal moves in order to improve the result of the EM algorithm, then terminate after a stopping criterion. NPMIX terminates when the desired number of clusters is reached, which is assumed to be either known or estimated in advance. In contrast, the minimization of the FOLD loss is based off a decision theoretic formulation that concretely quantifies the loss of co-clustering two objects. Furthermore, using the credible ball and posterior similarity matrix, we can coherently express uncertainty about the clustering estimate itself, whereas NPMIX and the other procedures characterize uncertainty in individual cluster allocation via misclassification probabilities.  

NPMIX and the other existing algorithms often rely on a two step procedure: first, they estimate and merge the mixture components, and then they assign observations to the merged components (typically using MAP estimation), ultimately creating clusters. In FOLD, we completely bypass the first step by associating each object with a component via their localized density $g(\cdot; \theta_i)$, then group the localized densities in a one-step procedure (i.e, when we minimize the risk). The main idea here is to use $\bs \theta = (\theta_1, \dots, \theta_n)$ as the parameter of interest for merging the components, rather than the unique values $\Tilde{\bs \theta} = (\tilde \theta_1, \dots, \tilde \theta_K)$. While these paradigms may seem equivalent at first since $\theta_i = \tilde \theta_k$ for some $k$ almost surely, focusing on $\bs \theta$ results in formulating the merging problem in terms of partitioning $[n]$ rather than partitioning $[K]$. In addition, the two step procedure greatly reduces the partition space. If, say, a merging procedure takes $L$ components and combines these into $M < L$ merged components, then the resulting clustering must have at most $M$ clusters. Conversely, $\mathcal R(\chat)$ is minimized over the space of \textit{all} partitions of $[n]$. Finally, a drawback of using a two-step procedure is sensitivity to the point estimates for $\bs a$ and $\Tilde{\bs \theta}$ when constructing the clustering. FOLD, on the other hand, does not need to estimate the components in order to derive clusters.

Recently, the Variation of Information (VI) \citep{meilua2007comparing} has attracted attention in the literature for use as a clustering loss function for Bayesian mixture models \citep{wade2018bayesian, de2023bayesian, denti2023common, wade2023bayesian}. The VI loss between $\chat$ and $\bs s$ is defined as
\begin{equation*}
    \mathcal L_\VI(\chat, \bs s) = - \sum_{h=1}^{\hat{K}} \frac{n_{h \cdot}}{n} \log_2 \left( \frac{n_{h \cdot}}{n} \right) - \omega \sum_{k=1}^{K} \frac{n_{\cdot k}}{n} \log_2 \left( \frac{n_{\cdot k}}{n} \right) - (1+\omega) \sum_{h=1}^{\hat{K}} \sum_{k=1}^K \frac{n_{hk}}{n} \log_2 \left( \frac{n_{hk} n}{n_{h \cdot} n_{\cdot k}} \right),
\end{equation*}
where $n_{h \cdot} = |\hat{C}_h|$, $n_{\cdot k} = |S_k|$, and $n_{hk} = |\hat{C}_h \cap S_k|$. The first term gives the entropy of a random assignment of an observation to a cluster, where the probability of an assignment is given by $n_{h \cdot}/n$. The second term is analogous but for the task of randomly assigning an object to a component with weights $n_{\cdot k}/n$. The third term is the mutual information of the two random allocation tasks, in which the joint probabilities of assigning the object to cluster $h$ and component $k$ are given by $n_{hk}/n$. The VI shares multiple key properties with Binder's loss, including that it is a metric when $\omega = 1$ and a quasimetric otherwise \citep{meilua2007comparing, wade2018bayesian, dahl2022search}. 

Unlike FOLD, the VI was not motivated as a loss function in Bayes estimation, but was instead originally proposed by \cite{meilua2007comparing} as a metric over the partition space. The VI only takes into account information from the component labels rather than the other parameters in the mixture model, as in Binder's loss. Hence, in the context of kernel misspecification, minimizing the VI may not be reasonable due to the inherent brittleness of $\bs s$. In addition, the VI is not a pair-counting loss, i.e. it cannot be expressed as a sum over all object pairs $i<j$, and therefore does not admit an interpretation in terms of quantifying the loss of co-clustering two objects. This also makes the role of $\omega$ in $\mathcal L_\VI(\chat, \bs s)$ less interpretable than for both Binder's loss and FOLD, despite its influence on the resulting point estimate $\bs c_\VI = \underset{\widehat{\bs c}}{\tx{argmin  }} \E_\Pi \lb \mathcal L_\VI(\chat, \bs s) \mid \X \rb$ in practice \citep{dahl2022search}. Finally, while the VI has been observed to yield fewer clusters than Binder's loss empirically when $\omega=1$ on synthetic and real data \citep{wade2018bayesian}, we show in a simulation study in the Section \ref{section:sims-supplementary} and an application in Section \ref{section:apps} of the main article that $\bs c_\VI$ is still vulnerable to the over-clustering problem, albeit usually to a lesser degree than $\bs c_{\tx{B}}$. 

\section{Additional Details on Credible Balls}
We express uncertainty in $\cFOLD$ using a $95\%$ credible ball. For visualizations, we rely on the notions of horizontal and vertical bounds, as introduced in \cite{wade2018bayesian}. For completeness, we include the formal definitions of these bounds below. The horizontal bounds are given by the clusterings in $B(\cFOLD)$ furthest from $\cFOLD$ by $\mathbb D(\cdot, \cdot)$.

\begin{definition}
    The horizontal bound(s) of $B(\cFOLD)$ are
    \begin{gather*}
        \tx{H}(\cFOLD) = \{ \bs c \in B(\cFOLD):
        \mathbb{D}(\bs c, \cFOLD) \geq \mathbb{D}(\bs c^\prime, \cFOLD) \: \forall \bs c^\prime \in B(\cFOLD)\}.
    \end{gather*}
\end{definition}

For any clustering $\bs c$, let $K_{\bs c}$ be the number of clusters. The vertical bounds also consider clusterings in $B(\cFOLD)$ that are far from $\cFOLD$, but impose additional constraints on the number of clusters.
\begin{definition}
    The vertical upper bound(s) of $B(\cFOLD)$ are
    \begin{gather*}
        \tx{VU}(\cFOLD) = \{ \bs c \in B(\cFOLD): K_{\bs c} \leq K_{\bs c^\prime} \: \forall \bs c^\prime \in B(\cFOLD),  \\
        \mathbb{D}(\bs c, \cFOLD) \geq \mathbb{D}(\bs c^\prime, \cFOLD) \: \forall \bs c^\prime \in B(\cFOLD) \tx{ with } K_{\bs c } = K_{\bs c^\prime} \}.
    \end{gather*}
\end{definition}

\begin{definition}
    The vertical lower bound(s) of $B(\cFOLD)$ are
    \begin{gather*}
        \tx{VL}(\cFOLD) = \{ \bs c \in B(\cFOLD): K_{\bs c} \geq K_{\bs c^\prime} \: \forall \bs c^\prime \in B(\cFOLD),  \\
        \mathbb{D}(\bs c, \cFOLD) \geq \mathbb{D}(\bs c^\prime, \cFOLD) \: \forall \bs c^\prime \in B(\cFOLD) \tx{ with } K_{\bs c } = K_{\bs c^\prime} \}.
    \end{gather*}
\end{definition}

In practice, we compute the credible balls using the \texttt{R} function \texttt{credibleball()} in the \texttt{mcclust.ext} package \citep{wade2015package}. $\mathbb D(\cdot, \cdot)$ can be either the Variation of Information distance or Binder's loss. The bounds may not be unique in general. 

\section{Extended Implementation Details}

\subsection{GSE81861 Cell Line Dataset Application} \label{subsect:cells-hyperpars}
We use the following priors for the Bayesian GMM: 
\begin{align*}
    (\tilde \mu_k, \tilde \Sigma_k) & \sim \mathcal{NIW}(0_p, 1, p+2, I); & \bs a & \sim \tx{Dir}( (1/2)  1_K).
\end{align*}
Recall that along with model-based clustering methods, we implement average linkage hierarchical clustering (HCA), k-means, DBSCAN, and spectral clustering. All methods were applied to the projection of the original cell line data onto the first $p=5$ principal components. For HCA, we set the number of clusters to be equal to $7$, corresponding to the correct number of types. In k-means, we select the number of clusters with an elbow plot. We use the R package \texttt{dbscan} \citep{hahsler2019dbscan} to implement the DBSCAN algorithm. Following the documentation, we set $\texttt{minPts} = p + 1 = 6$, then use an elbow plot to choose the \texttt{eps} parameter, which comes out to $\texttt{eps}=1.0$. For spectral clustering, we first compute a Gaussian kernel affinity matrix $A = (A_{ij})_{i \in [n], j \in [n]}$, where $A_{ij} = \exp \lb - \norm{X_i^0-X_j^0}^2 \rb.$ We compute the Laplacian matrix and plot the eigenvalues in ascending order to determine the number of clusters $\hat K_n$, then create a matrix $V$ consisting of the first $\hat K_n$ eigenvectors. Finally, we apply k-means to $V$ with $k=\hat{K}_n$. UMAP plots of the clusterings given by the algorithmic methods are displayed in Figure \ref{fig:alogrithmic-methods}. 

\begin{figure}
    \centering
    \includegraphics[scale=0.5]{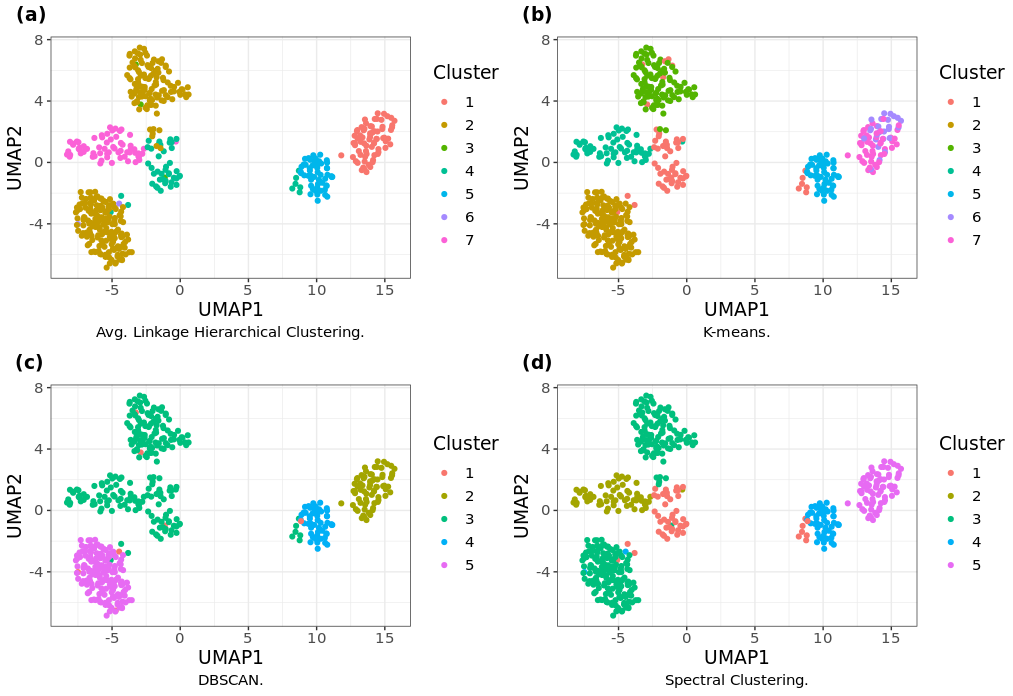}
    \caption{UMAP plots for the algorithmic clustering methods fit to the cell line dataset.}
    \label{fig:alogrithmic-methods}
\end{figure}

\subsection{\texttt{moons} Example}
The \texttt{moons} dataset is simulated using the \texttt{RSSL} package \citep{krijthe2015implicit, krijthe2016rssl} with $n=500$. We fit a Bayesian location GMM with $K=30$ components, i.e., $X_i \sim \N(\theta_i, 0.02I)$, $\tilde \theta_k \sim \N(0,2I)$, $\Pi(s_i = k \mid a_k) = a_k$, and $\bs a \sim \tx{Dir}(1/K, \dots, 1/K)$. We run a Gibbs sampler for $T=35,000$ iterations, remove the first $10,000$ as burn-in, and take every fourth iteration to compute $\Delta$. 

\subsection{Application to Benchmark Clustering Datasets}

For all datasets, the Gibbs sampler is run for $T=50,000$ iterations, the first $1,000$ are treated as burn-in, and we take every fourth iteration to compute $\Delta$. The FOLD loss parameter $\omega$ is chosen with an elbow plot, and the other clustering methods are tuned using the same workflow as in the cell line dataset application, with the exception that we fix \texttt{minPts}=5 when running DBSCAN.

\section{Uncertainty Quantification Simulation Study}
\begin{figure}
    \centering
    \includegraphics[scale=0.5]{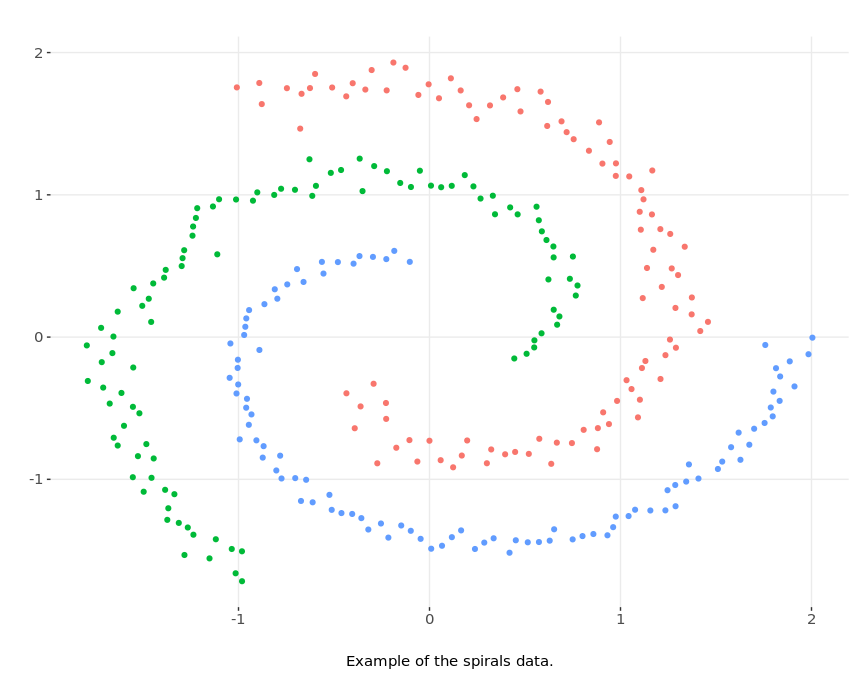}
    \caption{A sample of $n=300$ observations with the \texttt{spirals()} function, with colors corresponding to $\bs s^0$. The thin, interlocking nature of the spirals make this dataset a difficult problem for clustering algorithms.}
    \label{fig:spirals}
\end{figure}
In this section, we focus on the $95\%$ credible ball around $\cFOLD$. We simulate from a notably difficult scenario for clustering, the \texttt{spirals} dataset, as generated from the \texttt{KODAMA} package \citep{cacciatore2022kodama}. Figure \ref{fig:spirals} shows a scatterplot of one example of the \texttt{spirals} data. The data are partitioned into $M^0=3$ groups, each represented as interlocked spirals. Clustering these data is very difficult because the true clusters have long, thin, non-elliptical shapes, and are quite close to each other in the sample space. Our main interest here is whether $B(\cFOLD)$ includes the true grouping of the data, analogous to coverage of a true real parameter by a credible interval.

For each replication, we sample $n=300$ observations using the \texttt{spirals()} function, where each spiral comprises $100$ observations. We fit a Bayesian location-scale Gaussian mixture, with $K=30$, $\mu = 0_p$, $\kappa = 0.5/2$, $\nu = p+2$, and $\Psi = 0.5 I$. The Dirichlet concentration parameter for $\bs a$ is $\alpha = 1/2$ and we use $\omega = \omega^{\tx{AVG}}$ to compute $\cFOLD$ and generate samples from $\Pi(\bs c_{\bs \theta} \mid \X)$. In total, we simulate $R=100$ replications. In each replication, we run a Gibbs sampler for $15,000$ MCMC iterations, discard $1,000$ burn-in iterations, and keep every fourth MCMC draw. We compute the credible ball by setting $\mathbb{D}(\cdot, \cdot)$ to be the VI.

For these simulations, our main area of interest is whether $\bs s^0 \in B(\cFOLD)$. To do this, we simply evaluate if the horizontal bounds cover $\bs s^0$, that is, if $\tx{VI}(\bs s^0, \cFOLD) \leq \tx{VI}(\bs c_H, \cFOLD)$, where $\bs c_H$ is any clustering in $\tx{H}(\cFOLD)$. We also saved the number of clusters in $\cFOLD$ and the adjusted Rand index between $\cFOLD$ and $\bs s^0$. 

The results are given in Table \ref{table:credball}. In 89 of the replications, $\bs s^0$ is contained in $B(\cFOLD)$, despite $\cFOLD$ consistently underscoring in the adjusted Rand index. This latter phenomena comes about because $E_\Pi(\tilde \Sigma_k) = 0.5 I$, which means that the model is likely to fit Gaussian kernels that span multiple spirals, rather than approximating each spiral with multiple Gaussian kernels. Interestingly, the average number of clusters achieved by $\cFOLD$ is $3.120$, which is very close to the truth. We can then conclude that, even when the fitted model makes accurately clustering the data difficult, $\Pi(\bs c_{\bs \theta} \mid \X)$ can still achieve high coverage rates of the truth. However, note that FOLD could be applied with recent approaches for density estimation of manifold data (e.g. the GMM in \cite{berenfeld2022estimating}), which would likely improve results for the \texttt{spirals} data.

\begin{table}[]
\centering
\begin{tabular}{lll}
\toprule
Avg. $\tx{VI}(\bs s^0, \cFOLD) \leq \tx{VI}(\bs c_H, \cFOLD)$ & Avg. No. of Clusters (SD) & Avg. ARI (SD) \\
\midrule
0.890  & 3.120 (0.327)  & 0.223 (0.046)  \\
\bottomrule
\end{tabular}
\caption{Results from the credible ball simulation study. $\bs s^0$ is covered by the horizontal bounds of $B(\cFOLD)$ in 89\% of replications, despite the point estimator consistently scoring low adjusted Rand index with $\bs s^0$.}
\label{table:credball}
\end{table}

\end{document}